\documentclass[12pt,a4paper]{elsarticle}
\usepackage[nottoc]{tocbibind}
\usepackage{graphicx}
\usepackage{slashed}
\usepackage{enumerate}
\usepackage{amsmath}
\usepackage{mathtools,cancel}
\usepackage[bottom]{footmisc}
\usepackage[scr=rsfso]{mathalfa}
\usepackage[hidelinks]{hyperref}
{

}
\usepackage{amsfonts}
\usepackage{psfrag}
\usepackage{color}
\usepackage{mathtools}
\usepackage{pgf,tikz}
\usepackage{mathrsfs}
\usetikzlibrary{arrows}
\usepackage{fancyhdr}
\usepackage{amssymb}
\usepackage{slashed}
\usepackage{enumitem}
\usepackage{pst-node}
\usepackage{tikz-cd} 
\usepackage[T1]{fontenc}
\usepackage[utf8]{inputenc}

\newtheorem{definition}{Definition}
\newtheorem{remark}{Remark}

\newtheorem{proposition}{Proposition}[section]

\newtheorem{theorem}{Theorem}[section]
\newtheorem{corollary}{Corollary}[section]
\newtheorem{lemma}{Lemma}[section]

\numberwithin{equation}{section}
\allowdisplaybreaks[4]

\newenvironment{proof}{\smallskip\noindent\emph{Proof.}\hspace{1pt}}%
{\hspace{-5pt}{\nobreak\quad\nobreak\hfill\nobreak$\square$\vspace{8pt}%
		\par}\smallskip\goodbreak}

\newcommand{\hnabla}{\widehat{\nabla}}

\newcommand{\be}{\begin{equation}}
\newcommand{\ee}{\end{equation}}

\newcommand{\bm}{\begin{align*}}
\newcommand{\enm}{\end{align*}}

\newcommand{\bespeq}{\begin{equation}\begin{split}}
\newcommand{\espeq}{\end{split}\end{equation}}

\newcommand{\tr}{\mbox{tr}}

\newcommand\restri[2]{{
		\left.\kern-\nulldelimiterspace 
		#1 
		\right|_{#2} 
}}

\definecolor{ffqqqq}{rgb}{1.,0.,0.}
\definecolor{uuuuuu}{rgb}{0.26666666666666666,0.26666666666666666,0.26666666666666666}

\makeatletter
\def\ps@pprintTitle{%
  \let\@oddhead\@empty
  \let\@evenhead\@empty
  \let\@oddfoot\@empty
  \let\@evenfoot\@oddfoot
}
\def\@author#1{\g@addto@macro\elsauthors{\normalsize%
    \def\baselinestretch{1}%
    \upshape\authorsep#1\unskip\textsuperscript{%
      \ifx\@fnmark\@empty\else\unskip\sep\@fnmark\let\sep=,\fi
      \ifx\@corref\@empty\else\unskip\sep\@corref\let\sep=,\fi
      }%
    \def\authorsep{\unskip,\space}%
    \global\let\@fnmark\@empty
    \global\let\@corref\@empty  
    \global\let\sep\@empty}%
    \@eadauthor={#1}
}
\makeatother

\begin{document}
\begin{frontmatter}
	
\title{A Geometric Approach to the Yang-Mills Mass Gap}
\author{Puskar Mondal\fnref{fn1,fn2}}
\fntext[fn1]{Centre of Mathematical Sciences and Applications, Harvard University}
\fntext[fn2]{Department of Mathematics, Harvard University}
\ead{puskar_mondal@fas.harvard.edu}

\begin{abstract}
\vspace{5pt}
\par \noindent I provide a new idea based on geometric analysis to obtain a positive mass gap in pure non-abelian renormalizable Yang-Mills theory. The orbit space, that is the space of connections of Yang-Mills theory modulo gauge transformations, is equipped with a Riemannian metric that naturally arises from the kinetic part of reduced classical action and admits a positive definite sectional curvature. The corresponding regularized \textit{Bakry-\'Emery} Ricci curvature (if positive) is shown to produce a mass gap for $2+1$ and $3+1$ dimensional Yang-Mills theory assuming the existence of a quantized Yang-Mills theory on $(\mathbb{R}^{1+2},\eta)$ and $(\mathbb{R}^{1+3},\eta)$, respectively. My result on the gap calculation, described at least as a heuristic one, applies to non-abelian Yang-Mills theory with any compact semi-simple Lie group in the aforementioned dimensions. In $2+1$ dimensions, the square of the Yang-Mils coupling constant $g^{2}_{YM}$ has the dimension of mass, and therefore the spectral gap of the Hamiltonian is essentially proportional to $g^{2}_{YM}$ with proportionality constant being purely numerical as expected. Due to the dimensional restriction on $3+1$ dimensional Yang-Mills theory, it seems one ought to introduce a length scale to obtain an energy scale. It turns out that a certain `trace' operation on the infinite-dimensional geometry naturally introduces a length scale that has to be fixed by measuring the energy of the lowest glu-ball state. However, this remains to be understood in a rigorous way.  
\end{abstract}
\end{frontmatter}

\begin{abstract}
 
\end{abstract}

\medskip



\section{Introduction}
\noindent One of the most important questions in contemporary mathematical physics is to prove that the Hamiltonian or Schr\"odinger operator of non-abelian Yang-Mills fields admits a spectral gap. The importance of Yang-Mills theory needs no explanation since the standard model is built upon it. Although it is commonly believed that the non-abelian gauge theories in any dimensions less than or equal to four are confining and possess a mass gap, proof on the theoretical ground with complete generality is missing. The prime interest is the $3+1$ dimensional Yang-Mills theory for obvious reasons, however, I will consider both the $2+1$ and $3+1$ dimensional cases. While the most physically relevant model is the $3+1$ dimensional QCD, pure Yang-Mills theory in $2+1$ dimensions deserves attention since it is intermediate in complexity between the $3+1$ dimensional and nearly trivial $1+1$ dimensional theories yet possesses most of the characteristic features of $3+1$ dimensional theory. There are enough theoretical and numerical evidence that both $3+1$ and $2+1$ theories confine \cite{karabali1996gauge,karabali1996gauge,karabali1997gauge,karabali1998planar,karabali1998vacuum, frasca1,frasca2, athenodorou2020glueball,lucini2010glueball}. There are of course several fundamental differences between these two theories as well. In $2+1$ dimensions, the square of the bare coupling constant has a dimension of mass leading to super-renormalizability while $3+1$ dimensional theory is borderline renormalizable. As we shall see, this will turn out to be a crucial feature of the result that we derive in section \ref{explicit}.

\noindent At the level of perturbative quantum field theory, two fundamental breakthroughs in the context of Yang-Mills theory (in dimensions $n+1,n\leq 3$) are its renormalizability \cite{veltman1972regularization} and asymptotic freedom \cite{gross1973ultraviolet}. While the former can be interpreted in terms of suitable Sobolev embedding theorems (in the case of $3+1$ dimensions, such embedding turns out to be borderline as it fails to be compact, and for this reason, it is considered to be borderline renormalizable), the latter indicates approaching a free theory at a high energy limit. 
At low energies where the Yang-Mills coupling is strong, the non-linearities are not \textit{small} (in a suitable function space setting) and in such a regime of large data problems, a range of complicated processes are expected to occur that should fundamentally separate the behavior of non-abelian gauge theories from that of abelian theories such as pure QED. One such attribute of Yang-Mills theory associated with strong field processes is the expected existence of a gap in the spectrum of the Hamiltonian \cite{jaffe2006quantum}. Such a gap, if it exists, could represent the energy difference between the actual vacuum state and that of the lowest energy ‘glueball’ states and
confirm the expectation that massless gluons cannot propagate freely as
photons do. Keeping aside the perturbative treatment, little is known about the rigorous non-perturbative quantization of almost any interacting quantum field theory in $3+1$ dimensions. In $2+1$ and $1+1$ dimensions, the construction of quantum field theories with nonlinear interactions was made possible by the breakthrough work of Jaffe and Glimm \cite{glimm1985collected,jaffe2000constructive} among others. In $4$ dimensions, through a re-normalization group argument, \cite{aizenman2021marginal} proved the Gaussianity hence triviality of $\varphi^{4}$ theory.

The classical Yang-Mills theory on $\mathbb{R}^{1+n}$ is described by the extremum of the action functional $S_{YM}:=-\frac{1}{4}\int_{\mathbb{R}^{1+n}}\langle F, F\rangle$, where $F$ is the curvature associated with a principal bundle $(\mathfrak{P}, G,\mathbb{R}^{1+n})$ written in terms of the gauge covariant exterior derivative of a connection. The resulting Yang-Mills equations can be cast into a hyperbolic system (or a coupled elliptic-hyperbolic one) in a suitable choice of gauge and as such a solution can be thought of as a curve $t\mapsto (A(t),\mathcal{E}(t))$ in the reduced phase space $T^{*}(\mathcal{A}/\widehat{\mathcal{G}})$ ($\mathcal{A}$ is the space of \textit{spatial} connections belonging to an appropriate function space and $\widehat{\mathcal{G}}$ is the group of automorphisms of the bundle $\mathfrak{P}$ after modding out the set of covariantly constant elements; note that $\mathcal{E}$ is the momentum variable associated with the connection $A\in \mathcal{A}/\widehat{\mathcal{G}}$). Since the reduced orbit space $\mathcal{A}/\widehat{\mathcal{G}}$ is an infinite dimensional manifold, one could interpret a classical solution as a particle moving in this infinite-dimensional configuration space with prescribed initial \textit{position} $A(t=0)$ and momentum $\mathcal{E}(t=0)$. With this interpretation, a naive thought of writing down the Hamiltonian operator (a formal covariant Laplace-Beltrami operator defined on the configuration space $\mathcal{A}/\widehat{\mathcal{G}}$ together with a potential term) of the system and obtaining its spectrum becomes natural. Firstly, however, such a covariant Laplace-Beltrami operator generates infinities while acting even on smooth functionals and therefore a suitable regularization is necessary to make sense of this operator. Even after one makes sense of this operator, a canonical quantization proves to be monumentally difficult. Nevertheless, there has been some progress using the Microlocal technique developed by \cite{moncrief2,marini}. On the other hand, due to the equivalence between the Schroendinger and path integral quantization, one may invoke the stochastic quantization scheme of Parisi and Wu \cite{parisi}. In this later scheme, an Euclidean quantum field theory is obtainable as a stationary limit of the Langevin dynamics associated with the classical action. \cite{hairer1,hairer2} is able to employ Parisi-Wu stochastic quantization to 3-dimensional Yang-Mills-Higgs theory while the 4-dimensional case remains a daunting task. I wish to point out that the $2$ dimensional Euclidean quantum Yang-Mills theory is made completely rigorous by several researchers \cite{klimek1987construction, gross1989two,witten1991quantum}

Leaving aside the question of a rigorous quantization, if one simply assumes such to be true then the following question arises: what is the source of mass gap? Unlike the Abelian gauge theory, the orbit space for the non-abelian one is geometrically rich. I.M Singer \cite{singer1981geometry} computed several geometric entities including the Riemann curvature of the orbit space. The sectional curvature is positive definite and the Ricci curvature is formally positive definite. Upon such observation, one immediately attempts to use the theorem on the eigenvalue estimate of the Laplace-Beltrami operator due to Lichnerowicz \cite{lichnerowicz1970varietes,lichnerowicz1971varietes}. However, the estimation of the spectral gap in an infinite dimensional setting is technically challenging due to the lack of compactness. For example, on a finite-dimensional compact manifold, one may utilize a direct Lichnerowicz \cite{lichnerowicz1970varietes,lichnerowicz1971varietes} estimate to obtain the spectral gap of Laplacian (see \cite{singer1985estimate, li1980estimates, lieb1976bounds} for estimates on the spectrum of a Schrodinger operator). In fact, in finite dimensions, Bonet-Meyers theorem \cite{myers} guarantees that manifolds with positive definite Ricci curvature (uniform lower bound) are compact. Therefore, the positive definiteness of the Ricci curvature is sufficient to obtain a spectral gap in finite dimensions. However, in infinite dimensions, this luxury is lost since even a bounded ball in an infinite dimensional manifold is not compact in general. To get around this problem, let us recall that I assume the existence of a quantum Yang-Mills theory. This, by the basic axioms of a quantum field theory (see \cite{de2010mathematical} for example) (for gauge theory there may be additional axioms), produces a normalizable ground state wave functional $\Psi[A]:=N_{\hbar}e^{-S[A]/\hbar},~N_{\hbar}\in\mathbb{C}-{0}$. The normalizability condition $\int_{\mathcal{A}/\widehat{\mathcal{G}}}|N_{\hbar}|^{2}e^{-2S[A]/\hbar}\mu_{\mathfrak{G}}=1$,  (where $\mathfrak{G}$ is a Riemannian metric on the space $\mathcal{A}/\widehat{\mathcal{G}}$ induced by the kinetic part of the classical action and $\mu_{\mathfrak{G}}$ is the associated `infinite' volume element) automatically equips the orbit space $\mathcal{A}/\widehat{\mathcal{G}}$ with a measure $|N_{\hbar}|^{2}e^{-2S[A]/\hbar}\mu_{\mathfrak{G}}$ that can be utilized to estimate the spectral gap. But, this geometric contribution that arises from the curvature is purely kinetic. Therefore one obvious issue that arises is the role played by the potential energy.  Note that the Yang-Mills potential energy $\frac{1}{4}\int_{\mathbb{R}^{n}}\mathcal{F}_{ij}\cdot \mathcal{F}^{ij}$ contain terms that are quartic in the connection since $\mathcal{F}=dA+[A,A]$. Such a potential might rise rapidly enough to confine the wave functional $e^{-S[A]/\hbar}$. In other words, the influence of the potential is felt at the level of the functional $S[A]:=-\frac{\hbar}{2}\ln(|N_{\hbar}|^{-2}|\Psi[A]|^{2})$ that aids to normalize $e^{-S[A]/\hbar}$. At the level of the mass gap, this effect of potential is precisely felt through a term the Hessian of the functional $S[A]$ that is added to the Ricci curvature of the orbit space. The following theorem is the main result of this article regarding the estimation of the mass gap of the Yang-Mills theory under the assumption that a quantum theory exists. For this, we introduce a few notations which will be described in detail later in the section \ref{yangmills}. Let $(\mathbb{R}^{1+n},\eta)$ be the $n+1, ~n=2,3$ dimensional Minkowski space with its metric $\eta$ in usual rectangular coordinates $(t,x^{i})_{i=1}^{n}$. We consider a principle $G$ bundle $\mathfrak{P}$ over $(\mathbb{R}^{1+n},\eta)$ with the structure group being a compact semi-simple Lie group $G$. The Lie algebra $\mathfrak{g}$ associated with $G$ has an adjoint invariant positive definite inner product that we denote here by `$\cdot$'. The connection of this bundle is denoted by $A:=A^{\alpha}_{\mu}T_{\alpha}dx^{\mu}$, $T_{\alpha}$ denotes a basis of $\mathfrak{g}$. We sometimes ignore the lie algebra indices and write $A_{\mu}$ for the connection (whenever it is done, it should be understood that a Lie-algebra index is present). Since we will primarily utilize a Hamiltonian formalism, it is convenient to work with spatial connection i.e., components of $A$ parallel to a $t=\text{constant}$ hypersurface $\mathbb{R}^{n}$. Let $\mathcal{S}$ be the space of Schwartz connections on $\mathbb{R}^{n}$ modulo gauge transformations (i.e., the spatial connections verify $(1+|x|)^{k}|\partial^{l}A|\leq \text{constant}~\forall k,l\in \mathbb{Z}$). In other words, by $\mathcal{S}$, we denote the connections on the orbit space $\mathcal{A}/\widehat{\mathcal{G}}$ that decay rapidly towards the infinity of $\mathbb{R}^{n}$. Under the assumption of the existence of quantum Yang-Mills theory, the Hilbert space of the theory can be identified with $L^{2}(\mathcal{S}(\mathbb{R}^{n}),e^{-2S[A]/\hbar}\mu_{\mathfrak{G}})$. Let us denote the gauge covariant derivative associated with a connection $A$ by $\hnabla$ and the corresponding exterior derivative by $d^{\hnabla}$.      
\begin{theorem}
\label{main}
Let $F=d^{\hnabla}A$ be the curvature of the principle $G-$bundle $\mathfrak{P}$ over $\mathbb{R}^{1+n},~n=2,3$ and the associated Yang-Mills theory is defined by the action $I_{YM}:=-\frac{1}{4}\int_{\mathbb{R}^{1+n}}F_{\mu\nu}\cdot F^{\mu\nu} d^{n+1}x$~\footnote{Usually, the Yang-Mills action is defined as $-\frac{1}{4g^{2}_{YM}}\int_{\mathbb{R}^{1+n}}F_{\mu\nu}\cdot F^{\mu\nu} d^{n+1}x$, where $g_{YM}$ is the coupling constant. In my convention, I absorb it in the definition of $F$ and the coupling only appears through the commutation relation}. Let us assume that the corresponding quantum theory exists that has a normalizable ground state $\Psi[A]:=N_{\hbar}e^{-S[A]/\hbar}, N_{\hbar}\in \mathbb{C}-\{0\}$. Then the associated Hamiltonian operator $\widehat{H}$ verifies the mass gap 
\begin{eqnarray}
\label{eq:1}
 \Delta E\geq \frac{\hbar^{2}\Delta }{2}   
\end{eqnarray}
if the regularized and renormalized Bakry-Emery Ricci curvature $\mathcal{R}^{B.E}$ of the orbit space $\mathcal{A}/\widehat{\mathcal{G}}$ admits the following uniform bound after the removal of the regulator
\begin{eqnarray}
\label{eq:2}
\mathcal{R}^{B.E}(\alpha[A],\alpha[A])\geq \Delta\mathfrak{G}(\alpha[A],\alpha[A]),
\end{eqnarray}
where $\Delta>0$ being a constant and $\mathcal{R}^{B.E}$ is defined as follows 
\begin{eqnarray}
\mathcal{R}^{B.E}(\alpha[A],\alpha[A]):=\int_{x,x^{'}y,y^{'}}\nonumber\left(\underbrace{\mathfrak{G}^{A^{M}_{k}(y^{'})A^{L}_{k}(y)}\mathfrak{R}_{A^{M}_{k}(y^{'})A^{N}_{n}(x^{'})A^{L}_{k}(y)A^{P}_{i}(x)}\alpha^{N}_{n}(x^{'})\alpha^{P}_{i}(x)}_{I}\right.\\\nonumber 
\left.+\underbrace{\frac{2}{\hbar}\mathfrak{G}^{A^{P}_{I}(x)A^{Q}_{J}(x^{'})}\mathfrak{G}^{A^{M}_{K}(y)A^{N}_{L}(y^{'})}\frac{\mathfrak{D}}{\mathfrak{D}A^{P}_{I}(x)}\frac{\mathfrak{D}S[A]}{\mathfrak{D}A^{M}_{K}(y)}\alpha^{Q}_{I}(x^{'})\alpha^{N}_{L}(y^{'})}_{II}\right)    
\end{eqnarray}
and $\alpha[A]\in T_{A}(\mathcal{A}/\widehat{\mathcal{G}})$.
\end{theorem}
\begin{remark}
Note that the functional $S[A]$ is not to be confused with the Yang-Mills action functional. Since $\Psi[A]:=N_{\hbar}e^{-S[A]/\hbar}$ verifies the functional Schrodinger equation, $S[A]$ verifies a functional Ricatti type equation.  
\end{remark}

\noindent Let us try to understand the theorem and compare it with previous studies. Notice that the term $I$ corresponds to a pure geometric contribution (this term blows up without proper regularization). First, let us focus on the inequalities \ref{eq:1} and \ref{eq:2} and consider the $2+1$ dimensional case. $\Delta E$ is contributed by the spectral gap in the covariant Laplacian defined on the orbit space $\mathcal{A}/\widehat{\mathcal{G}}$. The Laplacian involves repeated functional derivatives at the same point and therefore singular. Therefore, one defines a regularized version through the introduction of a cut-off scale $\chi$ and $\chi\to\infty$ reproduces the original Laplacian. In the $2+1$ dimensional case, this entity diverges logarithmically with the cut-off scale $\chi$. Now $\Delta$ in the right-hand side of the inequality \ref{eq:1} is the spectral gap of the Bakry-Emery Ricci curvature of the orbit space $\mathcal{A}/\widehat{\mathcal{G}}$. This is logarithmically singular in the cut-off parameter $\chi$ and the singular term is exactly the same as that of $\Delta E$. Therefore, the singular parts of $\Delta E$ and $\Delta$ vary uniformly as one changes the cut-off parameter $\chi$. We will discuss this in more detail in section \ref{proof}.\\
The metric is induced by the kinetic term of the action, the term $I$ in the gap theorem \ref{main} is essentially a kinetic contribution while the term $II$ is a contribution from the potential. Due to Lorentz covariance, the kinetic and potential contributions are not completely independent as we shall see in section \ref{technical}, $S[A]$ is governed by both the metric and the potential.  Later in section \ref{explicit}, I shall perform an explicit calculation for the term $I$ and show that it has a uniform positive lower bound. In a series of works, \cite{karabali1996gauge,karabali1997gauge,karabali1998planar,karabali1998vacuum} handled $2+1$ dimensional Yang-Mills theory using the Hamiltonian approach and a gauge invariant matrix parametrization of gauge fields. They have computed the invariant volume element of the orbit space $\mathcal{A}/\widehat{\mathcal{G}}$ corresponding to a metric arising from the kinetic part of the action in terms of WZW action (in fact they proved in this $ 2+1$ dimensional setting that the volume of the orbit space is finite). Using this construction they have obtained a mass gap associated with the kinetic operator and the potential contribution is considered in an improved perturbation series.
At a heuristic level, the mass in the propagator of a gauge invariant definition of gluon verifies $\sim \frac{g^{2}_{YM}c_{A}}{2\pi}+O(k^{2})$ ($c_{A}$ is the Casimir of adjoint representation of $\mathfrak{g}$), where $\frac{g^{2}_{YM}c_{A}}{2\pi}$ is the kinetic contribution while $O(k^{2})$ term appears due to the potential that does not contribute by a strictly positive number since it can be made to be arbitrarily small by choosing large wavelengths. This gauge-invariant gluon mass ultimately leads to a positive gap in the spectra of the Hamiltonian or a `mass gap'. Our approach seems quite similar to this approach in spirit. Note again that the term $I$ is a pure kinetic contribution that provides a strictly positive gap and the term $II$ encodes the potential contribution that is expected to be non-negative for rapidly rising potentials such as Yang-Mills potentials wherever the latter does not admit flat directions. We shall sketch rather heuristic evidence towards non-negative definiteness of the term $II$ in section \ref{explicit}. It is almost tempting to state that in our analysis $I$ is a fixed positive number while $II$ is $O(k^{2})$. However, we shall observe that for Lorentz covariant field theories, this is not quite the case. Many years ago, Feynman \cite{feynman1981qualitative} presented a qualitative argument in supporting a strictly positive mass gap based on the geometry of the orbit space. The argument goes as follows. The ground state wave functional is essentially node-less and it can be taken as a real positive since the potential is only a functional of field configurations, not their time derivatives. The first excited state is orthogonal to the ground state and is positive in some regions of the orbit space and negative in others. The kinetic energy is essentially a gradient energy on the orbit space and scales with the inverse of the square distance between two regions where the first excited state is positive and negative. Feynman argued that this distance can not be arbitrarily large leading to a strictly positive lower bound on the kinetic energy that is supposed to be the mass gap of the theory. On finite-dimensional Riemannian manifolds, this is essentially equivalent to finding the spectral gap of the Laplace-Beltrami operator. As mentioned in the previous paragraph, lack of compactness causes a serious issue in infinite dimensions if one tries to carry out a procedure such as that of Lichnerowicz \cite{lichnerowicz1970varietes,lichnerowicz1971varietes}. However, instead of performing a Lichnerowicz-type estimate, one can perform a weighted estimate where the effect of potential is taken into consideration through a suitable weight. In fact, this article was motivated in part by the desire to adapt the geometric arguments of Feynman and that of Karabali and Nair \cite{karabali1996gauge}-\cite{karabali1998vacuum} (note that \cite{karabali1996gauge}-\cite{karabali1998vacuum} obtained a measure on the orbit space of $2+1$ dimensional Yang-Mills theory and this result was not available when Singer studied the geometry of the orbit space). I should mention that \cite{orland1} presented some results on the orbit space geometry and a proposal for the mass gap. In addition to these geometric arguments, I wish to point out that there are recent studies by \cite{frasca1,frasca2} on the mass gap estimates based on a direct approach of integrating the Schwinger-Dyson equations in both $2+1$ and $3+1$ dimensions. Substantial progress is made in the context of lattice gauge theory as well \cite{athenodorou2020glueball,lucini2010glueball}.

In the context of the weighted manifolds introduced by Lichnerowicz \cite{lichnerowicz1970varietes,lichnerowicz1971varietes}, \textit{Bakry-Emery} curvature (terms $I$ and $II$ together in the gap theorem) naturally appears (see \cite{lott2003some} for geometric properties of the Bakry-Emery Ricci tensor, on finite-dimensional weighted manifolds). It also appears in the context of scalar-tensor gravitational theories, including Brans-Dicke theory \cite{brans1961mach}, theories with Kaluza-Klein dimensional reduction \cite{branding2019stable} apart from the celebrated study by \cite{bakry1985diffusions}. Studies by \cite{galloway2014cosmological,woolgar2016cosmological} provide examples of the appearance of this modified Ricci curvature in the context of Lorentzian geometry.  In a finite-dimensional setting with potential satisfying suitable convexity conditions, then a spectral gap estimate for the Hamiltonian operator is given by the bound on the Bakry-Emery Ricci curvature \cite{moncrief}. For example, for a harmonic oscillator on a flat space, the Bakry-Emery Ricci tensor produces the exact gap that is presented in every quantum mechanics textbook. Even though in such a case the ordinary Ricci curvature vanishes, the Hessian of the negative Logarithm of the ground state wave function contributes in a strictly positive manner to produce the exact gap \cite{moncrief}. The spectral gap in the Hamiltonian of the Yang-Mills theory is absent in the perturbation theory. Recall in the perturbation theory, one splits the full gauge-invariant Lagrangian into an exactly soluble part and interactions. This procedure, however, corresponds to the breaking of the original gauge-invariance in the sense that the original gauge group $SU(N)$ undergoes a splitting $SU(N)\to \underbrace{U(1)\times U(1)\times U(1)\times U(1)\times \cdot\cdot\cdot\cdot U(1)}_{N^{2}-1}$. In fact, if one investigates our main theorem \ref{main} closely, it then becomes clear that the result is fully non-perturbative in nature since one requires a uniform bound on the Bakry-Emery Ricci curvature over the entire orbit space not just a neighborhood of the flat connection. 

One subtle issue that arises in the context of $3+1$ Yang-Mills theory is that the physical constants can not produce a mass scale purely based on dimensional analysis and one has to introduce a length scale that is to be fixed by measuring the mass of the lowest glu-ball state (\textit{dimensional transmutation}). In our context, this scale is introduced through the regularization of the trace of the Riemann curvature of the infinite-dimensional configuration space of the gauge theory. This seems to be a geometrically natural operation based on the observation that the Riemann curvature is not of trace class. Therefore, to make sense of the Ricci curvature, one requires suitable regularization. For purely dimensional reasons, it seems necessary to introduce a length scale in $3+1$ dimensions to regularize the Ricci curvature. It would be very interesting to study this issue of introducing length/ energy scale in $3+1$ dimensions from the perspective of the renormalization group flow (see \cite{balaban1988convergent} for the aspects of renormalization group flow in lattice gauge theory). I wish to investigate this in the future.

\section{Geometry of the orbit space $\mathcal{A}/\widehat{\mathcal{G}}$}
\label{yangmills}
\noindent We denote by $\mathfrak{P}$ a $C^{\infty}$ principal bundle with base an $n+1$ dimensional Lorentzian manifold $M$ and a Lie group $G$. We assume that $G$ is compact (for physical purposes) and therefore admits a positive definite non-degenerate bi-invariant metric. Its Lie algebra $\mathfrak{g}$ by construction admits an adjoint invariant, positive definite scalar product denoted by  $\langle~,~\rangle$ which enjoys the property: for $A,B,C\in \mathfrak{g}$,
\begin{eqnarray}
\label{eq:adinvpp}
\langle[A,B],C\rangle=\langle A,[B,C]\rangle.
\end{eqnarray}
as a consequence of adjoint invariance.
A Yang-Mills connection is defined as a $1-$form $\omega$ on $\mathfrak{P}$ with values in $\mathfrak{g}$ endowed with compatible properties. It's representative in a local trivialization of $\mathfrak{P}$ over $U\subset M$ 
\begin{eqnarray}
\varphi: p\mapsto (x,a),~p\in \mathfrak{P},~x\in U,~a\in G
\end{eqnarray}
is the $1-$form $s^{*}\omega$ on $U$, where $s$ is the local section of $\mathfrak{P}$ corresponding canonically to the local trivialization  $s(x)=\varphi^{-1}(x,e)$, called a \textit{gauge}. Let $A_{1}$ and $A_{2}$ be representatives of $\omega$ in gauges $s_{1}$ and $s_{2}$ over $U_{1}\subset M$ and $U_{2}\in M$. In $U_{1}\cap U_{2}$, one has 
\begin{eqnarray}
\label{eq:gauge}
A_{1}=Ad(u^{-1}_{12})A_{2}+u_{12}\Theta_{MC},
\end{eqnarray}
where $\Theta_{MC}$ is the Maurer-Cartan form on $G$, (or $A_{1}\mapsto u_{12}^{-1}A_{1}u_{12}+u_{12}du^{-1}_{12}$) and $u_{12}:U_{1}\cap U_{2}\to G$ generates the transformation between the two local trivializations: \begin{eqnarray}
s_{1}=R_{u_{12}}s_{2},
\end{eqnarray}
$R_{u_{12}}$ is the right translation on $\mathfrak{P}$ by $u_{12}$. Given the principal bundle $\mathfrak{P}\to  M$, a Yang-Mills potential $A$ on $M$ is a section of the fibered tensor product $T^{*}M\otimes_{M}\mathfrak{P}_{Affine,\mathfrak{g}}$ where $\mathfrak{P}_{Affine,\mathfrak{g}}$ is the affine bundle with base $M$ and typical fiber $\mathfrak{g}$ associated to $\mathfrak{P}$ via relation (\ref{eq:gauge}) (in other words, the connection does not transform as a tensor under a gauge transformation). If $\widehat{A}$ is another Yang-Mills potential on $M$, then $A-\widehat{A}$ is a section of the tensor product of vector bundles $T^{*}M\otimes_{M}\mathfrak{P}_{Ad,\mathfrak{g}}$, where $\mathfrak{P}_{Ad,\mathfrak{g}}:=\mathfrak{P}\times _{Ad}\mathfrak{g}$ is the vector bundle associated to $\mathfrak{P}$ by the adjoint representation of $G$ on $\mathfrak{g}$ (the difference of two connections does transform as a tensor under a gauge transformation). There is an inner product in the fibers of $\mathfrak{P}_{Ad,\mathfrak{g}}$, deduced from that on $\mathfrak{g}$. The curvature $\Omega$ of the connection $\omega$ considered as a $1-$ form on $\mathfrak{P}$ is a $\mathfrak{g}$-valued $2-$form on $\mathfrak{P}$. Its representative in a gauge where $\omega$ is represented by $A$ is given by 
\begin{eqnarray}
F:=dA+[A,A],
\end{eqnarray}
and the relation between two representatives $F_{1}$ and $F_{2}$ on $U_{1}\cap U_{2}$ is $F_{1}=Ad(u^{-1}_{12})F_{2}$ and therefore $F$ is a section of the vector bundle $\Lambda^{2}T^{*}M\otimes _{M}\mathfrak{P}_{Ad,\mathfrak{g}}$. For a section $\mathfrak{O}$ of the vector bundle $\otimes^{k}T^{*}M\otimes _{M}\mathfrak{P}_{Ad,\mathfrak{g}}$, a natural covariant derivative is defined as follows 
\begin{eqnarray}
\label{eq:covariant}
\widehat{\nabla}\mathfrak{O}:=\nabla\mathfrak{O}+[A,\mathfrak{O}],
\end{eqnarray}
where $\nabla$ is the usual covariant derivative induced by the Lorentzian structure of $M$ and by construction $\widehat{\nabla}\mathfrak{O}$ is a section of the vector bundle $\otimes^{k+1}T^{*}M\otimes _{M}\mathfrak{P}_{Ad,\mathfrak{g}}$. The associated exterior derivative is denoted by $d^{\hnabla}$. The Yang-Mills coupling constant $g_{YM}$ is kept hidden within the structure constants of the commutators.

The classical Yang-Mills equations (in the absence of sources)
correspond to setting the natural (spacetime and gauge as defined in \ref{eq:covariant}) covariant divergence of this curvature
two-form $F$ to zero. By virtue of its definition in terms of the connection, this curvature also satisfies
the Bianchi identity that asserts the vanishing of its gauge covariant exterior derivative. Taken
together these equations provide a geometrically natural nonlinear generalization of Maxwell's
equations (when the latter are written in terms of a `vector potential') and of course, play a
fundamental role in modern elementary particle physics. If nontrivial bundles are considered
or nontrivial spacetime topologies are involved, then the foregoing so-called `local trivializations'
of the bundles in question must be patched together to give global descriptions but, by
the covariance of the formalism, there is a natural way of carrying out this patching procedure
at least over those regions of spacetime where the connections are well-defined. From now on, we set $M=\mathbb{R}^{1+n},~n=2,3$ equipped with the Minkowski metric $\eta$. In addition, in a chosen Lie algebra basis we write the commutation $[~,~]$ on $\mathfrak{g}$ explicitly in terms of the structure constants i.e., $[A_{i},A_{j}]^{P}=f^{PQR}A^{Q}_{i}A^{R}_{j}$ and absorb the Yang-Mills coupling constant $g_{YM}$ in the structure constants $f^{PQR}$. We denote the space of connections in Schwartz class by $\mathcal{A}$. The following lemma yields a local expression for the metric on the orbit space $\mathcal{A}/\mathcal{G}$ (note that \cite{singer1981geometry, babelon1981riemannian} also obtained metrics on the orbit space). We provide an explicit expression for the metric for the convenience of forthcoming calculations. 
\begin{lemma}
\label{metriclemma1}
Let $F=d^{\hnabla}A$ be the curvature of the principle $G-$bundle $\mathfrak{P}$ over $\mathbb{R}^{1+n}$. The associated Yang-Mills action functional is defined as
$I_{YM}=-\frac{1}{4}\int_{\mathbb{R}^{1+n}}\langle F~_{\alpha\beta},F^{\alpha\beta}\rangle$. The metric induced by the action functional $I$ on the orbit space $\mathcal{A}/\mathcal{G}$ verifies the following expression in the local Coulomb coordinates (i.e., connection verifies $\widehat{\nabla}^{i}(A_{i}-0)=\partial_{i}A_{i}+[A_{i},A_{i}]=\partial_{i}A_{i}=0$ in a small enough open neighborhood of the flat connection $A_{i}=0$) in a distributional sense
\begin{eqnarray}
\mathfrak{G}[A]_{A^{P}_{i}(x)A^{Q}_{j}(x^{'})}=\delta_{ij}\delta_{PQ}\delta(x-x^{'})\nonumber+f^{PRV}A^{V}_{i}(x)\Delta^{-1}_{A}(x,x^{'})f^{RUQ}A^{U}_{j}(x^{'}), 
\end{eqnarray}
where $f^{PQR}$ are the structure constants defined via $[A_{i},A_{j}]^{P}=f^{PQR}A^{Q}_{i}A^{R}_{j}$ in a chosen Lie algebra basis. Here $\mathcal{G}$ is the group of automorphisms of the bundle $\mathfrak{P}$ i.e., the group of gauge transformations (under a gauge transformation $\varphi(x)$, a connection $A\in \mathcal{A}$ transforms as $A\mapsto \varphi A\varphi^{-1}+\varphi d\varphi^{-1}$).
\end{lemma}
\begin{remark} Note that the space $\mathcal{A}/\mathcal{G}$ is in general \textit{not} a manifold since the group action $\mathcal{G}$ on $\mathcal{A}$ is not free due to the potential presence of gauge \textit{symmetry} i.e., the gauge transformations that leave a connection $A$ invariant or equivalently solutions of the equation $\varphi A\varphi^{-1}+\varphi d\varphi^{-1}=A$ or $d\varphi^{-1}+[A,\varphi^{-1}]=0$ i.e., the elements of the bundle automorphism group that are covariantly constant. However, we can work with the space of irreducible connections i.e., $\mathcal{A}/\widehat{\mathcal{G}}$ where $\widehat{\mathcal{G}}$ is obtained by modding out the covariantly constant elements of $\mathcal{G}$. $\mathcal{A}/\widehat{\mathcal{G}}$ is an infinite dimensional manifold. This property is important as we shall see in the later sections. From now on $\mathcal{A}/\widehat{\mathcal{G}}$ is to be understood as the space of connections belonging to Schwartz space $\mathcal{S}(\mathbb{R}^{n})$.
\end{remark}
\begin{proof}
The Gauss Law constraint
\begin{eqnarray}
\widehat{\nabla}_{\nu}F^{0\nu}=0
\end{eqnarray}
yields 
\begin{eqnarray}
\widehat{\nabla}_{i}\widehat{\nabla}_{i}A_{0}=\nabla_{i}\partial_{0}A_{i}+[A_{i},\partial_{0}A_{i}]
\end{eqnarray}
which after an application of the Coulomb coordinate condition $\partial_{i}A_{i}=0$ yields 
\begin{eqnarray}
\label{eq:A}
\widehat{\nabla}_{i}\widehat{\nabla}_{i}A_{0}=[A_{i},\partial_{0}A_{i}]
\end{eqnarray}
and therefore $A_{0}$ may be obtained by formally inverting the elliptic operator $\widehat{\nabla}_{i}\widehat{\nabla}_{i}=\Delta_{A}$
\begin{eqnarray}
A_{0}=\Delta^{-1}_{A}([A_{i},\partial_{t}A_{i}]).
\end{eqnarray}
Now write the usual commutation for the elements of the Lie algebra $\mathfrak{g}$
\begin{eqnarray}
[\chi_{i},\chi_{j}]^{P}=f^{ABC}\chi^{B}_{i}\chi^{C}_{j}~i.e.,~[A_{i},\partial_{t}A_{i}]^{P}=f^{ABC}A^{B}_{i}\partial_{t}A^{C}_{i}.
\end{eqnarray}
We may obtain $A^{P}_{0}$ by formally inverting $\Delta_{A}$
\begin{eqnarray}
A^{P}_{0}=\Delta^{-1}_{A}(f^{PQR}A^{Q}_{i}\partial_{t}A^{R}_{i}).
\end{eqnarray}
In the Coulomb coordinate, the action functional $I_{YM}=\int_{\mathbb{R}^{1,n}}\left(\frac{1}{2}F^{P}~_{0i}F^{P}~_{0i}-\frac{1}{4}F^{P}~_{ij}F^{P}~_{ij}\right)d^{n+1}x$ 
takes the following form  
\begin{eqnarray}
I_{YM}=\int_{\mathbb{R}^{1,n}}\left(\frac{1}{2}\partial_{t}A^{P}_{i}\partial_{t}A^{P}_{i}-\partial_{t}A^{P}_{i}\partial_{i}A^{P}_{0}+\frac{1}{2}\partial_{i}A^{P}_{0}\partial_{i}A^{P}_{0}+\partial_{t}A^{P}_{i}[A_{0},A_{i}]^{P}\right.\\\nonumber 
\left.-\partial_{i}A^{P}_{0}[A_{0},A^{T}_{i}]^{P}+\frac{1}{2}[A_{0},A_{i}]^{P}[A_{0},A_{i}]^{P}-\frac{1}{4}F^{P}~_{ij}F^{P}~_{ij}\right)d^{n+1}x\\\nonumber
=\int_{\mathbb{R}^{1,n}}\left(\frac{1}{2}\partial_{t}A^{P}_{i}\partial_{t}A^{P}_{i}-\frac{1}{2}A^{P}_{0}\Delta A^{P}_{0}+\partial_{t}A^{P}_{i}[A_{0},A_{i}]^{P}-\partial_{i}A^{P}_{0}[A_{0},A_{i}]^{P}+\right.\\\nonumber 
\left.\frac{1}{2}[A_{0},A_{i}]^{P}[A_{0},A_{i}]^{P}-\frac{1}{4}F^{P}~_{ij}F^{P}~_{ij}\right)d^{n+1}x
-\int_{\partial\mathbb{R}^{1,n}}(\partial_{t}A^{P}_{i}A^{P}_{0}-\frac{1}{2}A^{P}_{0}\partial_{i}A^{P}_{0}).
\end{eqnarray}
Notice that there are problematic terms such as $\int_{\mathbb{R}^{1,n}}\partial_{i}A^{P}_{0}[A_{0},A^{T}_{i}]^{P}$. However, this term is canceled in a point-wise manner after expanding $\Delta A^{P}_{0}$ using equation (\ref{eq:A})
\begin{eqnarray}
-\frac{1}{2}A^{P}_{0}\Delta A^{P}_{0}-\partial_{i}A^{P}_{0}[A_{0},A_{i}]^{P}=A^{P}_{0}[A_{i},\partial_{i}A_{0}]^{P}+\frac{1}{2}A^{P}_{0}[A_{i},[A_{i},A_{0}]]^{P}\nonumber\\\nonumber-\frac{1}{2}A^{P}_{0}[A_{i},\partial_{t}A_{i}]^{P}
-\partial_{i}A^{P}_{0}[A_{0},A_{i}]^{P}.
\end{eqnarray}
Now $A^{P}_{0}[A_{i},\partial_{i}A_{0}]^{P}-\partial_{i}A^{P}_{0}[A_{0},A_{i}]^{P}$ vanishes due to the property (\ref{eq:adinvpp}). Therefore ignoring the boundary terms (assuming fields belong to the Schwartz space), the action reads 
\begin{eqnarray}
I_{YM}=\int_{\mathbb{R}^{n+1}}\left(\frac{1}{2}\partial_{t}A^{P}_{i}\partial_{t}A^{P}_{i}+\frac{1}{2}A^{P}_{0}[A_{i},\partial_{t}A_{i}]^{P}\nonumber
-\frac{1}{4}F^{P}~_{ij}F^{P}~_{ij}\right)d^{n+1}x.
\end{eqnarray}
Now after an explicit calculation using the Lie-algebra commutation relation, one writes the Lagrangian in the usual form, that is, as the difference between the kinetic and potential terms (through solving the Gauss-law constraint i.e., $A_{0}=\Delta^{-1}_{A}([A_{i},\partial_{t}A_{i}])$)
\begin{eqnarray}
L=\int_{(\mathbb{R}^{n})^{2}}\left(\frac{1}{2}\partial_{t}A^{P}_{i}(x)\partial_{t}A^{P}_{i}(x^{'})\delta(x-x^{'})\right.\\\nonumber
\left.+\frac{1}{2}f^{PQR}A^{Q}_{i}(x)\partial_{t}A^{R}_{i}(x)\Delta^{-1}_{A}(x,x^{'})f^{PUV}A^{U}_{k}(x^{'})\partial_{t}A^{V}_{k}(x^{'})\right)\\\nonumber 
-\frac{1}{4}\int_{\mathbb{R}^{n}}\mathcal{F}^{P}~_{ij}\mathcal{F}^{P}~_{ij}\\\nonumber 
=\int_{\mathbb{R}^{n}\times \mathbb{R}^{n}}\frac{1}{2}\mathfrak{G}[A]_{A^{P}_{i}(x)A^{Q}_{j}(x^{'})}\partial_{t}A^{P}_{i}(x)\partial_{t}A^{Q}_{j}(x^{'})-\frac{1}{4}\int_{\mathbb{R}^{n}}\mathcal{F}^{P}~_{ij}\mathcal{F}^{P}~_{ij},
\end{eqnarray}
where 
\begin{eqnarray}
\mathfrak{G}[A]_{A^{P}_{i}(x)A^{Q}_{j}(x^{'})}=\delta_{ij}\delta_{PQ}\delta(x-x^{'})\nonumber+f^{PRV}A^{V}_{i}(x)\Delta^{-1}_{A}(x,x^{'})f^{RUQ}A^{U}_{j}(x^{'}). 
\end{eqnarray}
This concludes the proof of the lemma. Note that $\Delta^{-1}(x,x^{'}):=\frac{1}{4\pi}\frac{-1}{|x-x^{'}|}$ for $n=3$ and $\Delta^{-1}(x,x^{'}):=\frac{1}{2}\ln(|x-x^{'}|/a)$ for $n=2$, a some fixed constant with dimension of length. This metric was obtained by \cite{singer1981geometry,babelon1981riemannian} by a different method (mention that). Essentially, $\mathfrak{G}[A]_{A^{P}_{i}(x)A^{Q}_{j}(x^{'})}$ is a distribution.
\end{proof}
\begin{proposition}
\textit{$\mathfrak{G}$ is a Riemannian metric.}
\end{proposition}
\begin{proof} Follows from the positive definiteness of the Kinetic energy (a consequence of the compactness of the gauge group). 
\end{proof}
The Riemannian metric induced by the action functional on the configuration space $\mathcal{A}/\widehat{\mathcal{G}}$ is in general curved. As such one may compute the Riemann curvature of this metric $\mathfrak{G}[A]$ at any point $\widehat{A}$ of $\mathcal{A}/\widehat{\mathcal{G}}$ by explicit calculations or by expanding it in the normal coordinate around $\widehat{A}$. We compute the Riemann curvature at $\widehat{A}=0$ in the following lemma. Note that \cite{singer1981geometry, babelon1981riemannian} performed similar calculations as well.\\
\noindent We define the formal single trace operation on sections of suitable bundles on the infinite-dimensional manifold $\mathcal{A}/\widehat{\mathcal{G}}$ as follows 
\begin{eqnarray}
(\tr\Phi)_{A^{P_{1}}_{I_{1}}(x_{1})A^{P_{2}}_{I_{2}}(x_{2})A^{P_{3}}_{I_{3}}(x_{3})\cdot\cdot\cdot\cdot \widehat{A}^{P_{i}}_{I_{i}}(x_{i})\cdot\cdot\cdot\cdot \widehat{A}^{P_{j}}_{I_{j}}(x_{J})\cdot\cdot A^{P_{n}}_{I_{n}}(x_{n})}\\\nonumber :=\int_{x_{i},x_{j}}\mathfrak{G}^{A^{P_{i}}_{I_{i}}(x_{i})A^{P_{j}}_{I_{j}}(x_{J})} \Phi_{A^{P_{1}}_{I_{1}}(x_{1})A^{P_{2}}_{I_{2}}(x_{2})A^{P_{3}}_{I_{3}}(x_{3})\cdot\cdot\cdot\cdot A^{P_{i}}_{I_{i}}(x_{i})\cdot\cdot\cdot\cdot A^{P_{j}}_{I_{j}}(x_{J})\cdot\cdot A^{P_{n}}_{I_{n}}(x_{n})},
\end{eqnarray}
where the hat symbol implies the deletion of the respective indices.
For example, if we consider Riemann curvature i.e., $\Phi:=\mathcal{R}_{A^{P_{1}}_{I_{1}}(x_{1})A^{P_{2}}_{I_{2}}(x_{2})A^{P_{3}}_{I_{3}}(x_{3})A^{P_{4}}_{I_{4}}(x_{4})}$, then the formal Ricci curvature would be defined as follows 
\begin{eqnarray}
\mathcal{R}ic_{A^{P_{2}}_{I_{2}}(x_{2})A^{P_{4}}_{I_{4}}(x_{4})}:=\int_{x_{1},x_{3}}\mathfrak{G}^{A^{P_{1}}_{I_{1}}(x_{1})A^{P_{3}}_{I_{3}}(x_{3})}\mathcal{R}_{A^{P_{1}}_{I_{1}}(x_{1})A^{P_{2}}_{I_{2}}(x_{2})A^{P_{3}}_{I_{3}}(x_{3})A^{P_{4}}_{I_{4}}(x_{4})}.
\end{eqnarray}
Here $\mathfrak{G}^{A^{P_{1}}_{I_{1}}(x_{1})A^{P_{3}}_{I_{3}}(x_{3})}:=(\mathfrak{G}^{-1})^{A^{P_{1}}_{I_{1}}(x_{1})A^{P_{3}}_{I_{3}}(x_{3})}$ (notice the Hamiltonian reads\\ $\frac{1}{2}\int_{\mathbb{R}^{n}\times \mathbb{R}^{n}}(\mathfrak{G}^{-1})^{A^{P}_{I}(x)A^{Q}_{J}(y)}(\pi^{T})^{P}_{I}(x)(\pi^{T})^{Q}_{J}(y)+\text{potential}$, where $(\pi^{T})^{P}_{I}$ is the transverse momentum conjugate to $A^{P}_{I}\in \mathcal{A}/\widehat{\mathcal{G}}$). The following lemma provides an explicit expression of the inverse metric that is obtained through a Legendre transformation.
\begin{lemma}
\label{metriclemma2}
The inverse metric $\mathfrak{G}^{-1}$ induced by the kinetic part of the Yang-Mills Lagrangian on the orbit space $\mathcal{A}/\widehat{\mathcal{G}}$ in local Coulomb chart around $A=0$ reads 
\begin{eqnarray}
(\mathfrak{G}^{-1})^{A^{P}_{i}(x)A^{Q}_{j}(y)}=\delta(x-y)\delta^{PQ}\delta_{ij}-f^{PVU}A^{U}_{i}(x)\Delta^{-1}_{A}(x,y)f^{VRQ}A^{R}_{j}(y)
\end{eqnarray}
\end{lemma}
\begin{proof}
Recall the momentum conjugate to $A^{P}_{i}$
\begin{eqnarray}
\pi^{P}_{i}=\frac{\delta L}{\delta(\partial_{t}A^{P}_{i})}=F^{P}_{0i}
\end{eqnarray}
and the definition of the classical Hamiltonian 
\begin{eqnarray}
H=\int_{\mathbb{R}^{n}}\pi^{P}_{i}\partial_{t}A^{P}_{i}-L=\int_{\mathbb{R}^{n}}\left(\frac{1}{2}\pi^{P}_{i}\pi^{P}_{i}-A^{P}_{0}(\partial_{i}\pi^{P}_{i}+[A_{i},\pi_{i}]^{P})+\frac{1}{4}F^{P}_{ij}F^{P}_{ij}\right).
\end{eqnarray}
Now $A\in \mathcal{A}/\widehat{\mathcal{G}}$ verifies $\partial_{i}A^{P}_{i}=0$ in the local Coulomb chart around $A=0$. Therefore the conjugate momentum $\pi^{P}_{i}$ is decomposed into the transverse and longitudinal parts
\begin{eqnarray}
\pi^{P}_{i}=(\pi^{T})^{P}_{i}+(\pi^{L})^{P}_{i}
\end{eqnarray}
that verify 
\begin{eqnarray}
\partial_{i}(\pi^{T})^{P}_{i}=0,~\partial_{i}(\pi^{L})^{P}_{i}+[A_{i},(\pi^{L})_{i}]^{P}=-[A_{i},\pi^{T}_{i}]^{P}.
\end{eqnarray}
Writing $\pi^{L}_{i}=\partial_{i}\kappa^{P}+[A_{i},\kappa]^{P}=\widehat{\nabla}_{i}\kappa^{P}$ yields $\kappa^{P}=-\Delta^{-1}_{A}[A_{j},\pi^{T}_{j}]^{P}$ and $\pi^{L}_{i}=-\widehat{\nabla}_{i}(\Delta^{-1}_{A}[A_{j},\pi^{T}_{j}]^{P})$. After substituting $\pi^{L}$ in the Hamiltonian, it is a functional on the co-tangent bundle of the orbit space $\mathcal{A}/\widehat{\mathcal{G}}$ reads 
\begin{eqnarray}
H=\int_{\mathbb{R}^{n}}\left(\frac{1}{2}(\pi^{T})^{P}_{i}(\pi^{T})^{P}_{i}\nonumber+(\pi^{T})^{P}_{i}(\pi^{L})^{P}_{i}+\frac{1}{2}(\pi^{L})^{P}_{i}(\pi^{L})^{P}_{i}+\frac{1}{4}F^{P}_{ij}F^{P}_{ij}\right)\\
=\int_{\mathbb{R}^{n}}\left(\frac{1}{2}(\pi^{T})^{P}_{i}(\pi^{T})^{P}_{i}\nonumber+(\pi^{T})^{P}_{i}\widehat{\nabla}_{i}\kappa^{P}+\frac{1}{2}\widehat{\nabla}_{i}\kappa^{P}\widehat{\nabla}_{i}\kappa^{P}+\frac{1}{4}F^{P}_{ij}F^{P}_{ij}\right)\\
=\int_{\mathbb{R}^{n}}\left(\frac{1}{2}(\pi^{T})^{P}_{i}(\pi^{T})^{P}_{i}\nonumber-\widehat{\nabla}_{i}(\pi^{T})^{P}_{i}\kappa^{P}-\frac{1}{2}\kappa^{P}\Delta_{A}\kappa^{P}+\frac{1}{4}F^{P}_{ij}F^{P}_{ij}\right)+\int_{\partial \mathbb{R}^{n}}(\kappa^{P}(\pi^{T})^{P}_{i}+\kappa^{P}\widehat{\nabla}_{i}\kappa^{P})\widehat{n}^{i}\\\nonumber
=\int_{\mathbb{R}^{n}}\left(\frac{1}{2}(\pi^{T})^{P}_{i}(\pi^{T})^{P}_{i}\nonumber+\frac{1}{2}[A_{i},\pi^{T}_{i}]^{P}\kappa^{P}+\frac{1}{4}F^{P}_{ij}F^{P}_{ij}+\int_{\partial \mathbb{R}^{n}}(\kappa^{P}(\pi^{T})^{P}_{i}+\kappa^{P}\widehat{\nabla}_{i}\kappa^{P})\widehat{n}^{i}\right)\\\nonumber 
=\int_{\mathbb{R}^{n}}\left(\frac{1}{2}(\pi^{T})^{P}_{i}(\pi^{T})^{P}_{i}\nonumber-\frac{1}{2}[A_{i},\pi^{T}_{i}]^{P}\Delta^{-1}_{A}[A_{j},\pi^{T}_{j}]^{P}+\frac{1}{4}F^{P}_{ij}F^{P}_{ij}\right)+\int_{\partial \mathbb{R}^{n}}(\kappa^{P}(\pi^{T})^{P}_{i}+\kappa^{P}\widehat{\nabla}_{i}\kappa^{P})\widehat{n}^{i}\\\nonumber 
=\int_{\mathbb{R}^{n}\times \mathbb{R}^{n}}\frac{1}{2}(\mathfrak{G}^{-1})^{A^{P}_{i}(x)A^{Q}_{j}(y)}(\pi^{T})^{P}_{i}(x)(\pi^{T})^{Q}_{j}(y)+\int_{\mathbb{R}^{n}}\frac{1}{4}F^{P}_{ij}F^{P}_{ij}+\int_{\partial \mathbb{R}^{n}}(\kappa^{P}(\pi^{T})^{P}_{i}+\kappa^{P}\widehat{\nabla}_{i}\kappa^{P})\widehat{n}^{i},
\end{eqnarray}
where $\widehat{n}$ is a unit normal vector to the boundary sphere $\mathbb{S}^{2}_{\infty}:=\partial \mathbb{R}^{n}$ and the inverse metric $(\mathfrak{G}^{-1})^{A^{P}_{i}(x)A^{Q}_{j}(y)}$ reads in local coordinate 
\begin{eqnarray}
(\mathfrak{G}^{-1})^{A^{P}_{i}(x)A^{Q}_{j}(y)}=\delta(x-y)\delta^{PQ}\delta_{ij}-f^{PVU}A^{U}_{i}(x)\Delta^{-1}_{A}(x,y)f^{VRQ}A^{R}_{j}(y).
\end{eqnarray}
This completes the proof.
\end{proof}

\begin{remark}
Observe $\int_{\mathbb{R}^{n}} (\mathfrak{G}^{-1})^{A^{P}_{i}(x)A^{R}_{k}(z)} \mathfrak{G}_{A^{R}_{k}(z)A^{Q}_{j}(y)}d^{n}z=\delta^{P}_{Q}\delta^{i}_{j}\delta(x-y)$.    
\end{remark}

\noindent A vital point worth mentioning is that the metric expressions obtained in lemma \ref{metriclemma1} and \ref{metriclemma2} are valid only in a chart (Coulomb) centered at the flat connection $A=0$. One can not extend this definition to the whole orbit space due to the Gribov phenomenon. The orbit space $\mathcal{A}/\widehat{G}$ is topologically non-trivial and one requires more than one chart to cover the entire orbit space. For example, suppose one chooses a Coulomb chart around another reference connection $\widehat{A}$. In that case, one may perform similar calculations by choosing the coordinate condition $\eta^{ij}\widehat{\nabla}^{\widehat{A}}_{i}(A-\widehat{A})_{j}=0$ (generalized Coulomb coordinate). However, by the covariance of the formulation, all the charts can be glued together in a compatible way to yield a global description (note that even in ordinary finite-dimensional Riemannian geometry, one is required to work with multiple charts for topologically non-trivial manifolds). In the end, computation of the gauge-invariant entities (`diffeomorphism invariant' in the context of Riemannian geometry) does not depend on the local charts.

\begin{lemma}
The formal Ricci curvature of the metric $\mathfrak{G}$ in local Coulomb coordinates at $A=0$ satisfies
\begin{eqnarray}
\mathcal{R}ic(X,Y)=3(f^{VPR}X^{R}_{i}(x)\tr\Delta^{-1}(x,x^{'})f^{VPU}Y^{U}_{i}(x^{'})).
\end{eqnarray}
where $\Delta^{-1}:L^{2}(\mathbb{R}^{n})\to H^{2}(\mathbb{R}^{n})$ is the inverse of the Laplacian $\Delta:=\eta^{ij}\nabla_{i}\nabla_{j}$ and $\tr$ denotes the formal trace operation defined by multiplication of $\delta^{PQ}$ and the distribution $\delta(x-x^{'})$ to yield the coincident limit at $A=0$.
\end{lemma} 
\begin{proof} First recall the definition of the Covariant derivative $\mathfrak{D}$ 
\begin{eqnarray}
2\mathfrak{G}(Z,\mathfrak{D}_{X}Y)=X\cdot \mathfrak{G}(Z,Y)+Y\cdot\mathfrak{G}(Z,X)-Z\cdot\mathfrak{G}(X,Y),
\end{eqnarray}
and that of the Riemann curvature 
\begin{eqnarray}
\mathcal{R}(X,Y)Z:=\mathfrak{D}_{X}\mathfrak{D}_{Y}Z-\mathfrak{D}_{Y}\mathfrak{D}_{X}Z-\mathfrak{D}_{[X,Y]}Z,
\end{eqnarray}
for $X,Y,Z,W\in \mathfrak{H}_{\mathcal{A}}$. Using the expression, one may explicitly compute at $A=0$ 
\begin{eqnarray}
\mathcal{R}(W,Z,X,Y)=-2\langle[Y_{j},W_{j}],\Delta^{-1}[X_{i},Z_{i}]\rangle-\langle[Z_{j},W_{j}],\Delta^{-1}[X_{i},Y_{i}]\rangle\nonumber+\langle[X_{j},W_{j}],\Delta^{-1}[Z_{i}, Y_{i}]\rangle.
\end{eqnarray}
The quadratic form associated with the Ricci curvature may be computed by taking the formal \textit{trace} (infinite dimensional) of the Riemann curvature 
\begin{eqnarray}
\mathcal{R}ic(X,Y)=3\tr(\langle [X,~\cdot],\Delta^{-1}[Y,~\cdot]\rangle)
\end{eqnarray}
Expanding the bracket in terms of structure constants yields the result.

\noindent Notice the following point of view that is different from direct calculations. Remarkably, the Coulomb coordinate chart based at $A=0$ is naturally a geodesic normal chart (based at $A=0$) since $\mathfrak{G}_{\dot{A}^{P}_{i}(x)\dot{A}^{Q}_{j}(x^{'})}|_{A=0}=\delta_{ij}\delta_{PQ}\delta(x-x^{'})$ and the connections $\Gamma^{A^{P}_{i}}_{A^{Q}_{j}A^{R}_{k}}|_{A=0}=0$ since $\frac{\delta \mathfrak{G}}{\delta A}|_{A=0}=0$ due to the non-constant terms in the metric being at least quadratic in $A$. Now recall the expression of the metric derived in the previous lemma and compare it with the expression of the metric in a normal neighborhood based at the flat connection $A=0$
\begin{eqnarray}
\mathfrak{G}_{\dot{A}^{P}_{i}(x)\dot{A}^{Q}_{j}(x^{'})}=\delta_{ij}\delta_{PQ}\delta(x-x^{'})-\frac{1}{3}\mathcal{R}_{A^{P}_{i}(x)A^{R}_{k}(x_{1})A^{Q}_{j}(x_{2})A^{U}_{l}(x^{'})}A^{R}_{k}(x_{1})A^{U}_{l}(x_{2})+O(|A|^{3})
\end{eqnarray}
to yield 
\begin{eqnarray}
\mathcal{R}_{A^{P}_{i}(x)A^{R}_{k}(x_{1})A^{Q}_{j}(x^{'})A^{U}_{l}(x_{2})}A^{R}_{k}(x_{1})A^{U}_{l}(x_{2})=3f^{VPR}A^{R}_{i}(x)\Delta^{-1}(x,x^{'})f^{VQU}A^{U}_{j}(x^{'}).
\end{eqnarray}
The invariant quadratic form for the Ricci tensor is then obtained by taking \textit{formal} trace of the Riemann tensor i.e., 
\begin{eqnarray}
\label{eq:Ricci}
\mathcal{R}ic(X,Y)=3(f^{VPR}X^{R}_{i}(x)\tr\Delta^{-1}(x,x^{'})f^{VPU}Y^{U}_{i}(x^{'})).
\end{eqnarray}
This concludes the lemma.
\end{proof} 
\begin{remark}
\label{sectional}
It is not difficult to see that at an arbitrary point $\widehat{A}\in \mathcal{A}/\mathcal{G}$, the formal Ricci quadratic form is simply $\mathcal{R}ic(X,Y)=3\text{tr}(f^{VPR}X^{R}_{i}(x)\Delta^{-1}_{\widehat{A}}(x,x^{'})f^{VQU}Y^{U}_{j}(x^{'}))$, where $\Delta_{\widehat{A}}:=\eta^{ij}\widehat{\nabla}^{\widehat{A}}_{i}\widehat{\nabla}^{\widehat{A}}_{j}$ is the gauge covariant Laplacian. The sectional curvature $\mathcal{K}_{X,Y}:=\langle \mathcal{R}(X,Y)Y,X\rangle$ of a $2-$plane spanned by the orthonormal vectors $X,Y\in \mathfrak{H}_{A}$ is then 
\begin{eqnarray}
\mathcal{K}_{X,Y}=3\langle[X,Y],\Delta^{-1}_{\widehat{A}}[X,Y]\rangle.
\end{eqnarray}
This can be achieved by choosing the generalized Coulomb coordinate chart based at $\widehat{A}$ defined by $\eta^{ij}\widehat{\nabla}^{\widehat{A}}_{i}(A-\widehat{A})_{j}=0$ and obtaining an expression of the metric $\mathfrak{G}$ in this chart. $\mathfrak{H}_{A}$ is tangent at $A$ to the horizontal subspace of the bundle $\mathcal{A}\to \mathcal{A}/\widehat{\mathcal{G}}$.

\end{remark}

\section{Estimate of the spectra of the Hamiltonian operator}

\noindent In the finite-dimensional setting, a lower bound on the Ricci curvature and compactness yields a lower bound on the first eigenvalue of the Laplace-Beltrami operator due to Lichnerowicz \cite{lichnerowicz1958geometrie}. In the presence of a potential, a Bakry-Emery correction to the ordinary Ricci curvature is required to estimate a precise gap in the spectrum (there are several studies on estimating the gap of a Schrodinger operator in finite dimensions using direct analysis \cite{singer1985estimate, li1980estimates, lieb1976bounds}). In an infinite dimensional setting, a straightforward generalization does not work. Note in particular that the Riemann tensor of $\mathcal{A}/\widehat{\mathcal{G}}$ is not of trace class. Recall the definition of the trace. At the flat connection $A=0$, the trace would correspond to contraction with respect to the flat metric and therefore to setting $P=Q$ and $x=x^{'}$ in the expression (\ref{eq:Ricci}). This would correspond to the evaluation of the coincidence limit of $\Delta^{-1}(x,x^{'})$. However in $2$ dimensions $\Delta^{-1}(x,x^{'})=\frac{1}{2} \ln|x-x^{'}|$ and in $3+1$ dimensions, $\Delta^{-1}(x,x^{'})=-\frac{1}{4\pi} \frac{1}{|x-x^{'}|}$, whose coincident limits of course do not exist (or in the QFT terminology, one has occurrence of ultraviolet divergences). In order to make sense of the Ricci tensor, one needs to invoke a regularization scheme. In the regularization scheme that we adopt, we split the points by approximating Dirac's distribution and taking a suitable limit. From now on, we will write the inverse metric $(\mathfrak{G}^{-1})^{A^{P}_{i}(x)A^{Q}_{j}(y)}$ by $\mathfrak{G}^{A^{P}_{i}(x)A^{Q}_{j}(y)}$ for simplicity.

\begin{definition}
\textit{Let us endow the local coordinates $\{x^{i}\}$ of a smooth $n-$manifold with the dimension of length while the metric (co-variant) coefficients are left dimensionless. The point-splitting of Dirac's distribution associated with the usual Dirac's distribution $\delta(x,x_{0})=\frac{\delta(x-x_{0})}{\mu_{g}(x)}=\frac{\prod_{i=1}^{n}\delta(x^{i}-x^{i}_{0})}{\mu_{g}(x)}$ on a Riemannian $n$-manifold $(M,g)$, $x,x_{0}\in M$, is defined as follows 
\begin{eqnarray}
\label{eq:Dirac}
\delta_{\chi}(x,x^{0}):=\frac{\prod_{i=1}^{n}\frac{\chi}{\pi}e^{-(x^{i}-x^{i}_{0})^{2}\chi^{2}}}{\mu_{g}(x)}.
\end{eqnarray}
The usual distribution is recovered after letting $\chi\to\infty$ i.e., $\int_{x} f(x)\delta_{\chi}(x,x_{0})\to f(x_{0})$ as $\chi\to\infty$ for a rapidly decaying smooth $f$ (let us say a Schwartz function).
}
\end{definition}

\subsection{Regularization of the functional Hamiltonian}
\label{technical}
\noindent A rigorous quantum Yang-Mills theory if it exists should consist of a separable Hilbert space $\mathcal{H}$, a unitary representation of the Poincar\'e group in $\mathcal{H}$, an operator-valued gauged distribution $A$ on $\mathcal{S}(\mathbb{R}^{n})$ and a dense subspace $\mathcal{D}\subset \mathcal{H}$ such that appropriate axioms of quantum gauge theory hold. As we have mentioned in the introduction this is a monumental task even for non-gauge interacting field theories. Putting aside these issues we assume a rigorous quantum field theory exists. In other words, we dodge the hardest question and study its consequences for the mass gap. The functional Hamiltonian operator defined on the orbit space $\mathcal{A}/\widehat{\mathcal{G}}$ needs regularization since even while acting on a smooth functional, it generates infinities. The formal Schr\"odinger operator for a Yang-Mills field in $n+1$ dimensions, of the type that we shall consider, is given by 
\begin{eqnarray}
\widehat{H}=\int_{\mathbb{R}^{n}}\left(-\frac{\hbar^{2}}{2}\int_{\mathbb{R}^{n}}\mathfrak{G}^{A^{P}_{I}(x)A^{Q}_{J}(y)}\frac{\mathfrak{D}}{\mathfrak{D}A^{P}_{I}(x)}\frac{\mathfrak{D}}{\mathfrak{D}A^{Q}_{J}(y)}+\frac{1}{4}\mathcal{F}_{IJ}\cdot\mathcal{F}_{IJ}\right)d^{n}x,
\end{eqnarray}
where $\frac{\mathfrak{D}}{\mathfrak{D}A^{P}_{I}}$ is the covariant derivative on the Riemannian manifold $(\mathcal{A}/\widehat{\mathcal{G}},\mathfrak{G})$ \footnote{Notice that the potential is gauge invariant and therefore naturally descends to the quotient i.e., the orbit space}. The delta distribution in $\mathfrak{G}$ is replaced by the point-split distribution defined in (\ref{eq:Dirac}). Note that contrary to the Laplacian, the Hessian $\frac{\mathfrak{D}}{\mathfrak{D}A^{P}_{I}(x)}\frac{\mathfrak{D}}{\mathfrak{D}A^{Q}_{J}(y)}$ is well-defined on a smooth functional. The flat part of the covariant functional Laplacian is ill defined. Utilizing the point-splitting of Dirac's distribution introduced previously in (\ref{eq:Dirac}) we define the regularization of the flat Laplacian as follows 
\begin{eqnarray}
\int_{x} \frac{\delta}{\delta A^{P}_{I}(x)}\frac{\delta}{\delta A^{P}_{I}(x)}\mapsto \int_{x,y} \delta_{\chi}(x,y)\frac{\delta}{\delta A^{P}_{I}(x)}\Theta_{PQ}(x,y)\frac{\delta}{\delta A^{Q}_{I}(y)},
\end{eqnarray}
where $\Theta_{AB}(x,y)$ is a parallel propagator between $x$ and $y$ and defined as a solution of the parallel propagation equation, $\Theta_{PQ}(x,y):=(\mathcal{P}e^{-\int_{y}^{x}A_{i}dz^{i}})_{PQ}$, $\mathcal{P}$ denotes the path ordering of the exponential. This is inserted in order to preserve the gauge invariance (note $\Theta_{AB}(x,y)$ transforms under a gauge transformation $\varphi\in \mathcal{G}$ as $\Theta_{PQ}(x,y)\mapsto (\varphi(x)\Theta(x,y)\varphi^{-1}(y))_{PQ}$). The result would not depend on the choice of the path from $x$ to $y$ in the limit $\chi\to\infty$, which we are interested in after subtracting possible infinities. Naturally, this regularization descends to the orbit space $\mathcal{A}/\widehat{\mathcal{G}}$ due to its gauge invariance \footnote{In addition, the parallel propagator is chosen to be such that it is symmetric under the transformation $A\to B,~x\to y$ (see \cite{karabali1997gauge, karabali1998planar, karabali1998vacuum, krug2013yang} for detail). Since
\begin{eqnarray}
\int_{x,y} \delta_{\chi}(x,y)\left(\frac{\delta}{\delta A^{P}_{I}(x)}\Theta_{PB}(x,y)\right)\frac{\delta}{\delta A^{B}_{I}(y)}=0,
\end{eqnarray}
we may write the regularization (see \cite{karabali1997gauge, karabali1998planar, karabali1998vacuum, krug2013yang} for 2+1 dimensions and \cite{nair2004invariant, freidel2006towards,freidel2006pure} (also see the thesis \cite{Krug}) for $3+1$ dimensions) as 
\begin{eqnarray}
\label{eq:gaugeinvariant}
\int_{x} \frac{\delta}{\delta A^{A}_{I}(x)}\frac{\delta}{\delta A^{A}_{I}(x)}\mapsto \int_{x,y} \delta_{\chi}(x,y)\Theta_{AB}(x,y)\frac{\delta}{\delta A^{A}_{I}(x)}\frac{\delta}{\delta A^{B}_{I}(y)},
\end{eqnarray}
where note that we recover the usual flat functional Laplacian in the limit $\chi\to\infty$.} We will proceed with this regularization scheme. Therefore we write the regularized Hamiltonian that we shall work with as follows 
\begin{eqnarray}
\label{eq:regularized}
\widehat{H}:=-\frac{\hbar^{2}}{2}\int_{\mathbb{R}^{n}\times \mathbb{R}^{n}}(\mathfrak{G}^{-1}_{\delta_{\chi}})^{A^{P}_{I}(x)A^{Q}_{J}(y)}\frac{\mathfrak{D}}{\mathfrak{D}A^{P}_{I}(x)}\Theta^{PQ}(x,y)\frac{\mathfrak{D}}{\mathfrak{D}A^{Q}_{I}(y)}+\int_{\mathbb{R}^{n}}\frac{1}{4}\mathcal{F}_{IJ}\cdot\mathcal{F}_{IJ}d^{n}x,
\end{eqnarray}
where we have point-split the Dirac's distribution appearing in the metric $\mathfrak{G}$\footnote{Notice that the second term in the metric is simply the sectional curvature and it does not involve a coincident limit.} i.e., 
\begin{eqnarray}
(\mathfrak{G}^{-1}_{\chi})^{A^{P}_{i}(x)A^{Q}_{j}(y)}=\delta_{\chi}(x,y)\delta^{PQ}\delta_{ij}-f^{PVU}A^{U}_{i}(x)\Delta^{-1}_{A}(x,y)f^{VRQ}A^{R}_{j}(y).
\end{eqnarray}
In order to estimate the gap in the spectrum of the Hamiltonian, we must perform a Bochner-type analysis on the gauge covariant Hamiltonian acting on wave functionals. Under the assumption of the existence of a quantum Yang-Mills theory, let us write the normalizable ground state wave functional as follows
\begin{eqnarray}
 \Psi[A]=N_{\hbar}e^{-S[A]/\hbar},~N_{\hbar}\in \mathbb{C}-\{0\}, ~A\in \mathcal{A}/\widehat{\mathcal{G}}.   
\end{eqnarray}
The question arises is how to obtain the ground state $\Psi[A]$. I mention two potential rigorous ways. Martin Hairer \cite{hairer1,hairer2} initiated the program of stochastic quantization where a path integral measure of the Euclidean quantum field theory can be constructed by means of studying Langevin dynamics. Once the Euclidean measure is constructed, one may analytically continue the solution to the Lorentz signature. Substantial progress is made in 2 and 3-dimensional Euclidean field theory whereas 4 dimensional case still remains open. Another approach that seems promising is the Euclidean signature semi-classical (ESSC) introduced by Moncrief \cite{moncrief2,marini,moncrief} for renormalizable interacting Bosonic field theories (borderline Sobolev embedding for $3+1$ dimensional Yang-Mills theory). This technique is in a similar spirit to the microlocal method (see \cite{martinez2002introduction} for a comprehensive review) used for the analysis of Schr\"odinger eigenvalue problems even though the latter has not previously been applicable to field theoretic problems due to technical reasons. In this approach, one substitutes the following node-less formal expression for the semi-classical expansion of the logarithm of the ground state wave functional i.e., \footnote{Contrary to the microlocal approach, if one assumes a WKB ansatz, then the tree level processes are governed by a Lorentz signature Hamilton-Jacobi equation that yields finite time blow up even in finite-dimensional problems due to the presence of caustics in the configuration space} 
\begin{eqnarray}
\label{eq:formal}
S[A]\simeq S_{0}[A]+\hbar S_{1}[A]+\frac{\hbar^{2}}{2!}S_{2}[A]+\cdot\cdot\cdot\cdot \frac{\hbar^{k}}{k!}S_{k}[A]+\cdot\cdot\cdot\cdot,\\
E^{0}\simeq \hbar\left(E_{0}+\hbar E_{1}+\hbar^{2} E_{2}+\cdot\cdot\cdot\cdot\hbar^{k}E_{k}+\cdot\cdot\cdot\cdot\right)
\end{eqnarray}
into the Schr\"odinger equation 
\begin{eqnarray}
\label{eq:covariant}
\widehat{H}\Psi[A]=E^{0}\Psi[A]
\end{eqnarray}
and impose equality order by order in the Planck constant to conclude that $S_{0}$ satisfies the following functional Hamilton-Jacobi
equation
\begin{eqnarray}
\label{eq:HJ}
\int_{\mathbb{R}^{n}\times \mathbb{R}^{n}}\frac{1}{2}\mathfrak{G}^{A^{P}_{i}(x_{1})A^{Q}_{j}(x_{2})}\frac{\delta S_{0}}{\delta A^{P}_{i}(x_{1})}\frac{\delta S_{0}}{\delta A^{Q}_{j}(x_{2})}-\int_{\mathbb{R}^{n}}\frac{1}{4}\mathcal{F}_{jk}\cdot\mathcal{F}_{jk}=0.
\end{eqnarray}
Now notice that $\frac{\delta S_{0}}{\delta A(x)}$ is well defined (no need for regularization at this tree level) and $S_{0}$ can be obtained as Hamilton's principal function for the Euclidean signature Yang-Mills action functional
 i.e., 
 \begin{eqnarray}
 S_{0}:=\inf_{\mathcal{A}\in H^{1}(\mathbb{R}^{n+1})}\mathcal{I}_{es}[\mathcal{A}],
 \end{eqnarray}
 where $\mathcal{I}_{es}[\mathcal{A}]:=\frac{1}{2}\int_{\mathbb{R}^{-}\times \mathbb{R}^{n}} \left(\sum_{\mu,\nu=0}^{n}\mathcal{F}[\mathcal{A}]^{I}_{\mu\nu}\mathcal{F}[\mathcal{A}]^{I}_{\mu\nu}\right)d^{n+1}x$. The minimization procedure may be described as follows. Given $A$ as the boundary condition for $\mathcal{A}$ on $\{0\}\times \mathbb{R}^{n}$ in the respective Sobolev trace space, one wants to minimize the Euclidean signature action functional in $\mathbb{R}^{-}\times \mathbb{R}^{n}$ with $\mathcal{A}$ approaching the flat connection on $\{-\infty\}\times \mathbb{R}^{n}$. This minimization procedure is essentially solving a semi-linear elliptic equation with a prescribed Dirichlet boundary value in a suitable choice of gauge (generalized Coulomb or Hodge gauge is one such choice). However, the non-linearity is critical for $n+1=4$ dimensions in the sense that the Sobolev embedding $H^{1}(\mathbb{R}^{4})\hookrightarrow L^{4}(\mathbb{R}^{4})$ is continuous but just fails to be compact and therefore a straightforward application of variational techniques on $\mathcal{I}_{es}[\mathcal{A}]$ having proved its convexity, coercivity, and lower semi-continuity does not work. This can be handled by means of refined elliptic estimates.
 ~Another vital problem that appears is the presence of self-dual solutions that are absolute minimizers of the Euclidean signature Yang-Mills action functional in 4 dimensions and constitute a finite-dimensional moduli space (if the action functional is same in the upper and lower half-spaces for two different self-dual solutions, then the minimization is no longer unique causing trouble). These could in turn prove to be an obstruction to the uniqueness of the minimizer $S_{0}$ leading to its not everywhere differentiability property. This, however, does not seem to cause a substantial problem at the \textit{tree} level (semi-classical) but rather causes complications when one attempts to compute the quantum loop corrections to the $S_{0}$ functional and obtain the $S_{\hbar}[A]$ functional which is what one ultimately wants. This is due to the fact that in order to compute the quantum loop corrections to the $S_{0}$ functional, one ought to solve a sequence of transport equations that are sourced by the differentiated $S_{0}$ functional that is obtained by the minimization procedure. For example, at the level of 1 loop (i.e., $O(\hbar)$), $S_{1}$ is obtained by solving the following transport equation
 \begin{eqnarray}
 \label{eq:transport}
 -\int_{\mathbb{R}^{n}\times\mathbb{R}^{n}}\mathfrak{G}^{A^{P}_{i}(x)A^{Q}_{j}(x^{'})}\frac{\delta S_{0}}{\delta A^{P}_{i}(x)}\frac{\delta S_{1}}{\delta A^{Q}_{j}(x^{'})}+\frac{1}{2}\int_{\mathbb{R}^{n}\times\mathbb{R}^{n}}\mathfrak{G}^{A^{P}_{i}(x)A^{Q}_{j}(x^{'})}\frac{\mathfrak{D}}{\mathfrak{D}A^{P}_{i}(x)}\frac{\delta S_{0}}{\delta A^{Q}_{j}(x^{'})}=E_{0}
 \end{eqnarray}
 But, since $S_{0}$ appears in a differentiated manner, the transport equation does not seem to make sense at all if $S_{0}$ is not differentiable at least almost everywhere in the orbit space. Secondly, the $S_{0}$ functional appearing as a source term for the transport equation is acted on by the functional covariant Laplacian. This problem can however be circumvented by employing the gauge-invariant point-splitting regularization procedure mentioned in (\ref{eq:gaugeinvariant}). One could proceed to compute all the tree-level processes and obtain the associated formal series (almost surely diverges). This complete task, however, can be handled in the Euclidean signature semi-classical or micro-local approach by means of the analysis of the zero energy Hamilton-Jacobi equation (\ref{eq:HJ}). In fact, as we have mentioned previously, tree-level processes should be obtainable in a rigorous way through this technique. However, it is not clear at the moment if this series solution would be able to produce the physical ground state even after renormalization and regularization. In addition, it is also unclear if a mass gap is detectable at the semiclassical level. Therefore, from now on we will not consider the split form (\ref{eq:formal}) of the functional $S[A]$, rather assume the quantum Yang-Mills theory exists and $S[A]$ makes sense all by itself. 

\begin{remark}
My analysis only works in renormalizable cases i.e., the cases where $H^{1}\hookrightarrow L^{4}$ holds (roughly the quartic term in connection in yang-mills potential is controllable by the gradient term). In higher dimensions i.e., on $\mathbb{R}^{1+n},~n\geq 4$, this embedding fails and therefore I can not  make sense of the $S[A]$ functional even formally (fails even at the level of $S_{0}$ according to the previous paragraph)     
\end{remark}

\noindent An important point worth mentioning is that I am working on the orbit space $\mathcal{A}/\widehat{G}$. In other words, I descended to the orbit space first and then defined the quantization operation. However, there is another way to proceed in the context of canonical quantization. Instead of working directly on the orbit space, one could use the temporal gauge $A_{0}=0$, impose the canonical quantization condition, solve for the functional Schrodinger's equation, and then descend to the orbit space by imposing the Gauss law constraint on the wave functional as a functional equation. Explicitly, on the co-tangent bundle $T^{*}\mathcal{A}$, one promotes the connection $A$ and its conjugate momentum $\mathcal{E}$ to operator-valued distributions in the Hilbert space $\mathcal{H}(\mathcal{A})$ of the theory and applies the equal time commutation relation (let us denote this canonical quantization operation by $Q$)
\begin{eqnarray}
 [A^{a}_{i}(x),\mathcal{E}^{b}_{i}(y)]=-\sqrt{-1}\delta^{ab}\delta_{ij}\delta(x-y)   
\end{eqnarray}
in the temporal gauge $A_{0}=0$. This operation, however, forces the Gauss-law constraint as an operator equation on the wave functional $\Psi[A]$. By virtue of satisfying the Gauss law constraint, the resulting wave functional is gauge invariant. An advantage of working up in the bundle the full space of connection $\mathcal{A}$ and then descending to the orbit space is that the Hamiltonian in this picture takes a simpler form \footnote{Notice that the functional $\Phi[A]:=e^{-\frac{1}{2}\int_{\mathbb{R}^{3}}A^{a}\cdot(\nabla\times A^{a})+\frac{g_{YM}}{3}A^{a}\cdot[A, A]^{a}}$ exactly solves $H\Phi=0$ and also verifies the Gauss Law constraint and therefore gauge invariant in $3+1$ dimensions. The problem is this functional is not normalizable.}
\begin{eqnarray}
H:=-\frac{\hbar^{2}}{2}\int_{\mathbb{R}^{n}}\frac{\delta^{2}}{\delta A^{a}_{i}(x)\delta A^{a}_{i}(x)}+\frac{1}{4}\int_{\mathbb{R}^{n}}F^{a}_{ij}F^{a}_{ij}    
\end{eqnarray}
This is the usual canonical quantization scheme for Yang-Mills theory (see \cite{hatfield2018quantum} for a detail). In my approach, one only needs to solve the Schrodinger equation (\ref{eq:covariant}), where $\widehat{H}$ is given by the following (or the regularized one in \ref{eq:regularized})
\begin{eqnarray}
\widehat{H}=\int_{\mathbb{R}^{n}}\left(-\frac{\hbar^{2}}{2}\int_{\mathbb{R}^{n}}\mathfrak{G}^{A^{P}_{I}(x)A^{Q}_{J}(y)}\frac{\mathfrak{D}}{\mathfrak{D}A^{P}_{I}(x)}\frac{\mathfrak{D}}{\mathfrak{D}A^{Q}_{J}(y)}+\frac{1}{4}\mathcal{F}_{IJ}\cdot\mathcal{F}_{IJ}\right)d^{n}x.
\end{eqnarray}
Here on the cotangent bundle $T^{*}(\mathcal{A}/\widehat{\mathcal{G}})$ one promotes the connection $A$ (essentially equivalence class of connections since I have descended down to the orbit space) and the conjugate momentum $\mathcal{E}$ to operator-valued distributions in the Hilbert space $\mathcal{H}(\mathcal{A}/\widehat{\mathcal{G}})$ and applies the following equal time commutation relation (I denote this quantization operation by $\mathcal{Q}_{*}$) 
\begin{eqnarray}
 [A^{a}_{i}(x),\mathcal{E}^{b}_{j}(y)]=-\sqrt{-1}\delta^{ab}(\delta_{ij}-(\nabla_{x})_{i}(\Delta^{-1}(x,y)(\nabla_{y})_{j})\delta(x-y).   
\end{eqnarray}
These two approaches of quantization are equivalent or the diagram \ref{comm} below commutes. 
\[\begin{tikzcd}
\label{comm}
T^{*}\mathcal{A} \arrow{r}{Q} \arrow[swap]{d}{\widehat{\mathcal{G}}} & \mathcal{H}(\mathcal{A}) \arrow{d}{\int_{\mathbb{R}^{n}}\hnabla\cdot \frac{\delta}{\delta A}(\cdot)=0} \\%
T^{*}\mathcal{A}/\widehat{\mathcal{G}} \arrow{r}{Q_{*}}& \mathcal{H}(\mathcal{A}/\widehat{\mathcal{G}})
\end{tikzcd}
\]
It suffices to verify that the wave functional $\Psi[A]$ constructed by solving $\widehat{H}\Psi[A]=E^{0}\Psi[A]$ verifies the Gauss law constraint. This follows trivially. Let $\delta A=d\alpha+[A,\alpha]$ be any smooth infinitesimal gauge transformation $\alpha$ (let's assume $\alpha$ is an element of Schwartz space i.e., decays rapidly at infinity of $\mathbb{R}^{n}$). By the definition of the orbit space, $\delta A$ should be $L^{2}$-orthogonal to any vector tangent to the orbit space and in particular
\begin{eqnarray}
\int_{x}\frac{\mathfrak{D} \Psi[A]}{\mathfrak{D} A}\delta A=0.    
\end{eqnarray}
Here we have suppressed the tensor and gauge indices for convenience. 
Now substitute $\delta A=d\alpha+[A,\alpha]$ and integrate by parts to yield 
\begin{eqnarray}
 \int_{x}\alpha\widehat{\nabla}\cdot \frac{\mathfrak{D} \Psi[A]}{\mathfrak{D} A}=0   
\end{eqnarray}
which holds for any smooth gauge transformation $\alpha$ that decays rapidly at infinity (i.e., the corresponding gauge group element $g_{\alpha}:=\exp(\sqrt{-1}\alpha)$ decays to identity). Therefore by a density argument, I have 
\begin{eqnarray}
\widehat{\nabla}\cdot \frac{\mathfrak{D} \Psi[A]}{\mathfrak{D} A}=0.  
\end{eqnarray}
Therefore $\Psi[A]$ verifies the Gauss law constraint. Notice that I constructed the metric $\mathfrak{G}$ on the orbit space $\mathcal{A}/\widehat{G}$ in lemma \ref{metriclemma1} and \ref{metriclemma2} essentially using the Gauss law to eliminate the gauge redundancy and descend to the orbit space. Therefore it is only natural that any functional on the orbit space should verify the Gauss law constraint by construction. We refer the reader to \cite{nair2012quantum} for computation of the $2+1$ dimensional Yang-Mills wave functional in approximate forms.  

\subsection{Gap estimation of the regularized Yang-Mills Hamiltonian}
\noindent Here we assume that there exists a rigorous quantization. In other words, appropriate axioms of the quantum gauge theory are satisfied. In particular, a unique ground state exists that is P\'oincare invariant and this state has zero energy. This ground state is an element of a separable \textit{Hilbert space} of the theory. Our goal is to present some geometrical arguments that suggest if there is a rigorous quantization of the Yang-Mills fields, then the associated Hamiltonian (suitably regularized) exhibits a positive mass gap.  
 Under such a bold assumption, the ground state wave functional is normalizable 
\begin{eqnarray}
\label{eq:normal}
\int_{\mathcal{A}/\widehat{\mathcal{G}}}\Psi[A]^{\dag}_{g}\Psi[A]_{g}\mu_{\mathfrak{G}}=|N_{\hbar}|^{2}\int_{\mathcal{A}/\widehat{\mathcal{G}}}e^{-2\mathcal{S}[A]/\hbar}\mu_{\mathfrak{G}}=1 
\end{eqnarray}
for $N_{\hbar}\in \mathbb{C}-\{0\}$ and with corresponding eigenvalue $E^{0}$. Note that to respect the boost-invariance $E^{0}\equiv 0$ (in fact the whole energy-momentum vector of the ground state must vanish). The formal naive measure $\mu_{\mathfrak{G}}=[DA]\sqrt{\det(\mathfrak{G})}$ does not make sense, where $[DA]:=\prod_{x}dA(x)$. However, due to (\ref{eq:normal}), we can use $|N_{\hbar}|^{2}e^{-2S[A]/\hbar}\mu_{\mathfrak{G}}$ as a measure on the orbit space $\mathcal{A}/\widehat{\mathcal{G}}$ (total measure is finite precisely due to the normalizibility of the ground state). Once again, we stress the fact that all of these hold under the assumption that we have a rigorous quantum Yang-Mills theory. The first excited state wave functional may be written as \begin{eqnarray}
\label{eq:excited}
\Psi^{*}[A]=\varphi[A]e^{-S[A]/\hbar}
\end{eqnarray}
with $\varphi:\frak{\mathcal{A}}/\widehat{\mathcal{G}}\to \mathbb{C}$ and energy $E^{*}$. Notice that the first excited state is orthogonal to the ground state and in fact not an eigenstate of the Hamiltonian due to the issue of non-renormalizability. We discuss this when we perform the gap estimation.  We are interested in estimating $E^{*}-E^{0}$. But first, we state the following integration by parts property on the metric measure space $(\mathcal{A}/\widehat{\mathcal{G}},\mathfrak{G},|N_{\hbar}|^{2}e^{-2S[A]/\hbar}\mu_{\mathfrak{G}})$ 

\begin{remark}
\label{integration}
Normalizability of the ground state yields a measure $e^{-2S[A]/\hbar}\mu_{\frak{G}}$ on $\mathcal{A}/\widehat{\mathcal{G}}$. Having constructed the complete $S$ functional, ideally one should be able to prove an estimate of the type $S[A]\geq ||A||^{k}_{H^{s}(\mathbb{R}^{n})}$ for an appropriate $k\geq 2,~s\geq \frac{1}{2}$ and therefore a rapid decay of $e^{-2S[A]/\hbar}$ at large norms of the connections on the orbit space $\mathcal{A}/\widehat{\mathcal{G}}$. In particular, with respect to this measure, one could integrate the total divergence term to yield zero i.e., 
\begin{eqnarray}
\int_{\mathcal{A}/\widehat{\mathcal{G}}}\int_{x^{1},x^{2}}\frac{\mathfrak{D}}{\mathfrak{D}A^{P}_{I}(x^{1})}(\mathfrak{G}^{A^{P}_{I}(x^{1})A^{Q}_{J}(x^{2})}\frac{\mathfrak{D}}{\mathfrak{D}A^{Q}_{J}(x^{2})}\frak{F}[A]|N_{\hbar}|^{2}e^{-2S[A]/\hbar}))\mu_{\mathfrak{G}}=0.
\end{eqnarray}
Moreover, note that since we are interested in the energy difference $E^{*}-E^{0}$, we do not need to normal order the Hamiltonian in an appropriate way.
\end{remark}

\section{Proof of the main theorem \ref{main}}
\label{proof}
\noindent The main idea behind the estimate is to obtain a Bochner-like formula on the metric measure space $(\mathcal{A}/\widehat{\mathcal{G}},\mathfrak{G},e^{-2S[A]/\hbar})$. The proof is similar to the finite-dimensional setting (see e.g., \cite{moncrief} for finite-dimensional calculations) with a vital modification being the introduction of a regulator that needs to be tracked carefully in each step.  We present the main ideas here and the detailed calculations concerning the commutation of covariant derivatives are presented in the appendix. Since we are using regularized equations, the energy states will be indexed by the regulator $\chi$. We are primarily interested in obtaining the difference between the ground state and the first excited state of the (regularized) Hamiltonian $\widehat{H}$. This is formally equivalent to finding the lowest eigenvalue (bottom of the spectra in the continuous case) of the following second-order operator
\begin{eqnarray}
\widehat{\widehat{H}}=-\frac{\hbar^{2}}{2}\int_{\mathbb{R}^{n}}\int_{\mathbb{R}^{n}}\left(\mathfrak{G}^{A^{P}_{I}(x)A^{Q}_{J}(y)}\frac{\mathfrak{D}}{\mathfrak{D}A^{P}_{I}(x)}\frac{\mathfrak{D}}{\mathfrak{D}A^{Q}_{J}(y)}-\underbrace{\frac{2}{\hbar}\mathfrak{G}^{A^{a}(x)A^{b}(y)}\frac{\mathfrak{D} S[A]}{\mathfrak{D}A^{a}(x)}\frac{\mathfrak{D}}{\mathfrak{D} A^{b}(y)}}_{\text{potential~contribution}}\right).  
\end{eqnarray}
($\widehat{\widehat{H}}$ is nothing but $\widehat{H}-E^{0}$ i.e., $(\widehat{H}-E^{0})\varphi[A]=(E^{*}-E^{0})\varphi[A]$) Here note that the potential contribution is manifested in terms of the derivative of the functional $S[A]$.
Now we regularize this operator using the same regularization scheme used for $\widehat{H}$. In addition, we also renormalize this operator. Since the case of $3+1$ dimensional Yang-Mills theory is subtle, we focus on the $2+1$ dimensional case (see section 6 for explicit calculations in $3+1$ dimensional case)
\begin{eqnarray}
\label{eq:renormalization}
\widehat{\widehat{H}}_{\chi}:=  -\frac{\hbar^{2}}{2}\int_{\mathbb{R}^{2n}}\left((\mathfrak{G}^{-1}_{\delta_{\chi}})^{A^{P}_{I}(x)A^{Q}_{J}(y)}\frac{\mathfrak{D}}{\mathfrak{D}A^{P}_{I}(x)}\Theta^{PQ}(x,y)\frac{\mathfrak{D}}{\mathfrak{D}A^{Q}_{I}(y)}-\frac{2}{\hbar}\mathfrak{G}^{A^{a}(x)A^{b}(y)}\frac{\mathfrak{D} S[A]}{\mathfrak{D}A^{a}(x)}\frac{\mathfrak{D}}{\mathfrak{D} A^{b}(y)} \right)\\\nonumber 
-\frac{3C_{2}(G)g^{2}_{YM}\ln\chi|x_{0}|}{16\pi^{3}},
\end{eqnarray}
where $\chi$ is the cut-off scale and $x_{0}$ is the subtraction scale. In section \ref{subtraction}, I discuss how to fix a subtraction scale. Note that for every finite $\chi$, $\widehat{\widehat{H}}_{\chi}$ is bounded from below allowing me to obtain $\chi$ dependent estimates. In the end, one ought to smoothly remove the cut-off scale to yield the physical result. \\  
Let us define the following entity 
 \begin{eqnarray}
 \mathcal{Q}:=\int_{\mathbb{R}^{n}\times \mathbb{R}^{n}}\mathfrak{G}_{\delta_{\chi}}^{A^{P}_{i}(x)A^{Q}_{j}(x^{'})}\Theta^{PQ}(x,x^{'})\left(\frac{\mathfrak{D}\varphi[A]^{\dag}}{\mathfrak{D}A^{P}_{i}(x)}\frac{\mathfrak{D}\varphi[A]}{\mathfrak{D}A^{Q}_{j}(x^{'})}|N_{\hbar}|^{2}e^{-2S[A]/\hbar}\right)d^{n}xd^{n}x^{'}
 \end{eqnarray}
and apply the regularized covariant functional Laplacian to yield (denote $\mathbb{R}^{n}\times \mathbb{R}^{n}$ by $\mathfrak{K}$) to yield the following identity (see the appendix for a detailed derivation of this identity)
\begin{eqnarray}
\label{eq:manipulation2}
\int_{\mathfrak{K}}\Theta^{LM}(y,z)\mathfrak{G}^{A^{L}_{I}(y)A^{M}_{J}(z)}_{\chi}\frac{\mathfrak{D}}{\mathfrak{D}A^{L}_{I}(y)}\frac{\mathfrak{D}}{\mathfrak{D}A^{M}_{J}(z)} \mathcal{Q}\\\nonumber 
=-\frac{8}{\hbar^{4}}\left\{(\widehat{H}-E^{0}_{\chi})(\varphi e^{-S/\hbar})\right\}\left\{(\widehat{H}-E^{0}_{\chi})(\varphi^{\dag} e^{-S/\hbar})\right\}\\\nonumber +\int_{\mathfrak{K}}\mathfrak{G}_{\delta_{\chi}}^{A^{L}_{k}(y)A^{M}_{l}(z)}\Theta^{LM}(y,z)\int_{\mathfrak{K}}\mathfrak{G}_{\delta_{\chi}}^{A^{P}_{i}(x)A^{Q}_{j}(x^{'})}\Theta^{PQ}(x,x^{'})\\\nonumber \left(\mathfrak{R}_{A^{M}_{l}(z)A^{N}_{n}(x^{''})A^{L}_{k}(y)A^{P}_{i}(x)}\frac{\mathfrak{D}\varphi^{\dag}}{\mathfrak{D}A^{N}_{n}(x^{''})}\frac{\mathfrak{D}\varphi}{\mathfrak{D}A^{Q}_{j}(x^{'})}e^{-2S/\hbar}\right.\\\nonumber \left.+\mathfrak{R}_{A^{M}_{l}(z)A^{N}_{n}(x^{''})A^{L}_{k}(y)A^{P}_{i}(x)}\frac{\mathfrak{D}\varphi}{\mathfrak{D}A^{N}_{n}(x^{''})}\frac{\mathfrak{D}\varphi^{\dag}}{\mathfrak{D}A^{Q}_{j}(x^{'})}e^{-2S/\hbar}\right.\\\nonumber 
\left.+\frac{4}{\hbar}\frac{\mathfrak{D}}{\mathfrak{D}A^{P}_{i}(x)}\frac{\mathfrak{D}S}{\mathfrak{D}A^{M}_{k}(z)}\frac{\mathfrak{D}\varphi^{\dag}}{\mathfrak{D}A^{Q}_{j}(x^{'}}\frac{\mathfrak{D}\varphi}{\mathfrak{D}A^{L}_{k}(z)}e^{-2S/\hbar}\right),
\end{eqnarray}
 \noindent  where note that $\widehat{H}-E^{0}_{\chi}$ can be written in terms of the renormalized operator $\widehat{\widehat{H}}_{\chi}$ in \ref{eq:renormalization}.
Now assuming the existence of the $S[A]$ functional and the rapid decay of $e^{-2S[A]/\hbar}$ at infinity, I may neglect the boundary terms while integrating over the reduced configuration space $\mathcal{A}/\widehat{\mathcal{G}}$ (remark \ref{integration}). 
In field theories, one would expect the existence of a continuous spectrum and in the current case, the spectrum would have a finite gap at the bottom. The excited states are not eigenstates of the Hamiltonian. Strictly speaking, in the case of continuous spectra, one needs to construct wave packets for excited states that are not eigenstates of the Hamiltonian since the eigenstates are not normalizable (think of a free particle in ordinary quantum mechanics). We denote the first excited state by $\varphi_{\epsilon}e^{-S[A]/\hbar}$, where the appearance of $\epsilon$ is clear from the following definitions. The first excited state is orthogonal to the ground state and it satisfies
\begin{eqnarray}
\label{eq:1stexcited}
\int_{\mathcal{A}/\widehat{\mathcal{G}}}\varphi_{\epsilon}[A]^{\dag}e^{-S[A]/\hbar}(\widehat{H}-E^{*}_{\chi})\varphi_{\epsilon}[A]e^{-S[A]/\hbar}\mu_{\mathfrak{G}}\geq 0,\\
\label{eq:manipulation1}
\int_{\mathcal{A}/\widehat{\mathcal{G}}}\left\{(\widehat{H}-E^{*}_{\chi})\varphi_{\epsilon}[A]e^{-S[A]/\hbar}\right\}^{\dag}(\widehat{H}-E^{*}_{\chi})\varphi_{\epsilon}[A]e^{-S[A]/\hbar}\mu_{\mathfrak{G}}\leq \epsilon^{2}\int_{\mathcal{A}/\widehat{\mathcal{G}}}\varphi^{\dag}_{\epsilon}\varphi_{\epsilon}[A]e^{-2S[A]/\hbar}\mu_{\mathfrak{G}},\\\nonumber 
\int_{\mathcal{A}/\widehat{\mathcal{G}}}\varphi_{\epsilon}[A]^{\dag}e^{-2S[A]/\hbar}\mu_{\mathfrak{G}}=0
\end{eqnarray}
for any $\epsilon>0$. Notice that the last condition is essential for the validity of the first condition. Otherwise, one could take $\varphi_{\epsilon}\to 1$ and yield a contradiction. Application of the Cauchy-Schwartz with respect to the measure $e^{-2S[A]/\hbar}\mu_{\mathfrak{G}}$, using the property (\ref{eq:1stexcited}), and expanding $\widehat{H}-E^{*}_{\chi}$ yields 
\begin{eqnarray}
\label{eq:derived}
0\leq \int_{\mathcal{A}/\widehat{\mathcal{G}}}\varphi_{\epsilon}[A]^{\dag}e^{-S[A]/\hbar}(\widehat{H}-E^{*}_{\chi})\varphi_{\epsilon}[A]e^{-S[A]/\hbar}\mu_{\mathfrak{G}}\leq  \epsilon\int_{\mathcal{A}/\widehat{\mathcal{G}}}\varphi_{\epsilon}[A]^{\dag}\varphi_{\epsilon}[A]e^{-2S[A]/\hbar}\mu_{\mathfrak{G}}.
\end{eqnarray}
Similarly, I may write the following
\begin{eqnarray}
\int_{\mathcal{A}/\widehat{\mathcal{G}}}\left\{(\widehat{H}-E^{0}_{\chi})\varphi_{\epsilon}[A]\nonumber e^{-S[A]/\hbar}\right\}^{\dag}(\widehat{H}-E^{0}_{\chi})\varphi_{\epsilon}[A]e^{-S[A]/\hbar}\mu_{\mathfrak{G}}\\\nonumber 
=\int_{\mathcal{A}/\widehat{\mathcal{G}}}\left\{(\widehat{H}-E^{*}_{\chi})\varphi_{\epsilon}[A]e^{-S[A]/\hbar}\right\}^{\dag}(\widehat{H}-E^{*}_{\chi})\varphi_{\epsilon}[A]e^{-S[A]/\hbar}\mu_{\mathfrak{G}}\\\nonumber 
+\int_{\mathcal{A}/\widehat{\mathcal{G}}}\left((E^{*}_{\chi}-E^{0}_{\chi})^{2}\varphi[A]^{*}_{\epsilon}\varphi[A]_{\epsilon}e^{-2S[A]/\hbar}+2(E^{*}_{\chi}-E^{0}_{\chi})\varphi[A]_{\epsilon}^{\dag}e^{-S[A]/\hbar}(\widehat{H}-E^{*}_{\chi})\varphi[A]_{\epsilon}e^{-S[A]/\hbar}\right)\mu_{\mathfrak{G}}\\\nonumber 
\leq (E^{*}_{\chi}-E^{0}_{\chi}+\epsilon)^{2}\int_{\mathcal{A}/\widehat{\mathcal{G}}}\varphi[A]^{\dag}_{\epsilon}\varphi[A]_{\epsilon}e^{-2S[A]/\hbar}\mu_{\mathfrak{G}}
\end{eqnarray}
where I have utilized (\ref{eq:derived}) and (\ref{eq:manipulation1}). Now I utilize the identity (\ref{eq:manipulation2}) but replace $\varphi$ with $\varphi_{\epsilon}$ to obtain 
\begin{eqnarray}
(E^{*}_{\chi}-E^{0}_{\chi}\nonumber+\epsilon)^{2}\int_{\mathcal{A}/\widehat{\mathcal{G}}}\varphi[A]^{\dag}_{\epsilon}\varphi[A]_{\epsilon}e^{-2S[A]/\hbar}\mu_{\mathfrak{G}}\\\nonumber \geq \int_{\mathcal{A}/\widehat{\mathcal{G}}}\left\{(\widehat{H}-E^{*}_{\chi})\varphi_{\epsilon}[A]e^{-S[A]/\hbar}\right\}^{\dag}(\widehat{H}-E^{*}_{\chi})\varphi_{\epsilon}[A]e^{-S[A]/\hbar}\mu_{\mathfrak{G}}\\\nonumber 
\geq \frac{\hbar^{4}}{4} \left[\int_{\mathcal{A}/\widehat{\mathcal{G}}}\mathfrak{R}icci(\alpha_{\epsilon}[A],\alpha_{\epsilon}[A])\mu_{\mathfrak{G}}e^{-2S/\hbar}\right.\\\nonumber 
\left.+\frac{2}{\hbar}\int_{\mathcal{A}/\widehat{\mathcal{G}}}\left\{\int_{\mathfrak{K}}\int_{\mathfrak{K}}\mathfrak{G}^{A^{P}_{I}(x)A^{Q}_{J}(x^{'})}_{\delta_{\chi}}\mathfrak{G}^{A^{M}_{K}(y)A^{N}_{L}(y^{'})}_{\delta_{\chi}}\frac{\mathfrak{D}}{\mathfrak{D}A^{P}_{I}(x)}\frac{\mathfrak{D}S}{\mathfrak{D}A^{M}_{K}(y)}\frac{\mathfrak{D}\varphi^{\dag}_{\epsilon}}{\mathfrak{D}A^{Q}_{I}(x^{'})}\frac{\mathfrak{D}\varphi_{\epsilon}}{\mathfrak{D}A^{N}_{L}(y^{'})} \right\}\mu_{\mathfrak{G}}e^{-2S/\hbar}
\right],
\end{eqnarray}
where 
\begin{eqnarray}
\mathfrak{R}icci(\alpha_{\epsilon}[A],\alpha_{\epsilon}[A]):=\int_{\mathfrak{K}}\Theta^{LM}(y,z) \nonumber\int_{\mathfrak{K}}\left(\mathfrak{G}^{A^{M}_{k}(z)A^{L}_{k}(y)}_{\delta_{\chi}}\mathfrak{R}_{A^{M}_{k}(z)A^{N}_{n}(x^{''})A^{L}_{k}(y)A^{P}_{i}(x)}\frac{\mathfrak{D}\varphi^{\dag}_{\epsilon}}{\mathfrak{D}A^{N}_{n}(x^{''})}\frac{\mathfrak{D}\varphi_{\epsilon}}{\mathfrak{D}A^{P}_{i}(x)}\right).
\end{eqnarray}
Now expanding $\widehat{H}-E^{0}$, consider the following identity
\begin{eqnarray}
\frac{2}{\hbar^{2}}\int_{\mathcal{A}/\widehat{\mathcal{G}}}\varphi_{\epsilon}[A]^{\dag}e^{-S/\hbar}(\widehat{H}-E^{0}_{\chi})\varphi_{\epsilon}[A]e^{-S/\hbar}\\\nonumber =\int_{\mathcal{A}/\widehat{\mathcal{G}}}\left(\int_{\mathbb{R}^{n}\times \mathbb{R}^{n} }\mathfrak{G}^{A^{P}_{i}(x)A^{Q}_{j}(y)}_{\delta_{\chi}}\frac{\mathfrak{D}\varphi_{\epsilon}[A]^{\dag}}{\mathfrak{D}A^{P}_{i}(x)}\frac{\mathfrak{D}\varphi_{\epsilon}[A]}{\mathfrak{D}A^{Q}_{j}(y)}\right)e^{-2S/\hbar}
\end{eqnarray}
and assume that the regulated Bakry-Emery Ricci curvature verifies the point-wise bound for a $\Delta_{\chi}>0$
\begin{eqnarray}
\mathfrak{R}icci(\alpha_{\epsilon}[A],\alpha_{\epsilon}[A])+\frac{2}{\hbar}\int_{\mathfrak{K}\times \mathfrak{K}}\mathfrak{G}^{A^{P}_{I}(x)A^{Q}_{J}(x^{'})}\mathfrak{G}^{A^{M}_{K}(y)A^{N}_{L}(y^{'})}\frac{\mathfrak{D}}{\mathfrak{D}A^{P}_{I}(x)}\frac{\mathfrak{D}S}{\mathfrak{D}A^{M}_{K}(y)}\frac{\mathfrak{D}\varphi^{\dag}_{\epsilon}}{\mathfrak{D}A^{Q}_{I}(x^{'})}\frac{\mathfrak{D}\varphi_{\epsilon}}{\mathfrak{D}A^{N}_{L}(y^{'})}\\\nonumber 
\geq \Delta_{\chi}\int_{\mathbb{R}^{n}\times \mathbb{R}^{n} }\mathfrak{G}^{A^{P}_{i}(x)A^{Q}_{j}(y)}\frac{\mathfrak{D}\varphi_{\epsilon}[A]^{\dag}}{\mathfrak{D}A^{P}_{i}(x)}\frac{\mathfrak{D}\varphi_{\epsilon}[A]}{\mathfrak{D}A^{Q}_{j}(y)}.
\end{eqnarray}
Therefore I obtain 
\begin{eqnarray}
(E^{*}_{\chi}-E^{0}_{\chi}\nonumber+\epsilon)^{2}\int_{\mathcal{A}/\widehat{\mathcal{G}}}\varphi[A]^{\dag}_{\epsilon}\varphi[A]_{\epsilon}e^{-2S[A]/\hbar}\mu_{\mathfrak{G}}\\\nonumber 
\geq \frac{\hbar^{4}}{4}\Delta_{\chi} \int_{\mathcal{A}/\widehat{\mathcal{G}}}\left(\int_{\mathbb{R}^{n}\times \mathbb{R}^{n} }\mathfrak{G}^{A^{P}_{i}(x)A^{Q}_{j}(y)}\frac{\mathfrak{D}\varphi_{\epsilon}[A]^{\dag}}{\mathfrak{D}A^{P}_{i}(x)}\frac{\mathfrak{D}\varphi_{\epsilon}[A]}{\mathfrak{D}A^{Q}_{j}(y)}\right)e^{-2S/\hbar}\\\nonumber =\Delta_{\chi}\frac{\hbar^{2}}{2}\int_{\mathcal{A}/\widehat{\mathcal{G}}}\varphi_{\epsilon}[A]^{\dag}e^{-S/\hbar}(\widehat{H}-E^{0}_{\chi})\varphi_{\epsilon}[A]e^{-S/\hbar}\\\nonumber 
=\Delta_{\chi}\frac{\hbar^{2}}{2}\int_{\mathcal{A}/\widehat{\mathcal{G}}}\varphi_{\epsilon}[A]^{\dag}e^{-S/\hbar}(\widehat{H}-E^{*}_{\chi}+E^{*}_{\chi}-E^{0}_{\chi})\varphi_{\epsilon}[A]e^{-S/\hbar}\\\nonumber \geq (E^{*}_{\chi}-E^{0}_{\chi})\frac{\hbar^{2}}{2}\Delta_{\chi} \int_{\mathcal{A}/\widehat{\mathcal{G}}}\varphi[A]^{\dag}_{\epsilon}\varphi[A]_{\epsilon}e^{-2S[A]/\hbar}\mu_{\mathfrak{G}}
\end{eqnarray}
or 
\begin{eqnarray}
(E^{*}_{\chi}-E^{0}_{\chi})+\frac{\epsilon^{2}}{E^{*}_{\chi}-E^{0}_{\chi}}+2\epsilon \geq \frac{\hbar^{2}\Delta_{\chi}}{2}~\forall \epsilon>0
\end{eqnarray}
yielding 
\begin{eqnarray}
\label{eq:inter}
E^{*}_{\chi}-E^{0}_{\chi}\geq \frac{\hbar^{2}\Delta_{\chi}}{2}.
\end{eqnarray}
Notice that $E^{*}_{\chi}-E^{0}_{\chi}$ can not be zero since that would imply 
\begin{eqnarray}
\frac{\hbar^{4}}{4}\Delta \int_{\mathcal{A}/\widehat{\mathcal{G}}}\left(\int_{\mathbb{R}^{n}\times \mathbb{R}^{n}}\mathfrak{G}^{A^{P}_{i}(x)A^{Q}_{j}(y)}\frac{\mathfrak{D}\varphi_{\epsilon}[A]^{\dag}}{\mathfrak{D}A^{P}_{i}(x)}\frac{\mathfrak{D}\varphi_{\epsilon}[A]}{\mathfrak{D}A^{Q}_{j}(y)}\right)e^{-2S/\hbar}\\\nonumber \leq \epsilon^{2} \int_{\mathcal{A}/\widehat{\mathcal{G}}}\varphi[A]^{\dag}_{\epsilon}\varphi[A]_{\epsilon}e^{-2S[A]/\hbar}\mu_{\mathfrak{G}}\forall\epsilon>0
\end{eqnarray}
yielding a contradiction to the non-constancy of $\varphi[A]$ (and hence its orthogonality to the ground state). Now the vital question is how to smoothly take the limit $\chi\to\infty$ in the inequality \ref{eq:inter}. Does a strictly positive gap survive after such an operation is performed? First notice that the Hessian contribution to the Bakry-emery Ricci curvature is well defined (if one assumes that a quantum Yang-Mills theory exists and ground state the wave functional has a finite norm; this is the assumption we are making in the current context). The problem occurs when defining the Ricci curvature of the orbit space. In $2+1$ dimensions, the term $E^{*}_{\chi}-E^{0}_{\chi}$ on the left-hand side of inequality \ref{eq:inter} would have a Logarithmic singularity arising from the Kinetic part (the electric part of the Hamiltonian is dimensionless and therefore the singular term diverges logarithmically in the cut-off scale $\chi$, see \cite{karabali1998planar} for an exact calculation regarding the appearance of this logarithmic term in the kinetic energy). Since the regularization of the term $\Delta_{\chi}$ is induced by the same regularization applied to the functional Laplacian on the orbit space through the Bochner-Lichnerowicz type analysis, the singular term in $\frac{\hbar^{2}}{2}\Delta_{\chi}$ diverges logarithmically in the cut-off scale $\chi$ in a similar fashion as that of $E^{*}_{\chi}-E^{0}_{\chi}$. This is precisely seen from the renormalization of the operator $\widehat{\widehat{H}}=\widehat{H}-E^{0}$ in \ref{eq:renormalization}. In other words, one would have through the renormalization \ref{eq:renormalization}
\begin{eqnarray}
E^{*}_{\chi}-E^{0}_{\chi}=[\Delta E]_{\text{independent of}~\chi}+\frac{3C_{2}(G)g^{2}_{YM}\ln\chi|x_{0}|}{16\pi^{3}}    
\end{eqnarray}
and we shall see in section \ref{explicit} that for $2+1$ dimensional case considered here 
\begin{eqnarray}
 \frac{\hbar^{2} \Delta_{\chi}}{2}=\frac{\hbar^{2}[\Delta]_{\text{independent of}~\chi}}{2}+\frac{3C_{2}(G)g^{2}_{YM}\ln\chi|x_{0}|}{16\pi^{3}}.       
\end{eqnarray}
Therefore, one could make sense of the inequality \ref{eq:inter} as $\chi\to \infty$ as 
\begin{eqnarray}
 [\Delta E]_{\text{independent of}~\chi}\geq \frac{\hbar^{2}[\Delta]_{\text{independent of}~\chi}}{2}.   
\end{eqnarray}
This is what I meant by the appearance of the uniform logarithmic singularity at both sides of \ref{eq:inter}. For convenience, we can drop the subscripts to write 
\begin{eqnarray}
 \Delta E \geq \frac{\hbar^{2}\Delta}{2}.   
\end{eqnarray}
This concludes the proof of the theorem for the $2+1$ dimensional case. The case of $3+1$ dimensional Yang-Mills theory is complicated because the coupling constant $g^{2}_{YM}$ runs with the energy scale. I discuss this in section \ref{3+1}. At present, I am not able to give a solid basis for a proof of the 3+1 dimensional case but rather a heuristic argument in section 6. I would like to investigate this $3+1$ dimensional case in the future.  


\begin{remark}
Note that the term $\int_{\mathfrak{K}}\Theta^{LM}(y,z) \nonumber\int_{\mathfrak{K}}\left(\mathfrak{G}_{\delta_{\chi}}^{A^{M}_{k}(z)A^{L}_{k}(y)}\mathfrak{R}_{A^{M}_{k}(z)A^{N}_{n}(x^{''})A^{L}_{k}(y)A^{P}_{i}(x)}\frac{\mathfrak{D}\varphi^{\dag}_{\epsilon}}{\mathfrak{D}A^{N}_{n}(x^{''})}\frac{\mathfrak{D}\varphi_{\epsilon}}{\mathfrak{D}A^{P}_{i}(x)}\right)\\\nonumber 
+\frac{2}{\hbar}\mathfrak{G}^{A^{P}_{I}(x)A^{Q}_{J}(x^{'})}\mathfrak{G}^{A^{M}_{K}(y)A^{N}_{L}(y^{'})}\frac{\mathfrak{D}}{\mathfrak{D}A^{P}_{I}(x)}\frac{\mathfrak{D}S}{\mathfrak{D}A^{M}_{K}(y)}\frac{\mathfrak{D}\varphi^{\dag}_{\epsilon}}{\mathfrak{D}A^{Q}_{I}(x^{'})}\frac{\mathfrak{D}\varphi_{\epsilon}}{\mathfrak{D}A^{N}_{L}(y^{'})}$ is nothing but the regularized \textit{Bakry-Emery} Ricci curvature of the configuration space $\mathcal{A}/\widehat{\mathcal{G}}$.
\end{remark}

\begin{remark}
Notice that Ricci curvature always requires regularization indicating its certain ``quantum" nature. In the perturbation theory calculations, one can show that Ricci curvature shows up in the expression of the loop amplitudes. In fact, a natural conjecture would be that the re-normalization group flow of the metric on the moduli space is a type of infinite dimensional Ricci flow.
\end{remark}

\noindent Notice an important point: The regularized Ricci curvature together with the Hessian of the functional $S$ constitute the so-called functional \textit{Bakry-Emery} Ricci tensor. In the finite-dimensional setting, it appears as the ordinary \textit{Bakry-Emery} Ricci tensor. I note studies of this Ricci tensor that naturally appear in the study of weighted manifolds by \cite{lott2003some, baldauf2022spinors}. Our setting could formally be an infinite dimensional version of a weighted manifold of the type $(\mathcal{A}/\widehat{\mathcal{G}},\mathfrak{G},e^{-2S/\hbar})$. This micro-local or Euclidean signature semi-classical technique can be used to study the quantum mechanical systems satisfying suitable conditions (see \cite{moncrief2012modified} for the study of the nonlinear anharmonic oscillators).

\section{Mass of Elementary bosonic particles through the spectrum of the Bakry-Emery Ricci curvature of the weighted true configuration space: explicit example}
\label{bosonic}
\noindent To motivate the use of Bakry-Emery Ricci curvature of the \textit{true} configuration space of the current case of Yang-Mills theory, let us first review some elementary examples. Recall the Free mass-less and massive scalar field theory on the $3+1$ dimensional Minkowski space for which the exact ground state is available. The classical action reads $I[\xi]=-\frac{1}{2}\int_{\mathbb{R}^{1+n}}\eta^{\mu\nu}(\partial_{\mu}\xi\partial_{\nu}\xi+m^{2}\xi^{2}),~~\xi:\mathbb{R}^{1+3}\to \mathbb{R}~$ which may be explicitly written as 
\begin{eqnarray}
I[\xi]=\int_{\mathbb{R}}\left(\frac{1}{2}\int_{\mathbb{R}^{3}\times \mathbb{R}^{3}}\delta(x-y)\partial_{t}\xi(x)\partial_{t}\xi(y)-\frac{1}{2}\int_{\mathbb{R}^{3}}(\eta^{ij}\partial_{i}\xi\partial_{j}\xi+m^{2}\xi^{2})\right),
\end{eqnarray}
$m$ denoting the mass.
If I denote the configuration space by $\mathfrak{M}_{\xi}$, then the kinetic term induces a flat Riemannian metric (in local coordinates $\xi$)
\begin{eqnarray}
\mathcal{M}_{\xi(x)\xi(y)}=\delta(x-y)
\end{eqnarray}
on $\mathfrak{M}_{\xi}$. The classical energy $E(k)$ has the following expression in terms of the mass  and $3-$momentum $k$
\begin{eqnarray}
E(k)=\sqrt{k^{2}+m^{2}}
\end{eqnarray}
i.e, $E(k)\geq m$. In the quantum version, the mass appears as a parameter of the irreducible representation of the Poincare group $SO(1,3)\ltimes \mathbb{R}^{1+3}$ the isometry group of the Minkowski space $\mathbb{R}^{1+3}$. In quantum field theory, this representation defines a one-particle Hilbert space $\mathcal{H}_{m}$ for a particular particle in the full spectrum of the particles. The full Hilbert space has the direct sum structure 
\begin{eqnarray}
\mathcal{H}=\mathbb{C}\oplus\left(\sum_{I}\oplus\mathcal{H}_{m_{I}}\right)\oplus m.p.s,
\end{eqnarray}
where $m.p.s$ denotes spaces of multi-particle states that are tensor products of one particle spaces. $\mathbb{C}$ corresponds to the ground state (vacuum) and has zero energy. Then there is a positive continuous spectrum starting from $\min_{I}(m_{I})=m$ and extending to infinity of the formal Hamiltonian (normal ordered and regularized) of the theory 
\begin{eqnarray}
\int_{\mathbb{R}^{3}}\left(-\int_{\mathbb{R}^{3}}\frac{\hbar^{2}}{2}\frac{\delta^{2}}{\delta\xi(x)\delta\xi(x)}\right)+\frac{1}{2}\int_{\mathbb{R}^{3}}\eta^{ij}\partial_{i}\xi\partial_{j}\xi+m^{2}\xi^{2}).
\end{eqnarray}
According to our calculations, the spectral gap i.e., the least mass $m$ is supposed to be obtainable from the Bakry-Emery Ricci curvature associated with the infinite-dimensional weighted Riemannian manifold $(\mathfrak{M}_{\xi},\mathcal{M},e^{-2S[\xi]/\hbar})$, where $S[\xi]$ is explicitly given as 
\begin{eqnarray}
S[\xi]=\frac{1}{2}\int_{k}\xi(k)\sqrt{k^{2}+m^{2}}\xi(-k)d^{3}k.
\end{eqnarray}
Now since the metric $\mathcal{M}_{\xi}$ is flat, the Bakry-Emery curvature consists of only the Hessian part of the $S$ functional. An explicit calculation for the Bakry-Emery quadratic form in this particular case yields 
\begin{eqnarray}
\mathfrak{R}icci_{B.E}(\frac{\mathfrak{D}\varphi^{\dag}_{\epsilon}}{\mathfrak{D}\xi},\frac{\mathfrak{D}\varphi^{\dag}_{\epsilon}}{\mathfrak{D}\xi})\\\nonumber
:=\mathfrak{R}icci(\frac{\mathfrak{D}\varphi^{\dag}_{\epsilon}}{\mathfrak{D}\xi},\frac{\mathfrak{D}\varphi^{\dag}_{\epsilon}}{\mathfrak{D}\xi})+\frac{2}{\hbar}\int_{\mathfrak{K}\times \mathfrak{K}}\mathcal{M}^{\xi(x)\xi(x^{'})}\mathcal{M}^{\xi(y)\xi(y^{'})}\frac{\mathfrak{D}}{\mathfrak{D}\xi(x)}\frac{\mathfrak{D}S[\xi]}{\mathfrak{D}\xi(y)}\frac{\mathfrak{D}\varphi^{\dag}_{\epsilon}}{\mathfrak{D}\xi(x^{'})}\frac{\mathfrak{D}\varphi_{\epsilon}}{\mathfrak{D}\xi(y^{'})}\\\nonumber 
=0+\frac{2}{\hbar}\int_{\mathfrak{K}\times \mathfrak{K}}\mathcal{M}^{\xi(x)\xi(x^{'})}\mathcal{M}^{\xi(y)\xi(y^{'})}\frac{\mathfrak{D}}{\mathfrak{D}\xi(x)}\frac{\mathfrak{D}S[\xi]}{\mathfrak{D}\xi(y)}\frac{\mathfrak{D}\varphi^{\dag}_{\epsilon}}{\mathfrak{D}\xi(x^{'})}\frac{\mathfrak{D}\varphi_{\epsilon}}{\mathfrak{D}\xi(y^{'})}\\\nonumber 
\geq \frac{2m}{\hbar}\int_{\mathbb{R}^{n}\times \mathbb{R}^{n}}\mathcal{M}^{\xi(x)\xi(y)}\frac{\mathfrak{D}\varphi_{\epsilon}[\xi]^{\dag}}{\mathfrak{D}\xi(x)}\frac{\mathfrak{D}\varphi_{\epsilon}[\xi]}{\mathfrak{D}\xi(y)},
\end{eqnarray}
or the energy gap $E^{*}-E^{0}\geq m\hbar$.
Notice that there is also a potential contribution in terms of the $3-$ momentum $k$ indicating a continuous spectrum starting from $m$ (i.e., the potential factor does not add a positive contribution). Therefore, the lowest (positive if exists) eigenvalue of the Bakry-Emery curvature of the weighted configuration space $(\mathfrak{M}_{\xi},\mathcal{M},e^{-2S[\xi]/\hbar})$ yields the mass gap or the lowest mass of the elementary particles. Since the configuration space is flat with respect to the induced metric (by the kinetic term), the mass gap is $m$ which is exactly what is expected. For a mass-less field, one would of course obtain a continuous spectrum starting from $0$.

\noindent Now consider the Maxwell theory on $3+1$ dimensional Minkowski space. This is of course a special case of the Yang-Mills case when the structure constants vanish. Therefore the configuration space metric is flat 
\begin{eqnarray}
\label{eq:maxwellorbit}
\mathcal{S}_{A_{I}(x)A_{J}(y)}=\delta(x-y)\delta_{IJ}
\end{eqnarray}
and the vacuum wave functional is exactly calculable i.e., the $S[A]$ functional reads 
\begin{eqnarray}
\label{eq:maxwell}
S[A]=\frac{1}{2}\int_{k}\frac{1}{|k|}(\overrightarrow{k}\times \overrightarrow{A}(\overrightarrow{k}))(\overrightarrow{k}\times \overrightarrow{A}(-\overrightarrow{k})).
\end{eqnarray}
An explicit calculation yields 
\begin{eqnarray}
\label{eq:U1positive}
\mathcal{R}icci_{B.E}(\frac{\mathfrak{D}\varphi^{\dag}_{\epsilon}}{\mathfrak{D}A_{I}},\frac{\mathfrak{D}\varphi^{\dag}_{\epsilon}}{\mathfrak{D}A_{J}})\geq 0.
\end{eqnarray}
i.e., the mass gap $E^{*}-E^{0}\geq 0$. Once again there is a $3-$momentum factor that only indicates a continuous spectrum starting from zero.
\noindent In other words, the Bakry-Emery correction term to the Ricci (i.e., the Hessian term) encodes the classical and potential contribution to the mass gap while \textit{pure Ricci is solely a quantum effect} since it contains divergence and needs to be regularized when non-zero \footnote{The quantum loop divergences are essentially arises due to the tracing of non-trace class operators. As it happens, on the current context, Riemann curvature of the configuration space is one such operator which is not of trace class. }. As it happens, in the non-abelian pure Yang-Mills theory, the regularized Ricci curvature admits a positive lower bound yielding a \textit{quantum} mass gap while the potential contribution (classical mass contribution is zero since the Yang-Mills action does not include a mass term classically and one such term can not be introduced due to gauge invariance) is expected to contribute by a non-negative continuous factor.
 
 \section{Explicit calculations for the gap in $2$ and $3$ dimensions for Yang-Mills theory}
 \label{explicit}
 \noindent Since the Ricci curvature appears in a regularized way, I can explicitly compute it and later take the limit $\chi\to\infty$. However, doing so would inevitably introduce infinities (an ultraviolet divergence; since $\chi$ has a dimension of inverse length). The regular value of the Ricci curvature is then obtained by subtracting these infinities as discussed in the proof section of the main theorem \ref{main}. Recall at the level of perturbation theory, one would remove the infinities starting at the level of the action by adding counter terms such that those counter terms generate the exact infinities (at the loop level) with opposite signs and therefore a cancellation occurs (in the process one obtains scaling differential equations for the coupling constants). Renormalizibility of the Yang-Mills theory (\cite{veltman1972regularization}) suggests that one only requires a finite number of counter terms to cancel out the infinities. I expect that a similar procedure of adding counter terms is to be carried out in the current context. However at this point the ideas of renormalization in the geometric settings is premature and therefore I do not discuss this. Ideas from lattice gauge theory \cite{balaban1988convergent} may become useful to this end.    
\begin{lemma}
The formal Ricci curvature satisfies the following expression at the flat connection $\widehat{A}=0$ in terms of the cut-off parameter $\chi$ in $2$ and $3$ spatial dimensions 
\begin{eqnarray}
\mathfrak{R}ic_{\chi}(\alpha,\alpha)&=&-\frac{3}{4\pi}\delta^{BP}f^{ABC}f^{APQ}\int_{x,x^{'}}\left(\gamma+\ln\chi |x_{0}|\right)\alpha^{C}(x)\alpha^{Q}(x^{'})d^{2}xd^{2}x^{'}\nonumber
\\
\label{eq:twodim}
&=&\frac{3C_{2}(G)g^{2}_{YM}}{16\pi^{3}}\int_{x,x^{'}}\left(\frac{\gamma}{2}+\ln\chi |x_{0}|\right)\alpha^{P}(x)\alpha^{P}(x^{'})d^{2}xd^{2}x^{'}, n=2,\\
\label{eq:riccidimension}
\mathfrak{R}ic_{\chi}(\alpha,\alpha)&=&\frac{3\chi C_{2}(G)g^{2}_{YM}}{2\pi^{3}}\int_{x,x^{'}}\alpha^{P}(x)\alpha^{P}(x^{'})d^{3}xd^{3}x^{'},~~n=3,
\end{eqnarray}
where $C_{2}(G)$ is the Casimir invariant for the adjoint representation of the compact gauge group $G=SU(N)$. Here $|x_{0}|$ is a reference constant with dimension of length.
\end{lemma}
\begin{proof}
At the flat connection $\widehat{A}=0$, the operator $\Delta^{-1}_{\widehat{A}}$ reduces to the ordinary inverse Laplacian $\Delta^{-1}$ on $\mathbb{R}^{n}$. Now recalling $\Delta^{-1}(x,x^{'})=\frac{1}{2}\ln|\frac{|x-x^{'}|}{|x_{0}|}|$ for $n=2$, $\Delta^{-1}(x,x^{'})=-\frac{1}{4\pi}\frac{1}{|x-x^{'}|}$ for $n=3$, I write the coincident limit by means of the point-splitting delta distribution $\delta_{\chi}$ as appears in the mass gap integral of lemma (3.1) to yield 
\begin{eqnarray}
\mathfrak{R}ic_{\chi}(\alpha,\alpha) =\frac{3\chi^{2}}{8\pi^{3}}\delta^{BP}f^{ABC}f^{APQ}\int_{x,x^{'}}\alpha^{C}(x)\alpha^{Q}(x^{'})\\\nonumber \left(\int_{0}^{\infty}r\ln(r/|x_{0}|)e^{-\chi^{2}r^{2}}dr\right) d^{2}x,~n=2\\
\mathfrak{R}ic_{\chi}(\alpha,\alpha)=-\frac{12\chi^{3}}{\pi^{2}}\delta^{BP}f^{ABC}f^{APQ}\int_{x,x^{'}}\alpha^{C}(x)\alpha^{Q}(x^{'})\\\nonumber \left(\int_{0}^{\infty}re^{-\chi^{2}r^{2}}dr\right)d^{3}x,~n=3.
\end{eqnarray}
Explicit integration and recalling $f^{ABC}f^{ABQ}=-C_{2}(G)g^{2}_{YM}\delta^{CQ}$ (note that our definition of the commutator is $[X,Y]^{A}=f^{ABC}X^{B}Y^{C}$ i.e., a factor $i=\sqrt{-1}$ is absorbed in the structure constants and that $g_{YM}$ is the Yang-Mills coupling constant), I obtain the result. 
\end{proof}

\noindent In the previous section, I observed that at the flat connection $A=0$, the regularized Ricci tensor enjoys a positive definite property. In this particular case the elliptic operator that appears is simply the inverse Laplacian which made the explicit calculations possible. However, away from the flat connection $A=0$, one ought to regularize the trace of the inverse gauge-covariant Laplacian $\Delta^{-1}_{\widehat{A}}$. While the spectrum is still positive, it is difficult to perform explicit calculations. Nevertheless, I may still prove that the trace of the regularized operator has a strictly positive lower bound. Recall the identity 
\begin{eqnarray}
\lambda^{-s}=\frac{1}{\Gamma[s]}\int_{0}^{\infty}t^{s-1}e^{-t\lambda}dt,
\end{eqnarray}
where $\Gamma[s]$ is the gamma function $\int_{0}^{\infty}t^{s-1}e^{-t}dt$ that has discrete poles for negative $s$. The above formula is valid for any $\lambda\in \mathbb{C}$ with $Re(\lambda)>0$. Now recall the definition of the heat kernel associated with the positive elliptic operator $\Delta_{\widehat{A}}$
\begin{eqnarray}
e^{t\Delta_{\widehat{A}}}:=\int_{\text{Spec}(\Delta_{\widehat{A}})}e^{-t\lambda}dE_{\lambda},
\end{eqnarray}
where $E_{\lambda}$ is the spectral resolution of $\Delta_{\widehat{A}}|_{H^{2}}$ in $L^{2}$. I have the following proposition for a $L^{2}$ section of the bundle $\mathfrak{P}_{Ad,\mathfrak{g}}$ (assuming that the kernel of the gauge-covariant derivative $\hnabla$ is trivial which is the case for irreducible connections; irreducible connections are generic).
\begin{proposition}
The heat kernel $e^{t\Delta_{\widehat{A}}}$ is smoothing on $L^{2}(\mathbb{R}^{n})$, more precisely
\begin{eqnarray}
||e^{t\Delta_{\widehat{A}}}f||_{H^{2k}(\mathbb{R}^{n})}\lesssim (1+t^{-k})||f||_{L^{2}(\mathbb{R}^{n})}~~~~\forall k\in \mathbb{Z}^{+}.
\end{eqnarray}
\end{proposition}
\begin{proof}
For a section of the bundle $\mathfrak{P}_{AD,\mathfrak{g}}$, the natural gauge invariant Sobolev norm of order $2k$ is defined by means of the positive operator $(-\Delta_{\widehat{A}})^{k}$ i.e., for a compactly supported section $h$ of the bundle $\mathfrak{P}_{Ad,\mathfrak{g}}$, 
\begin{eqnarray}
||h||^{2}_{H^{k}}:=\sum_{I=0}^{k}\int_{\mathbb{R}^{n}}\langle h,(-\Delta_{\widehat{A}})^{I}h\rangle.
\end{eqnarray}
Now 
\begin{eqnarray}
||(-\Delta_{\widehat{A}})^{I}e^{t\Delta_{\widehat{A}}}f||_{L^{2}}=\left(\int_{0}^{\infty}(\lambda^{I}e^{-t\lambda})^{2}d||E_{\lambda}f||^{2}_{L^{2}}\right)^{\frac{1}{2}}
\leq \sup_{\lambda\in (0,\infty)}(\lambda^{I}e^{-t\lambda})||f||_{L^{2}}\leq (I/t)^{I}e^{-I}||f||_{L^{2}}\nonumber
\end{eqnarray}
and therefore 
\begin{eqnarray}
||e^{t\Delta_{\widehat{A}}}f||_{H^{2k}}\leq C\left(1+\sum_{I=0}^{k}(\frac{I}{t})^{I}e^{-I}\right)||f||_{L^{2}}\lesssim (1+t^{-k})||f||_{L^{2}}.
\end{eqnarray}
\end{proof}
Using this heat kernel, I may therefore formally write the following
\begin{eqnarray}
(-\Delta_{\widehat{A}})^{-s}(x,x^{'})f(x^{'}):=\int_{Spec(\Delta_{\widehat{A}})}\lambda^{-s}dE_{\lambda}f(x)=\frac{1}{\Gamma[s]}\int_{Spec(\Delta_{\widehat{A}})}\int_{0}^{\infty}t^{s-1}e^{-t\lambda}dtdE_{\lambda}f(x)\\\nonumber 
=\frac{1}{\Gamma[s]}\int_{0}^{\infty}t^{s-1}\left(\int_{Spec(-\Delta_{\widehat{A}})}e^{-t\lambda}dE_{\lambda}f(x)\right)dt=\frac{1}{\Gamma[s]}\int_{0}^{\infty}t^{s-1} \left(e^{t\Delta_{\widehat{A}}}(x,x^{'})f(x^{'})\right)dt
\end{eqnarray}
where I have used the boundedness of the inner integral $\int_{0}^{\infty}t^{s-1}e^{-t\lambda}dt$ for $\lambda>0$ to interchange the order of the integrals. This integral can have the problem of producing infinities near $t=0$ and $t=\infty$. The later happens if the $\text{Spec}(-\Delta_{\widehat{A}})$ contains zero or negative numbers. This is the so called \textit{infrared} divergence issue while divergence at $t=0$ is essentially the \textit{ultraviolet} divergence issue.
Denoting $e^{t\Delta_{\widehat{A}}}(x,y)$ as $K^{\widehat{A}}(t;x,y)$ the previous expression may also be expressed as follows
\begin{eqnarray}
(-\Delta_{\widehat{A}})^{-s}(x,x^{'})f(x^{'})=\frac{1}{\Gamma[s]}\int_{0}^{\infty}t^{s-1} \left(K^{\widehat{A}}(x,x^{'})f(x^{'})\right)dt.
\end{eqnarray}
Now let us write down a formal power series expansion of $K^{\widehat{A}}(t;x,y)$ as $t\to0$
\begin{eqnarray}
K^{\widehat{A}}(t;x,y)=K(t;x,y)(1+ta_{1}(x,y)+t^{2}a_{2}(x,y)+\cdot\cdot\cdot\cdot),
\end{eqnarray}
 where $K(t;x,y)=\frac{e^{-|x-y|^{2}/4t}}{(4\pi t)^{n/2}}$ is the usual heat Kernel on $\mathbb{R}^{n}$. The coincident limits $\{a_{k}(x,x)\}$ are local invariants (invariant polynomials of curvature) given in terms of the curvature of the connection $\widehat{A}$. On $\mathbb{R}^{n}$ equipped with the flat metric, one may find through explicit calculations that $a_{1}(x,x)=0$, $a_{2}(x,x)=\frac{11}{96}F^{P}[\widehat{A}]_{ij}F^{P}[\widehat{A}]^{ij}$ (see \cite{vassilevich2003heat} for a detailed computation) i.e, 
 \begin{eqnarray}
 K^{\widehat{A}}(t;x,x)=K(t;x,x)\left(1+\frac{11t^{2}}{96}F^{P}[\widehat{A}]_{ij}F^{P}[\widehat{A}]^{ij}+O(t^{3})\right).
 \end{eqnarray}
Setting $\Lambda$ to be a small but fixed positive number, I write the trace integral as follows 
 \begin{eqnarray}
 I_{\epsilon}=\int_{\mathbb{R}^{n}}\alpha(x)(-\Delta)^{-s}_{\widehat{A}}(x,x)\alpha(x)d^{n}x\\\nonumber 
 =\frac{1}{\Gamma[s]}\int_{\mathbb{R}^{n}}\alpha(x)\left(\int_{\epsilon}^{\Lambda}t^{s-1}K(t;x,x)\left(1+\frac{11t^{2}}{96}F[\widehat{A}]_{ij}F[\widehat{A}]^{ij}+O(t^{3})\right)\right.\\\nonumber 
 \left.+\int_{\Lambda}^{\infty}t^{s-1}K^{\widehat{A}}(t;x,x)dt\right)\alpha(x)d^{n}x.
 \end{eqnarray}
 I recover the original integral after taking the limit $\epsilon\to 0$ in a suitable way. 
 Note that the infrared divergence is absent since the spectrum of $\Delta_{\widehat{A}}$ does not contain zero or negative modes (generic connections are considered) yielding a finite positive contribution from the integral\\ $\frac{1}{\Gamma[s]}\int_{\mathbb{R}^{n}}\alpha(x)\left(\int_{\Lambda}^{\infty}t^{s-1}K^{\widehat{A}}(t;x,x)dt\right)\alpha(x)d^{n}x$. The problem of ultraviolet divergence occurs at the flat value necessarily since $K(t;x,x)=\frac{1}{(4\pi t)^{\frac{n}{2}}}$ and therefore $\int_{\epsilon}^{\Lambda}t^{s-1-\frac{n}{2}}dt$ yields a $\log\epsilon$ divergence for $n=2$ and $\epsilon^{-\frac{1}{2}}$ divergence for $n=3$ at $s=1$ as expected from the previous lemma concerning the trace of $\Delta^{-1}$. This is natural since $t\approx[length]^{2}$,  $\Delta^{-1}$ in $2$ and $3$ dim essentially behaves like $\ln[Length/Length_{0}]$ (for some arbitrary reference length $Length_{0}$) and $1/[Length]$, respectively. This ultraviolet divergence is then regularized by means of the  previous lemma 4.1. The following lemma yields an estimate for the finite part of the Ricci curvature away from the flat connection $\widehat{A}=0$
 \begin{lemma}
 \label{finite}
 The regularized Ricci quadratic form $ \int_{\mathbb{R}^{n}}\alpha(x)(-\Delta)^{-s}_{\widehat{A}}(x,x)\alpha(x)d^{n}x$ satisfies 
 \begin{eqnarray}
 F.P\left\{\int_{\mathbb{R}^{n}}\alpha(x)(-\Delta)^{-1}_{\widehat{A}}(x,x)\alpha(x)d^{n}x\right\}> F.P\left\{\left(\int_{\mathbb{R}^{n}}\alpha(x)(-\Delta)^{-1}_{\widehat{A}}(x,x)\alpha(x)d^{n}x\right)_{\widehat{A}=0}\right\}\nonumber+O(\Lambda^{4-\frac{n}{2}}),
 \end{eqnarray}
 $n=2,3$, where $F.P$ denotes the finite part.
 \end{lemma}
\begin{proof}
Note the fact that $t^{s-1+k-\frac{n}{2}}$ is integrable at zero with $s=1$ for $k>\frac{n}{2}-1$. Therefore recalling the divergences that occur near $t=0$, I obtain 
\begin{eqnarray}
F.P\left\{\int_{\mathbb{R}^{n}}\alpha(x)(-\Delta)^{-1}_{\widehat{A}}(x,x)\alpha(x)d^{n}x\right\}\geq F.P\left\{\left(\int_{\mathbb{R}^{n}}\alpha(x)(-\Delta)^{-1}_{\widehat{A}}(x,x)\alpha(x)d^{n}x\right)_{\widehat{A}=0}\right\}\\\nonumber
+\frac{11\Lambda^{3-\frac{n}{2}}}{96(3-\frac{n}{2})}\int_{\mathbb{R}^{n}}\alpha(x)F[\widehat{A}(x)]_{ij}F[\widehat{A}(x)]^{ij}\alpha(x)d^{n}x+O(\Lambda^{4-\frac{n}{2}})\\\nonumber 
> F.P\left\{\left(\int_{\mathbb{R}^{n}}\alpha(x)(-\Delta)^{-1}_{\widehat{A}}(x,x)\alpha(x)d^{n}x\right)_{\widehat{A}=0}\right\}+O(\Lambda^{4-\frac{n}{2}}).
\end{eqnarray}
\end{proof}

\subsection{Fixing the subtraction scale $x_{0}$ for $2+1$ dimensional case}
\label{subtraction}
\noindent Notice that in the expression for the regularized Ricci curvature for $2+1$ Yang-Mills theory, one has a logarithmic divergence if the regulator $\chi$ is taken to the limit $\infty$ i.e., 
\begin{eqnarray}
 \mathfrak{R}ic_{\chi}(\alpha,\alpha)=\frac{3C_{2}(G)g^{2}_{YM}}{16\pi^{3}}\int_{x,x^{'}}\left(\frac{\gamma}{2}+\ln\chi |x_{0}|\right)\alpha^{P}(x)\alpha^{P}(x^{'})d^{2}xd^{2}x^{'}, n=2,   
\end{eqnarray}
However, due to dimensional reasons, one also encounters a length scale $x_{0}$. Then the following question arises: how to fix this length scale $x_{0}$ in the renormalization of the operator $\widehat{\widehat{H}}$ (\ref{eq:renormalization}) that will ultimately cancel the logarithmically divergent term appearing in the Ricci curvature. This is vital since $x_{0}$ can not be arbitrary. If it were to be arbitrary, then I can choose a new scale $y_{0}=e^{10}x_{0}$ leading the new value of the Ricci quadratic form to be 
\begin{eqnarray}
  \mathfrak{R}ic_{\chi}(\alpha,\alpha)=\frac{3C_{2}(G)g^{2}_{YM}}{16\pi^{3}}\int_{x,x^{'}}\left(\frac{\gamma}{2}+\ln\chi |y_{0}|-10\right)\alpha^{P}(x)\alpha^{P}(x^{'})d^{2}xd^{2}x^{'}   
\end{eqnarray}
which completely destroys the positivity of the finite part since $10>\frac{\gamma}{2}$. Therefore, we need to address how to fix the subtraction scale $x_{0}$. This is motivated by the study of the volume of the orbit space by \cite{karabali1998planar}. Remarkably in $2+1$ Yang-Mills theory, if one endows the orbit space of the theory with a metric that is induced by the Kinetic energy part of the classical action, then the volume of the orbit space turns out to be finite after appropriate regularization \cite{karabali1998planar}. Since my metric on the orbit space is also induced by the kinetic part of the action (represented in a different local chart than that of \cite{karabali1998planar}), it is natural to look at the volume element $\sqrt{\det(\mathfrak{G})}$. Naively this is infinite. In order to make sense of it we need to regularize it in an appropriate way. First, recall the following expression of the volume element associated with the metric $\mathfrak{G}$ in a local Coulomb chart around a reference connection $\widehat{A}$ (i.e., in a chart $\hnabla_{\widehat{A}}\cdot (A-\widehat{A})=0$, where $\hnabla_{\widehat{A}}$ is the gauge covariant derivative with respect to the connection $\widehat{A}$) as derived by \cite{babelon1979geometrical}         
\begin{eqnarray}
\label{eq:volume}
\sqrt{\det(\mathfrak{G})}= \frac{\Delta_{FP}}{[\det(\Delta_{A})\det(\Delta_{\widehat{A}})]^{\frac{1}{2}}},    
\end{eqnarray}
where $\Delta_{FP}$ is the Fadeev-Popov determinant explicit expressed as 
\begin{eqnarray}
 \Delta_{FP}=\det(\nabla_{\widehat{A}}~*\nabla_{A})   
\end{eqnarray}
where $\nabla_{\widehat{A}}~*\nabla_{A}$ is nothing but the mixed Laplacian with connections $A$ and $\widehat{A}$. The very first point to note here is that unlike $\Delta_{A}$ or $\Delta_{\widehat{A}}$, the mixed Laplacian $\nabla_{\widehat{A}}~*\nabla_{A}$ does not have a sign i.e., its eigenvalues could be negative. In order to tackle this issue, first we note the following proposition 
\begin{proposition}
\label{gribovagain}
 Assume $||A-\widehat{A}||\leq \epsilon$ for a sufficiently small $\epsilon>0$ and $||\cdot||$ denotes an appropriate norm (let's say a Sobolev norm $H^{s}$ for sufficiently large $s$). Then the spectra of $\nabla_{\widehat{A}}~*\nabla_{A}$ is strictly positive i.e., $\nabla_{\widehat{A}}~*\nabla_{A}$ is strongly elliptic.  
\end{proposition}
\begin{proof}
In order to prove this statement, we perform the following manipulations for a section $\kappa$ of the bundle $\mathfrak{P}_{Ad,\frak{g}}$
\begin{eqnarray}
-\nabla[\widehat{A}]_{i}\nabla[A]_{i}\kappa=-\nabla[A]_{i}\nabla[A]_{i}\kappa+[A-\widehat{A},\nabla[A]\kappa].    
\end{eqnarray}
Now notice that the first term is positive (in the spectral sense) for a reduced connection $A\in \mathcal{A}/\widehat{\mathcal{G}}$. Now for $||A-\widehat{A}||<\epsilon$ for sufficiently small $\epsilon>0$, the first term dominates the second indefinite term $[A-\widehat{A},\nabla[A]\kappa]$. This concludes the proof. This is once again nothing but the Gribov ambiguity i.e., we can at once only work in a small patch around a reference connection in the orbit space. Finally, we can glue together all such charts using partition of unity (care must be taken since we are in infinite dimensions) and a density argument to extend this result over the entire manifold $\mathcal{A}/\widehat{\mathcal{G}}$.    
\end{proof}
Following proposition (\ref{gribovagain}), we will work in a chart around $\widehat{A}$ defined by $||A-\widehat{A}||<\epsilon$ within which we have $\text{spectra}(\nabla_{\widehat{A}}~*\nabla_{A})>C>0$. Taking logarithm on both sides we write the equation (\ref{eq:volume}) formally as follows
\begin{eqnarray}
\log(\det(\mathfrak{G}))=2\log(\det(\nabla_{\widehat{A}}~*\nabla_{A}))-\log(\det(\Delta_{A}))-\log(\det(\Delta_{\widehat{A}})).  
\end{eqnarray}
Now we will employ $\zeta$ function and heat kernel technology \cite{moretti1999one} to evaluate the logarithms of the determinants of the elliptic operators. First, recall the following identity in the sense of spectral resolution for a strongly elliptic operator $P$ and its associated zeta function $\zeta(s)$ 
\begin{eqnarray}
\zeta(s):=\int_{\text{spectra}(P)}\lambda^{-s}dE_{\lambda}=\frac{1}{\Gamma(s)}\int_{0}^{\infty}t^{s-1}K(t)dt,   
\end{eqnarray}
where $K(t)$ is the trace of the heat kernel $K(x,y;t)$ associated with $P$ i.e., 
\begin{eqnarray}
 \frac{\partial K(x,y;t)}{\partial t}=-PK(x,y;t) 
\end{eqnarray}
and $K(x,y;t)\rightharpoonup \delta(x,y)$ as $t\to 0$. The trace $K(t)$ is defined as follows 
\begin{eqnarray}
   K(t):=\int_{x}K(x,x;t)d^{2}x. 
\end{eqnarray}
Formally the $\det(P)$ is obtainable through the following identity 
\begin{eqnarray}
  \log(\det(P)):=-\frac{d}{ds}\zeta(s)|_{s=0}.  
\end{eqnarray}
For large $t\gg 1$, $K(t)\sim e^{-\delta t},~\delta>0$ and therefore $\int_{\frac{1}{4}}^{\infty}t^{s-1}K(t)dt$ is convergent. Therefore we need to worry about the small $t$ domain. We can obtain asymptotics of the trace $K(t)$ for the three operators in our context in the small $t$ limit. To this end, we use the result from \cite{vassilevich2003heat} which computes the asymptotic expansion of the heat kernel for any gauge covariant Laplacian acting on sections of suitable bundles in terms of the bundle curvature and the geometry of the physical space. Since for us, the physical space is flat, the expression simplifies. Explicitly, we write for small $t$
\begin{eqnarray}
K^{\widehat{A}}(t)=\frac{1}{4\pi t}(\int_{x}d^{2}x+t\int_{x}a_{1}(x,x)dx+t^{2}\int_{x}a^{\widehat{A}}_{2}(x,x)dx+\cdot\cdot\cdot\cdot),\\
K^{A}(t)=\frac{1}{4\pi t}(\int_{x}d^{2}x+t\int_{x}a_{1}(x,x)dx+t^{2}\int_{x}a^{A}_{2}(x,x)dx+\cdot\cdot\cdot\cdot),\\
K^{A,\widehat{A}}(t)=\frac{1}{4\pi t}(\int_{x}d^{2}x+t\int_{x}a^{'}_{1}(x,x)dx+t^{2}\int_{x}a^{A,\widehat{A}}_{2}(x,x)dx+\cdot\cdot\cdot\cdot),
\end{eqnarray}
where $K^{\widehat{A}}(t), K^{A}(t),$ and $K^{A,\widehat{A}}(t)$ are the traced heat kernels associated to the operators $\Delta_{\widehat{A}},\Delta_{A},$ and $\nabla_{\widehat{A}}~*\nabla_{A}$, respectively. The first coefficient $a_{1}(x,x)$ only depends on the curvature of physical space and hence vanishes in our case.  $a^{'}_{1}$ consists of term proportional to $\tr(A(x)-\widehat{A} (x))$ in the mixed kernel $K^{A,\widehat{A}}(t)$ but this is zero if we consider $\mathfrak{g}=\mathfrak{su}(N)$. The $O(t^{2})$ terms in the brackets are gauge invariant terms quadratic in curvature and read 
\begin{eqnarray}
a^{\widehat{A}}_{2}(x,x)\sim \tr(F[\widehat{A}](x)\cdot F[\widehat{A}](x)),~a^{A}_{2}(x,x)\sim \tr(F[A](x)\cdot F[A](x)),\\
a^{A,\widehat{A}}_{2}(x,x)\sim \tr(F[\widehat{A}](x)\cdot F[\widehat{A}](x))+\tr(\hnabla[\widehat{A}](A-\widehat{A})(x)\cdot \hnabla[\widehat{A}](A-\widehat{A})(x)).
\end{eqnarray}
With these expressions, we explicitly compute the zeta functions and then analytically continue to $s=0$ after taking the derivative with respect to $s$
\begin{eqnarray}
 \zeta^{A}(s)&=&\frac{1}{\Gamma[s]}\int_{0}^{\infty}t^{s-1}K^{A}(t)dt\\\nonumber &=&\frac{1}{4\pi\Gamma[s]}\left(\int_{0}^{\frac{1}{4}}t^{s-2}\left(\int_{x}d^{2}x+Ct^{2}\int_{x}\tr(F[A](x)\cdot F[A](x))+O(t^{3})\right)\right)\\
 &&+  \frac{1}{\Gamma[s]}\int_{\frac{1}{4}}^{\infty}t^{s-1}K^{A}(t)dt\\\nonumber 
 &=&\frac{\int_{x}d^{2}x}{4\pi\Gamma[s]}\int_{0}^{\frac{1}{4}}t^{s-2}dt+\frac{(C^{'})^{s+1}}{4\pi\Gamma[s]}\int_{x}\tr(F[A](x)\cdot F[A](x))+\mathcal{O}(|F|^{3})\\\nonumber 
 &=&\frac{\int_{x}d^{2}x}{4\pi(s-1)\Gamma[s]}\frac{1}{4^{s-1}}+\frac{(C^{'})^{s+1}}{4\pi\Gamma[s]}\int_{x}\tr(F[A](x)\cdot F[A](x))+\mathcal{O}(|F|^{3})
\end{eqnarray}
where note that the first term has a pole at $s=1$. $C,C^{'}$ are numerical constants i.e., independent of $s$ and $A$. Now note 
\begin{eqnarray}
 \frac{1}{\Gamma(s)}=s+\gamma s^{2}+O(s^{3})   
\end{eqnarray}
near $s=0$ (where we would like to analytically continue). Therefore, the derivative reads 
\begin{eqnarray}
 \frac{d}{ds} \zeta^{A}(s)|_{s=0}=-\frac{\int_{x}d^{2}x}{\pi}+\frac{C^{'}}{4\pi}\int_{x}\tr(F[A](x)\cdot F[A](x))+\mathcal{O}(|F|^{3})
\end{eqnarray}
and therefore
\begin{eqnarray}
    \log(\det(\Delta_{A}))=-\frac{d}{ds} \zeta^{A}(s)|_{s=0}=\frac{\int_{x}d^{2}x}{\pi}-\frac{C^{'}}{4\pi}\int_{x}\tr(F[A](x)\cdot F[A](x))+\mathcal{O}(|F|^{3})
\end{eqnarray}

\noindent Now to make sense of the determinant $\det(\Delta_{A})$, I need a cut-off so that $\int_{x}d^{2}x$ is finite. This gives a large but finite subtraction scale $x_{0}$. This divergent term $\int_{x}d^{2}x$ ultimately cancels out in the expression for $\log(\det(\mathfrak{G}))$ leaving out finite terms involving potential energy.  

\begin{corollary}
\label{kinetic}
The finite parts of the Ricci curvature verify the following bounds in $2+1$ and $3+1$ dimensions
\begin{eqnarray}
\mathfrak{R}ic_{\text{finite}}(\alpha,\alpha)>\frac{3\gamma C_{2}(G)g^{2}_{YM}}{16\pi^{3}}\int_{x,x^{'}}\alpha^{P}(x)\alpha^{P}(x^{'})d^{2}xd^{2}x^{'},n=2,\\
\label{eq:riccidimension}
\mathfrak{R}ic_{\chi}(\alpha,\alpha)>\frac{3\chi C_{2}(G)g^{2}_{YM}}{2\pi^{3}}\int_{x,x^{'}}\alpha^{P}(x)\alpha^{P}(x^{'})d^{3}xd^{3}x^{'},~~n=3,
\end{eqnarray}
where in $3+1$ dimensions $\chi$ dependence still remains if one were to yield a finite result. 
\end{corollary}
\begin{proof}
Now $2+1$ dimensional Yang-Mills theory is super renormlizable and hence the coupling constant does not run. The finite part of the Ricci curvature does not depend on the cut-off $\chi$. The eigenvalue of the Laplace-Beltrami operator is simply numerical constant times $g^{2}_{YM}$ as expected. However, for $3+1$ dimensions, I keep a finite $\chi$ for dimensional reasons. At a flat connection, the finite Ricci curvature verifies 
\begin{eqnarray}
 \mathfrak{R}ic_{\text{finite}}(\alpha,\alpha)=\frac{3C_{2}(G)g^{2}_{YM}}{16\pi^{3}}\int_{x,x^{'}}\alpha^{P}(x)\alpha^{P}(x^{'})d^{2}xd^{2}x^{'},n=2,\\
\label{eq:riccidimension}
\mathfrak{R}ic_{\chi}(\alpha,\alpha)=\frac{3\chi C_{2}(G)g^{2}_{YM}}{2\pi^{3}}\int_{x,x^{'}}\alpha^{P}(x)\alpha^{P}(x^{'})d^{3}xd^{3}x^{'},~~n=3,   
\end{eqnarray}
Note that away from the flat connection, the finite part of the  Ricci quadratic form coming from the trace of $\widehat{\Delta}^{-1}$ is modified by a strictly positive entity at the leading order according to lemma \ref{finite}. In addition, the tracing operation away from the flat connection involves multiplication by the non-trivial part of the metric $f^{PUV}A^{U}_{i}(x)\Delta^{-1}_{A}(x,y)f^{VRQ}A^{R}_{j}(y)$ that yields an additional term that is strictly positive 
\begin{eqnarray}
3\int_{x,y}f^{PU_{1}V_{1}}A^{U_{1}}_{i}(x)\Delta^{-1}_{A}(x,y)f^{V_{1}R_{1}Q}A^{R_{1}}_{j}(y)(f^{VPR}\alpha^{R}_{i}(x)\Delta^{-1}_{A}(x,y)f^{VQU}\alpha^{U}_{j}(y))>0.   
\end{eqnarray}
since this is nothing but the product of sectional curvatures at the connection $A$ which is strictly positive from remark \ref{sectional}. Therefore the strict inequality in the Ricci curvature follows. 
\end{proof}

\noindent Even though the regularized Ricci curvature term produces a strictly positive contribution, the Hessian term ($\text{Hessian}(S)$ contracted with the gradient of the excited state functional $\varphi$ in equation \ref{eq:excited}; $II$ is the main theorem \ref{main}) can potentially be problematic in the sense that it can contribute by a negative factor and cancel out the positive contribution from $I$. However, due to Lorentz invariance, it is expected that ultimately this Hessian term contributes by a strictly positive factor as well. This is motivated by the work of \cite{karabali1998vacuum} on computing the ground state wave functional. More precisely the form of $S[A]$ functional for $2+1$ dimensional Yang-Mills theory is given in \cite{karabali1998vacuum} as follows 
\begin{eqnarray}
\label{eq:S}
S[A]=\frac{1}{2g^{2}_{YM}}\int_{\mathbb{R}^{2}\times \mathbb{R}^{2}}B^{a}(x)\frac{1}{m+\sqrt{m+\Delta}} B^{a}(y) d^{2}xd^{2}y,   
\end{eqnarray}
where $\Delta:=-\eta^{ij}\partial_{i}\partial_{j}$ is the Laplacian on $\mathbb{R}^{2}$, $m$ is a strictly positive number, and $B^{a}:=B^{a}[A]$ is the chromomagnetic field. This form of the $S[A]$ functional suggests that the Hessian of $S[A]$ should produce a strictly positive number (at least in a measure-theoretic sense; notice that Hessian of a gauge invariant functional can have an arbitrarily large index due to non-trivial topology of the orbit space but that is expected to happen on a measure zero set). Proving such a statement with the $S[A]$ functional \ref{eq:S} is a monumental task and we leave it for future research. In the next section, I present a heuristic geometric argument.

\subsection{Non-negativity of the term $II$ in the main theorem \ref{main}}
\noindent The geometric part (the term $I$ in the theorem \ref{main}) is strictly positive. However, the potential contribution (term $II$ in the theorem \ref{main}) can be negative (note that on a topologically non-trivial space, Hessian of a gauge invariant entity can have an arbitrary large index at a point, see \cite{cliff} for example for index estimates of Yang-Mills potential at the critical points). This issue needs further investigation. Notice that $S[A]$ is not arbitrary since $\Psi[A]:=N_{\hbar}e^{-S[A]/\hbar}$ verifies the Schrodinger equation. Here I present a heuristic argument as to how the term $II$ could contribute by a strictly positive number under some mild assumptions. I use the finite-dimensional notation to denote the covariant derivatives on the orbit space for convenience i.e., $\frac{\mathfrak{D}S}{\mathfrak{D}A^{P}_{i}(x)}$ is simply denoted by $\nabla S$. Similarly, the infinite-dimensional Laplace-Beltrami operator on $\mathcal{A}/\widehat{\mathcal{G}}$ is denoted by $\Delta$ for simplicity in notations. First I derive the following identity for the $S[A]$ functional since $\Psi[A]:=N_{\hbar}e^{-S[A]/\hbar}$ verifies the functional Schrodinger equation with vanishing ground state energy $\widehat{H}\Psi[A]=0$
\begin{eqnarray}
\label{eq:minimax}
 \frac{1}{2}\Delta(|\nabla S|^{2}e^{-\frac{2S}{\hbar}})=\left(|\nabla^{2}S|^{2}+\frac{1}{\hbar}\nabla^{j}S\nabla_{j}(|\nabla S|^{2}-V)+\text{Ricci}(\nabla S,\nabla S)\right.\\\nonumber 
 \left.-\frac{2}{\hbar}\text{Hessian}(S)(\nabla S,\nabla S)-\frac{1}{\hbar^{2}}|\nabla S|^{2}(|\nabla S|^{2}-V)+\frac{2}{\hbar^{2}}|\nabla S|^{4}\right)e^{-\frac{2S}{\hbar}},   
\end{eqnarray}
where $V$ is twice the Yang-Mills potential i.e., $V=\frac{1}{2}\int_{\mathbb{R}^{n}}F^{P}_{ij}F^{P}_{ij}$. 
I integrate this expression over the orbit space and use the vanishing of the boundary term due to rapid fall-off of the measure $e^{-2S/\hbar}$ near the infinity of the orbit space $\mathcal{A}/\widehat{\mathcal{G}}$ to yield

\begin{eqnarray}
\label{Hessian}
\frac{2}{\hbar}\int_{\mathcal{A}/\widehat{\mathcal{G}}}\nonumber\text{Hessian}(S)(\nabla S,\nabla S)e^{-\frac{2S}{\hbar}}\\
=\int_{\mathcal{A}/\widehat{\mathcal{G}}}\left(|\nabla^{2}S|^{2}+\frac{1}{\hbar^{2}}(2|\nabla S|^{2}-V)(|\nabla S|^{2}+V)+\text{Ricci}(\nabla S,\nabla S)\right)e^{-\frac{2S}{\hbar}}
\end{eqnarray}
The Ricci quadratic form $R(\nabla S,\nabla S)$ is understood to be regularized (this regularization appears since we start with a regularized Laplacian in the identity \ref{eq:minimax}).
A natural obstruction that arises is the non-negativity of $2|\nabla S|^{2}-V$. To prove this would require the construction of the measure $e^{-2S/\hbar}$ and then one would need to investigate if $S[A]$ rises fast enough. Essentially, this is where all the difficulty lies. I expect if one were to construct the ground state rigorously, then the this desired inequality $2|\nabla S|^{2}\geq V$ should hold (at least weakly i.e., $\int_{\mathcal{A}/\widehat{\mathcal{G}}}(2|\nabla S|^{2}-V)(|\nabla S|^{2}+V)e^{-\frac{2S}{\hbar}}\geq 0$. Of course, for a complete argument, one ought to establish the non-negativity of the Hessian in any arbitrary directions not just along $\nabla S$. But if it were to be non-negative it must satisfy non-negativity along any directions including $\nabla S$ and therefore the inequality $\int_{\mathcal{A}/\widehat{\mathcal{G}}}(2|\nabla S|^{2}-V)(|\nabla S|^{2}+V)e^{-\frac{2S}{\hbar}}\geq 0$ is desired at a bare minimum. Notice an important fact. Since in Lorentz covariant field theories, the metric induced on the orbit space can not be arbitrary, we mentioned in the introduction that the functional can not be independent of the geometry. In fact, the appearance of the Ricci term in the integral expression \ref{Hessian} reflects this fact. Let us examine the known case i.e., $U(1)$ gauge theory. Clearly the ground state is exactly known (\ref{eq:maxwell}) and it verifies the equation $|\nabla S|^{2}=V$ and therefore we have $\int_{\mathcal{A}/\widehat{\mathcal{G}}}\nonumber\text{Hessian}(\nabla S,\nabla S)e^{-\frac{2S}{\hbar}}=\frac{\hbar}{2}\int_{\mathcal{A}/\widehat{\mathcal{G}}}\left(|\nabla^{2}S|^{2}+\frac{2}{\hbar^{2}}|\nabla S|^{4}\right)e^{-\frac{2S}{\hbar}}\geq 0$ since the orbit space of the $U(1)$ theory is flat (\ref{eq:maxwellorbit}).  

I present a second argument based on semi-classical expansion presented in section \ref{technical}. First note that at the flat connection $A=0$ (or its equivalence class), 
\begin{eqnarray}
\text{Hessian}(S)(\alpha,\alpha)> 0
\end{eqnarray}
for any $\alpha\in T_{A}\mathcal{A}/\widehat{\mathcal{G}}$. This is because near $A=0$, 
\begin{eqnarray}
S[A]=\frac{1}{2}\int_{k}\frac{1}{|k|}(\overrightarrow{k}\times \overrightarrow{A^{a}}(\overrightarrow{k}))(\overrightarrow{k}\times \overrightarrow{A^{a}}(-\overrightarrow{k}))+O(|A|^{3})    
\end{eqnarray}
and therefore result from section \ref{bosonic}, inequality \ref{eq:U1positive} implies $\text{Hessian}(S)$ is positive definite at the flat connection $A=0$. Now I argue that $\text{Hessian}(S)[A]$ never has vanishing eigenvalues based on a semi-classical expansion. Recall the expansion of $S[A]$ in $\hbar$ 
\begin{eqnarray}
\label{eq:semi}
S[A]\simeq S_{0}[A]+\hbar S_{1}[A]+\frac{\hbar^{2}}{2!}S_{2}[A]+\cdot\cdot\cdot\cdot \frac{\hbar^{k}}{k!}S_{k}[A]+\cdot\cdot\cdot\cdot    
\end{eqnarray}
In order for, $\text{Hessian}(S)$ to have a zero eigenvalue at some point $A\in \mathcal{A}/\widehat{\mathcal{G}}$ each of $\text{Hessian}(S_{0})$,\\$\text{Hessian}(S_{1}),\text{Hessian}(S_{2}),........$ must have zero eigenvalues at $A$ simultaneously. But this fails for $S_{0}[A]$. This follows from the Hamilton-Jacobi equation \ref{eq:HJ} that is verified by $S_{0}$ (at $O(\hbar^{0})$)
\begin{eqnarray}
\int_{\mathbb{R}^{n}\times \mathbb{R}^{n}}\frac{1}{2}\mathfrak{G}^{A^{P}_{i}(x_{1})A^{Q}_{j}(x_{2})}\frac{\delta S_{0}}{\delta A^{P}_{i}(x_{1})}\frac{\delta S_{0}}{\delta A^{Q}_{j}(x_{2})}-\int_{\mathbb{R}^{n}}\frac{1}{4}\mathcal{F}_{jk}\cdot\mathcal{F}_{jk}=0.
\end{eqnarray}
since Fr\'echet differentiation of this equation in an arbitrary direction $\alpha$ yields
\begin{eqnarray}
\text{Hessian}(S_{0})\left(\frac{\delta S_{0}}{\delta A},\alpha\right)=\frac{1}{4}\mathfrak{D}_{\alpha}\int_{\mathbb{R}^{n}}\mathcal{F}_{jk}\cdot\mathcal{F}_{jk}.    
\end{eqnarray}
Now in order for the right-hand side to vanish, the Euclidean action functional $\frac{1}{4}\int_{\mathbb{R}^{n}}\mathcal{F}_{jk}\cdot\mathcal{F}_{jk}$ must have a critical point in the orbit space $\mathcal{A}/\widehat{\mathcal{G}}$. In other words, one must solve for the Euclidean Yang-Mills equations on the Cauchy slice $\mathbb{R}^{n},~n=2,3$. But for $n=2$ and $3$, there is no non-trivial solution to Euclidean Yang-Mills equations with any finite energy \cite{cliff2} (notice this fails for $n=4$ due to conformal invariance and also if one includes a Higgs field). By scaling one may increase the Yang-Mills potential energy for $n=2$ and $n=3$ as large as desired (since conformal invariance does not hold in any other dimensions except $n=4$). Therefore $\text{Hessian}(S_{0})\left(\frac{\delta S_{0}}{\delta A},\alpha\right)\neq 0$ yielding
\begin{eqnarray}
 \text{Hessian}(S)\left(\frac{\delta S}{\delta A},\alpha\right)\neq 0.   
\end{eqnarray}
Now, if we assume that $S$ is a smooth functional (or at least thrice Fr\'echet differentiable) on the orbit space $\mathcal{A}/\widehat{\mathcal{G}}$ or at least almost everywhere smooth and the complement of the set on which it is non-smooth is connected, then $\text{Hessian}(S)\left(\frac{\delta S}{\delta A},\alpha\right)> 0$ at $A=0$ and $\text{Hessian}(S)\left(\frac{\delta S}{\delta A},\alpha\right)\neq 0$ everywhere else (on the connected full measure set where it is differentiable) implies $\text{Hessian}(S)\left(\frac{\delta S}{\delta A},\alpha\right)> 0$ almost everywhere on $\mathcal{A}/\widehat{\mathcal{G}}$. Now one point that needs to be addressed is whether $\mathcal{A}/\widehat{\mathcal{G}}$ is path connected or not. I argue this as follows. If I assume connections on $\mathcal{A}/\widehat{\mathcal{G}}$ to have finite energy, then every connection is described by its asymptote at $\infty$ of $\mathbb{R}^{n},n=2,3$ which is a flat connection $A=g^{-1}dg$ where $g$ is a map from the boundary sphere $\mathbb{S}^{n-1}$ at $\infty$ to the gauge group $SU(N)$ i.e.,
\begin{eqnarray}
g:\mathbb{S}^{n-1}\to SU(N).    
\end{eqnarray}
Therefore, $\mathcal{A}/\widehat{\mathcal{G}}$ is homotopy equivalence to the space $\text{Maps}(\mathbb{S}^{n-1}\to SU(N))$. Now for $n=2$ and $3$, we have $\pi_{1}(SU(N))=0, \pi_{2}(SU(N))=0$ indicating $\mathcal{A}/\widehat{\mathcal{G}}$ to be path connected ($\mathcal{A}/\widehat{\mathcal{G}}$ is a topologically complicated space). This completes the heuristic argument that indeed the Bakry-Emery Ricci curvature associated with the $n+1,~n=2,3$ dimensional quantum Yang-Mills theory admits a strictly positive lower bound $\Delta$ (for $2+1$ dimensions, $\Delta=\frac{3\gamma C_{2}(G)g^{2}_{YM}}{4\pi}$ and for $3+1$ dimensions, $\Delta=\frac{3\chi C_{2}(G)g^{2}_{YM}}{2\pi^{3}}$ from corollary \ref{kinetic}). Even though, a rigorous analysis is to be performed to make sense of the semi-classical series \ref{eq:semi}, these heuristic arguments tend to point toward a positive answer to the Yang-Mills mass gap problem.   

\subsection{Dimensional Analysis, Large $N$ limit}
\label{3+1}
\noindent Now I perform an elementary dimensional analysis and argue that for $3+1$ dimensional Yang-Mills theory, one needs to introduce a length scale $L$ in order to obtain a mass gap. I set the light speed $c$ equals to $1$. With this convention, I have $[t]=[x]=L$. The classical action $\int_{\mathbb{R}^{1,3}}\langle F,F\rangle dt d^{3}x$ has the dimension of $\hbar$. Therefore $A$ has the dimension of $\frac{\hbar^{\frac{1}{2}}}{L}$ and $g^{2}_{YM}$ has the dimension of $\frac{1}{\hbar}$. Now according to (\ref{eq:riccidimension}), the dimension of $\frac{\mathfrak{R}ic_{\chi}(\alpha,\alpha)}{\int_{\mathbb{R}^{3}\times \mathbb{R}^{3}}\alpha(x)\alpha(x^{'})}$ is $\frac{1}{\hbar L}$ since $\chi$ has dimension $\frac{1}{L}$. Therefore $\Delta$ in the main theorem \ref{main} has dimension $\frac{1}{\hbar L}$ yielding the dimension of $E^{*}-E^{0}$ to be $\frac{\hbar}{L}$ which is the correct dimension of energy. Therefore, the introduction of a finite $\chi$ (inverse length) is absolutely necessary to generate an energy scale in the quantum Yang-Mills theory in $3+1$ dimensions. Contrary to $3+1$ dimensions, in the chosen convention $c=1$, the Yang-Mills coupling constant has an appropriate dimension in $2+1$ dimensions. In other words, $g^{2}_{YM}$ has the dimension of $\frac{1}{\hbar L}$ that yields a dimension of $\frac{1}{\hbar L}$ for the Bakry-emery bound $\Delta$ from (\ref{eq:twodim}). This in turn implies that the gap $E^{*}-E^{0}$ has the correct dimension of $\frac{\hbar}{L}$. Therefore, I do not need to introduce an additional length scale for the energy gap in the $2+1$ dimensional quantum Yang-Mills theory. In fact, $2+1$ dimensional Yang-Mills theory is super-renormalizable and therefore does not run.  

In $3+1$ Yang-Mills theory, $g_{YM}$ is dimensionless (i.e., has the dimension of $\frac{1}{\hbar}$) and thus one can not create a mass out of the occurring constants (i.e., $g_{YM},\hbar, c=1$) (purely on dimensional grounds). In order to generate the dimension of mass, one must introduce an additional length scale as I have demonstrated previously. However, this naturally appears through the regularization process as one can not simply eliminate $\chi$. Therefore in the $3+1$ case, one may not take the limit $\chi\to\infty$ (or length approaching zero) but set it to $\frac{1}{L}$, $L>0$. This $L$ is then to be fixed possibly by measuring the mass of the lowest glue-ball state. Roughly the finite part of the Ricci curvature is proportional to $\frac{3 C_{2}(G)g^{2}_{YM}}{2\pi^{3}L}$. In $3+1$ dimensions, the introduction of a length scale $L$ introduces another scale the mass $m_{0}$ of the lowest glu-ball state. Essentially the ratio of the two scales $m_{0}L$ is the meaningful entity. It would be interesting to understand this issue from a perspective of renormalization group flow i.e., to introduction of a length via regularization and renormalization process. 

Lastly one of the most interesting features of the kinetic contribution from corollary \ref{kinetic} is its invariance under large $N$ or the 't Hooft limit \cite{t1993planar}. Notice that the kinetic contribution in corollary \ref{kinetic} scales as $C_{2}(G)g^{2}_{YM}$. Now let us assume that $G=SU(N)$ and therefore $C_{2}(G)g^{2}_{YM}=Ng^{2}_{YM}$ since $C_{2}(G)=N$. But the 't Hooft's large $N$ limit is nothing but increasing $N$ while keeping $g^{2}_{YM}N$ fixed. This property seems to indicate that the curvature contribution persists in the large $N$ limit. The strict $N\to\infty$ limit is essentially a free theory in the sense that all the correlation functions of single trace, gauge invariant operators factorize (maps onto a free string theory \cite{hatfield2018quantum,t1993planar}). Nevertheless, in the large $N$ limit, the theory exhibits a mass gap (in fact the strict lower bound does not depend on $N$ as long as the t'Hooft coupling is fixed as seen from the explicit expression). Therefore my calculation seems to support the belief that in the strict large $N$ limit, one a tower of massive free particles. It would perhaps be interesting to investigate from ADS-CFT perspective.

\section{Concluding Remarks}
\noindent Topology and geometry of the configuration space of the classical gauge theory have been studied previously \cite{singer1981geometry,babelon1981riemannian, narasimhan1979geometry}. At the classical level, the geometry of the configuration space is not known to play a vital role in the sense that one does not require the geometric information of the configuration space (its curvature, etc) while studying local Cauchy problems and even in understanding the long \textit{time} dynamics. Classical Yang-Mills fields are globally well-posed on both $\mathbb{R}^{1+2}$ and $\mathbb{R}^{1+3}$ and in the proof of such global well-posedness \cite{eardley1982global, eardley1982global2,ginibre1981cauchy} nowhere does the geometry of the configuration space enter crucially. It is suspected however that the geometry of the configuration space has an important role to play at the level of quantum field theory. While very little effort is paid to understanding the role of the geometry of the classical configuration space in quantized field theory in contemporary high-energy physics, it is certainly worth the attention. In a finite-dimensional setting, sharp estimates on the spectrum of the Hamiltonian operator of a quantum theory are obtainable through Lichnerowicz-type estimates on a constructed weighted manifold under a suitable convexity assumption on the potential \cite{moncrief}. In addition, several results on the estimates of the spectrum of the Schrodinger operator are available \cite{singer1985estimate, li1980estimates, lieb1976bounds}. At the level of field theory, this is much more delicate since the operators do not make sense without appropriate regularization. Even after one performs such regularization, making sense of a rigorous quantum theory requires new novel ideas that are yet to be thought of. But the non-perturbative techniques such as the semiclassical method developed by \cite{moncrief2,marini} and stochastic quantization scheme developed by \cite{hairer1,hairer2} seem promising at the moment.

The loop corrections to the tree solution (semi-classical approximation) obtained by solving the functional Hamilton-Jacobi equation require regularization. This is because the tree contribution appears as a source term acted upon by the functional Laplacian in the transport equations for the quantum corrections. At the semi-classical level, no such regularization is required since the functional Hamilton-Jacobi equation does not involve singular operators. This singular nature of the operators appearing at the loop level is essentially related to the divergences associated with the quantum field theory. Similarly, notice that the Riemann curvature of the configuration space is a purely classical object. However, to define the Ricci curvature, I had to invoke the same regularization scheme, and as such the formal non-regularized Ricci curvature contains ultraviolet divergence terms. This hints towards a conclusion that the quantum field theory is affected by the geometry at the level of Ricci curvature (Bakry-Emery Ricci where the potential contribution is considered). Appropriate invariants of the Riemann curvature should show up at the tree-level scattering amplitudes. 
It is almost certainly expected that the Ricci curvature would inevitably show up when one tries to compute the loop amplitudes indicating a \textit{quantum-}nature of the Ricci curvature of this infinite-dimensional configuration space. A natural conjecture would be that the renormalization group flow for the metric corresponds to a forced (due to the presence of the potential term) infinite-dimensional Ricci flow. This should result in a flat metric at the high energy limit indicating the asymptotic freedom. In addition, the idea of renormalization and how it can be used to obtain a length scale to define a mass in $3+1$ dimensional Yang-Mills theory is to be understood in a rigorous way. 

\noindent Another example of the orbit space geometry playing an important role in the case of scalar electrodynamics where photons remain gap-less due to vanishing Riemann curvature of the orbit space of the $U(1)$ connections while moduli of charged scalar fields are shown to have a non-vanishing curvature \cite{moncrief}. It should be interesting to consider the large $N$ limit of $SU(N)$ non-abelian gauge theories since I have noticed the constancy of  t'Hooft's coupling indicates that the curvature contribution is invariant in the large $N$ limit. It may be interesting to study this from the ADS-CFT perspective. Another interesting application would be to study the $\mathcal{N}=4$ super Yang-Mills theory and investigate whether the expected gapless spectra can be geometrically explained. 

\section{Acknowledgement}
\noindent P.M. thanks Prof. Vincent Moncrief, Prof. Shing-Tung Yau, Prof. Cliff Taubes for numerous useful discussions related to this project such as moduli spaces of Yang-Mills connections, geometric quantum field theory, metric measure spaces, etc. Special thanks to Prof. V.P Nair for pointing out several vital issues that needed to be ironed out and suggesting potential solutions. P.M. thanks Daniel Kapec for many physical insights and the referee for a very thorough review that improved the manuscript substantially. This work was supported by the Center of Mathematical Sciences and Applications (CMSA), Department of Mathematics at Harvard University.

\begin{appendix}
\section{Calculations for the proof of the main theorem}
\noindent Here I provide the calculations regarding the commutation of the covariant derivatives to obtain the identity (\ref{eq:manipulation}) (see \cite{moncrief} for the corresponding finite-dimensional calculations). First, consider the following entity
 \begin{eqnarray}
 \mathcal{Q}:=\int_{\mathbb{R}^{n}\times \mathbb{R}^{n}}\mathfrak{G}_{\delta_{\chi}}^{A^{P}_{i}(x)A^{Q}_{j}(x^{'})}\Theta^{PQ}(x,x^{'})\left(\frac{\mathfrak{D}\varphi[A]^{\dag}}{\mathfrak{D}A^{P}_{i}(x)}\frac{\mathfrak{D}\varphi[A]}{\mathfrak{D}A^{Q}_{j}(x^{'})}|N_{\hbar}|^{2}e^{-2S[A]/\hbar}\right)d^{n}xd^{n}x^{'}
 \end{eqnarray}
and apply the regularized covariant functional Laplacian to yield (denote $\mathbb{R}^{n}\times \mathbb{R}^{n}$ by $\mathfrak{K}$)
\begin{eqnarray}
\nonumber\int_{\mathfrak{K}}\mathfrak{G}_{\delta_{\chi}}^{A^{L}_{k}(y)A^{M}_{l}(z)}\Theta^{LM}(y,z)\frac{\mathfrak{D}}{\mathfrak{D}A^{L}_{k}(y)}\frac{\mathfrak{D}}{\mathfrak{D}A^{M}_{l}(z)} \mathcal{Q}\\\nonumber
=\int_{\mathfrak{K}}\mathfrak{G}_{\delta_{\chi}}^{A^{L}_{k}(y)A^{M}_{l}(z)}\Theta^{LM}(y,z)\frac{\mathfrak{D}}{\mathfrak{D}A^{L}_{k}(y)}\frac{\mathfrak{D}}{\mathfrak{D}A^{M}_{l}(z)}\int_{\mathfrak{K}}\mathfrak{G}_{\delta_{\chi}}^{A^{P}_{i}(x)A^{Q}_{j}(x^{'})}\Theta^{PQ}(x,x^{'})\\\nonumber \left(\frac{\mathfrak{D}\varphi[A]^{\dag}}{\mathfrak{D}A^{P}_{i}(x)}\frac{\mathfrak{D}\varphi[A]}{\mathfrak{D}A^{Q}_{j}(x^{'})}|N_{\hbar}|^{2}e^{-2S[A]/\hbar}\right)\\\nonumber 
=\int_{\mathfrak{K}}\mathfrak{G}_{\delta_{\chi}}^{A^{L}_{k}(y)A^{M}_{l}(z)}\Theta^{LM}(y,z)\int_{\mathfrak{K}}\mathfrak{G}_{\delta_{\chi}}^{A^{P}_{i}(x)A^{Q}_{j}(x^{'})}\Theta^{PQ}(x,x^{'})\\\nonumber \frac{\mathfrak{D}}{\mathfrak{D}A^{L}_{k}(y)}\left(\frac{\mathfrak{D}}{\mathfrak{D}A^{P}_{i}(x)}\frac{\mathfrak{D}\varphi^{\dag}}{\mathfrak{D}A^{M}_{l}(z)}\frac{\mathfrak{D\varphi}}{\mathfrak{D}A^{Q}_{j}(x^{'})}|N_{\hbar}|^{2}e^{-2S/\hbar}\right.\\\nonumber
\left. +\frac{\mathfrak{D}\varphi^{\dag}}{\mathfrak{D}A^{P}_{i}(x)}\frac{\mathfrak{D}}{\mathfrak{D}A^{Q}_{j}(x^{'})}\frac{\mathfrak{D}\varphi}{\mathfrak{D}A^{M}_{l}(z)}|N_{\hbar}|^{2}e^{-2S/\hbar}-\frac{2}{\hbar}\frac{\mathfrak{D}\varphi^{\dag}}{\mathfrak{D}A^{P}_{i}(x)}\frac{\mathfrak{D}\varphi}{\mathfrak{D}A^{Q}_{j}(x^{'})}\frac{\mathfrak{D}S}{\mathfrak{D}A^{M}_{l}(z)}|N_{\hbar}|^{2}e^{-2S/\hbar}\right)\\\nonumber 
=\int_{\mathfrak{K}}\mathfrak{G}_{\delta_{\chi}}^{A^{L}_{k}(y)A^{M}_{l}(z)}\Theta^{LM}(y,z)\int_{\mathfrak{K}}\mathfrak{G}_{\delta_{\chi}}^{A^{P}_{i}(x)A^{Q}_{j}(x^{'})}\Theta^{PQ}(x,x^{'})\left\{\left(\frac{\mathfrak{D}}{\mathfrak{D}A^{P}_{i}(x)}\frac{\mathfrak{D}}{\mathfrak{D}A^{L}_{k}(y)}\frac{\mathfrak{D}\varphi^{\dag}}{\mathfrak{D}A^{M}_{l}(z)}\right.\right.\\\nonumber 
\left. \left.+\mathfrak{R}_{A^{M}_{l}(z)A^{N}_{n}(x^{''})A^{L}_{k}(y)A^{P}_{i}(x)}\frac{\mathfrak{D}\varphi^{\dag}}{\mathfrak{D}A^{N}_{n}(x^{''})}\right)\frac{\mathfrak{D}\varphi}{\mathfrak{D}A^{Q}_{j}(x^{'})}|N_{\hbar}|^{2}e^{-2S/\hbar}\right.\\\nonumber\left.+\frac{\mathfrak{D}}{\mathfrak{D}A^{P}_{i}(x)}\frac{D\varphi^{\dag}}{\mathfrak{D}A^{M}_{l}(z)}\frac{\mathfrak{D}}{\mathfrak{D}A^{L}_{k}(y)}\frac{\mathfrak{D}\varphi}{\mathfrak{D}A^{Q}_{j}(x^{'})}|N_{\hbar}|^{2}e^{-2S/\hbar}\right.\\\nonumber 
\left. +\left(\frac{\mathfrak{D}}{\mathfrak{D}A^{P}_{i}(x)}\frac{\mathfrak{D}}{\mathfrak{D}A^{L}_{k}(y)}\frac{\mathfrak{D}\varphi}{\mathfrak{D}A^{M}_{l}(z)}+\mathfrak{R}_{A^{M}_{l}(z)A^{N}_{n}(x^{''})A^{L}_{k}(y)A^{P}_{i}(x)}\frac{\mathfrak{D}\varphi}{\mathfrak{D}A^{N}_{n}(x^{''})}\right)\frac{\mathfrak{D}\varphi^{\dag}}{\mathfrak{D}A^{Q}_{j}(x^{'})}|N_{\hbar}|^{2}e^{-2S/\hbar}\right.\\\nonumber 
\left.+\frac{\mathfrak{D}}{\mathfrak{D}A^{L}_{k}(y)}\frac{\mathfrak{D}\varphi^{\dag}}{\mathfrak{D}A^{P}_{i}(x)}\frac{\mathfrak{D}}{\mathfrak{D}A^{Q}_{j}(x^{'})}\frac{\mathfrak{D}\varphi}{\mathfrak{D}A^{M}_{l}(z)}|N_{\hbar}|^{2}e^{-2S/\hbar}\right.\\\nonumber
\left.-\frac{2}{\hbar}\frac{\mathfrak{D}}{\mathfrak{D}A^{P}_{i}(x)}\frac{\mathfrak{D}\varphi^{\dag}}{\mathfrak{D}A^{M}_{l}(z)}\frac{\mathfrak{D}\varphi}{\mathfrak{D}A^{Q}_{j}(x^{'})}\frac{\mathfrak{D}S}{\mathfrak{D}A^{L}_{k}(y)}|N_{\hbar}|^{2}e^{-2S/\hbar}\right.\\\nonumber 
\left.-\frac{2}{\hbar}\frac{\mathfrak{D}\varphi^{\dag}}{\mathfrak{D}A^{P}_{i}(x)}\frac{\mathfrak{D}}{\mathfrak{D}A^{Q}_{j}(x^{'})}\frac{\mathfrak{D}\varphi}{\mathfrak{D}A^{M}_{l}(z)}\frac{\mathfrak{D}S}{\mathfrak{D}A^{L}_{k}(y)}|N_{\hbar}|^{2}e^{-2S/\hbar}\right.\\\nonumber 
\left.-\frac{2}{\hbar}\frac{\mathfrak{D}}{\mathfrak{D}A^{L}_{k}(y)}\left(\frac{\mathfrak{D}\varphi^{\dag}}{\mathfrak{D}A^{P}_{i}(x)}\frac{\mathfrak{D}\varphi}{\mathfrak{D}A^{Q}_{j}(x^{'})}\frac{\mathfrak{D}S}{\mathfrak{D}A^{M}_{l}(z)}|N_{\hbar}|^{2}e^{-2S/\hbar}\right)\right\}. 
\end{eqnarray}
I have utilized the fact that the functional covariant derivative commutes with the parallel propagator (Wilson line). Now notice that the Riemann curvature of the space $\mathcal{A}/\widehat{\mathcal{G}}$ appears in the previous expression, which without regularization would lead to the formal Ricci curvature which would not make sense as a trace of a non-trace class operator. In order to relate this expression to the spectral gap of the regularized Hamiltonian operator $\widehat{H}$, first, recall the following identity 
\begin{eqnarray}
\label{eq:energy}
\int_{\mathfrak{K}}\mathfrak{G}_{\delta_{\chi}}^{A^{L}_{k}(y)A^{M}_{l}(z)}\Theta^{LM}(y,z)(\frac{\mathfrak{D}}{\mathfrak{D}A^{L}_{k}(y)}\frac{\mathfrak{D}\varphi}{\mathfrak{D}A^{M}_{l}(z)})e^{-S/\hbar}d^{n}yd^{n}z=-\frac{2}{\hbar^{2}}(\widehat{H}-E^{0})(\varphi e^{-S/\hbar})\\\nonumber +\frac{2}{\hbar}\int_{\mathfrak{K}}\mathfrak{G}_{\delta_{\chi}}^{A^{L}_{k}(y)A^{M}_{l}(z)}\Theta^{LM}(y,z)\frac{\mathfrak{D}S}{\mathfrak{D}A^{L}_{k}(y)}\frac{\mathfrak{D}\varphi}{\mathfrak{D}A^{M}_{l}(z)}e^{-S/\hbar}d^{n}x.
\end{eqnarray}
Now in order to obtain a lower bound for the spectrum of $\widehat{H}-E^{0}$, I need to manipulate the expression for the entity $\int_{\mathfrak{K}}\mathfrak{G}_{\delta_{\chi}}^{A^{L}_{k}(y)A^{M}_{l}(z)}\Theta^{LM}(y,z)\frac{\mathfrak{D}}{\mathfrak{D}A^{L}_{k}(y)}\frac{\mathfrak{D}}{\mathfrak{D}A^{M}_{k}(z)} \mathcal{Q}$. Under the assumption of the existence of a rigorous quantum field theory, I may take $S$ and $\varphi$ to be smooth functionals of $A$. Therefore, under this bold assumption, all the integrals supposedly yield finite values rendering an application of Fubini's theorem to interchange the integrals over $\mathfrak{K}$ whenever necessary. In addition, having assumed the existence of a quantized theory, the regularized operator $\widehat{H}-E^{0}$ is self-adjoint with respect to the measure $e^{-2S[A]/\hbar}\mu_{\mathfrak{G}}$ and as a consequence, I may discard the boundary terms that arise in the process (rapid decay of the measure $e^{-\frac{2S}{\hbar}}$). For now let us evaluate $\int_{\mathfrak{K}}\mathfrak{G}_{\delta_{\chi}}^{A^{L}_{k}(y)A^{M}_{l}(z)}\Theta^{LM}(y,z)\frac{\mathfrak{D}}{\mathfrak{D}A^{L}_{k}(y)}\frac{\mathfrak{D}}{\mathfrak{D}A^{M}_{k}(z)} \mathcal{Q}$
\begin{eqnarray}
\nonumber\int_{\mathfrak{K}}\mathfrak{G}_{\delta_{\chi}}^{A^{L}_{k}(y)A^{M}_{l}(z)}\Theta^{LM}(y,z)\frac{\mathfrak{D}}{\mathfrak{D}A^{L}_{k}(y)}\frac{\mathfrak{D}}{\mathfrak{D}A^{M}_{l}(z)} \mathcal{Q}\\\nonumber 
=\int_{\mathfrak{K}}\mathfrak{G}_{\delta_{\chi}}^{A^{L}_{k}(y)A^{M}_{l}(z)}\Theta^{LM}(y,z)\int_{\mathfrak{K}}\mathfrak{G}_{\delta_{\chi}}^{A^{P}_{i}(x)A^{Q}_{j}(x^{'})}\Theta^{PQ}(x,x^{'})\left(-2\frac{\mathfrak{D}}{\mathfrak{D}A^{L}_{k}(y)}\frac{\mathfrak{D}\varphi^{\dag}}{\mathfrak{D}A^{M}_{l}(z)}\frac{\mathfrak{D}}{\mathfrak{D}A^{P}_{i}(x)}\frac{\mathfrak{D}\varphi}{\mathfrak{D}A^{Q}_{j}(x^{'})}e^{-2S/\hbar}\right.\\\nonumber \left.+\frac{2}{\hbar}\frac{\mathfrak{D}}{\mathfrak{D}A^{L}_{k}(y)}\frac{\mathfrak{D}\varphi^{\dag}}{\mathfrak{D}A^{M}_{l}(z)}\frac{\mathfrak{D}\varphi}{\mathfrak{D}A^{Q}_{j}(x^{'})}\frac{\mathfrak{D}S}{\mathfrak{D}A^{P}_{i}(x)}e^{-2S/\hbar}+\frac{2}{\hbar}\frac{\mathfrak{D}}{\mathfrak{D}A^{L}_{k}(y)}\frac{\mathfrak{D}\varphi}{\mathfrak{D}A^{M}_{l}(z)}\frac{\mathfrak{D}\varphi^{\dag}}{\mathfrak{D}A^{Q}_{j}(x^{'})}\frac{\mathfrak{D}S}{\mathfrak{D}A^{P}_{i}(x)}e^{-2S/\hbar}\right.\\\nonumber 
\left.+\mathfrak{R}_{A^{M}_{l}(z)A^{N}_{n}(x^{''})A^{L}_{k}(y)A^{P}_{i}(x)}\frac{\mathfrak{D}\varphi^{\dag}}{\mathfrak{D}A^{N}_{n}(x^{''})}\frac{\mathfrak{D}\varphi}{\mathfrak{D}A^{Q}_{j}(x^{'})}e^{-2S/\hbar}+\mathfrak{R}_{A^{M}_{l}(Z)A^{N}_{n}(x^{''})A^{L}_{k}(y)A^{P}_{i}(x)}\frac{\mathfrak{D}\varphi}{\mathfrak{D}A^{N}_{n}(x^{''})}\frac{\mathfrak{D}\varphi^{\dag}}{\mathfrak{D}A^{Q}_{j}(x^{'})}e^{-2S/\hbar}\right.\\\nonumber 
\left. +\frac{\mathfrak{D}}{\mathfrak{D}A^{P}_{i}(x)}\frac{\mathfrak{D}\varphi^{\dag}}{\mathfrak{D}A^{M}_{l}(z)}\frac{\mathfrak{D}}{\mathfrak{D}A^{L}_{k}(y)}\frac{\mathfrak{D}\varphi}{\mathfrak{D}A^{Q}_{j}(x^{'})}e^{-2S/\hbar}+\frac{\mathfrak{D}}{\mathfrak{D}A^{L}_{k}(y)}\frac{\mathfrak{D}\varphi^{\dag}}{\mathfrak{D}A^{P}_{i}(x)}\frac{\mathfrak{D}}{\mathfrak{D}A^{Q}_{j}(x^{'})}\frac{\mathfrak{D}\varphi}{\mathfrak{D}A^{M}_{l}(z)}e^{-2S/\hbar}\right.\\\nonumber 
\left.-\frac{2}{\hbar}\frac{\mathfrak{D}}{\mathfrak{D}A^{P}_{i}(x)}\frac{\mathfrak{D}\varphi^{\dag}}{\mathfrak{D}A^{M}_{l}(z)}\frac{\mathfrak{D}\varphi}{\mathfrak{D}A^{Q}_{j}(x^{'})}\frac{\mathfrak{D}S}{\mathfrak{D}A^{L}_{k}(y)}e^{-2S/\hbar}-\frac{2}{\hbar}\frac{\mathfrak{D}\varphi^{\dag}}{\mathfrak{D}A^{P}_{i}(x)}\frac{\mathfrak{D}}{\mathfrak{D}A^{Q}_{j}(x^{'})}\frac{\mathfrak{D}\varphi}{\mathfrak{D}A^{M}_{l}(z)}\frac{\mathfrak{D}S}{\mathfrak{D}A^{L}_{k}(y)}e^{-2S/\hbar}\right.\\\nonumber 
\left. +\frac{\mathfrak{D}}{\mathfrak{D}A^{P}_{i}(x)}\left(\frac{\mathfrak{D}}{\mathfrak{D}A^{L}_{k}(y)}\frac{\mathfrak{D}\varphi^{\dag}}{\mathfrak{D}A^{M}_{l}(z)}\frac{\mathfrak{D}\varphi}{\mathfrak{D}A^{Q}_{j}(x^{'})}e^{-2S/\hbar}\right)+\frac{\mathfrak{D}}{\mathfrak{D}A^{P}_{i}(x)}\left(\frac{\mathfrak{D}}{\mathfrak{D}A^{L}_{k}(y)}\frac{\mathfrak{D}\varphi}{\mathfrak{D}A^{M}_{l}(z)}\frac{\mathfrak{D}\varphi^{\dag}}{\mathfrak{D}A^{Q}_{j}(x^{'})} e^{-2S/\hbar}\right)\right.\\\nonumber 
\left. -\frac{2}{\hbar}\frac{\mathfrak{D}}{\mathfrak{D}A^{L}_{k}(y)}\left(\frac{\mathfrak{D}\varphi^{\dag}}{\mathfrak{D}A^{P}_{i}(x)}\frac{\mathfrak{D}\varphi}{\mathfrak{D}A^{Q}_{j}(x^{'})}\frac{\mathfrak{D}S}{\mathfrak{D}A^{M}_{l}(z)}e^{-2S/\hbar}\right)\right).
\end{eqnarray}
Now utilizing the identity (\ref{eq:energy}), I may write the following 
\begin{eqnarray}
\int_{\mathfrak{K}}\mathfrak{G}_{\delta_{\chi}}^{A^{L}_{k}(y)A^{M}_{l}(z)}\Theta^{LM}(y,z)\int_{\mathfrak{K}}\mathfrak{G}_{\delta_{\chi}}^{A^{P}_{i}(x)A^{Q}_{j}(x^{'})}\Theta^{PQ}(x,x^{'})\frac{\mathfrak{D}}{\mathfrak{D}A^{L}_{k}(y)}\frac{\mathfrak{D}\varphi^{\dag}}{\mathfrak{D}A^{M}_{l}(z)}\nonumber \frac{\mathfrak{D}}{\mathfrak{D}A^{P}_{i}(x)}\frac{\mathfrak{D}\varphi}{\mathfrak{D}A^{Q}_{j}(x^{'})}e^{-2S/\hbar}\\\nonumber 
=\frac{4}{\hbar^{4}}\left\{(\widehat{H}-E^{0})(\varphi e^{-S/\hbar})\right\}\left\{(\widehat{H}-E^{0})(\varphi^{\dag} e^{-S/\hbar})\right\}+\frac{2}{\hbar}\int_{\mathfrak{K}}\mathfrak{G}_{\delta_{\chi}}^{A^{L}_{k}(y)A^{M}_{l}(z)}\Theta^{LM}(y,z)\int_{\mathfrak{K}}\mathfrak{G}_{\delta_{\chi}}^{A^{P}_{i}(x)A^{Q}_{j}(x^{'})}\\\nonumber \Theta^{PQ}(x,x^{'})\left(\frac{2}{\hbar}\frac{\mathfrak{D}}{\mathfrak{D}A^{L}_{k}(y)}\frac{\mathfrak{D}\varphi}{\mathfrak{D}A^{M}_{l}(z)}\frac{\mathfrak{D}S}{\mathfrak{D}A^{P}_{i}(x)}\frac{\mathfrak{D}\varphi^{\dag}}{\mathfrak{D}A^{Q}_{j}(x^{'})}e^{-2S/\hbar}+\frac{2}{\hbar}\frac{\mathfrak{D}}{\mathfrak{D}A^{P}_{i}(x)}\frac{\mathfrak{D}\varphi^{\dag}}{\mathfrak{D}A^{Q}_{j}(x^{'})}\frac{\mathfrak{D}S}{\mathfrak{D}A^{L}_{k}(y)}\frac{\mathfrak{D}\varphi}{\mathfrak{D}A^{M}_{l}(z)}e^{-2S/\hbar}\right.\\\nonumber 
\left.-\frac{4}{\hbar^{2}}\frac{\mathfrak{D}S}{\mathfrak{D}A^{P}_{i}(x)}\frac{\mathfrak{D}\varphi}{\mathfrak{D}A^{Q}_{j}(x^{'})}\frac{\mathfrak{D}S}{\mathfrak{D}A^{L}_{k}(y)}\frac{\mathfrak{D}\varphi^{\dag}}{\mathfrak{D}A^{M}_{l}(z)}e^{-2S/\hbar}\right)
\end{eqnarray}
substitution of which in the expression for $\int_{\mathfrak{K}}\delta_{\chi}(y,z)\Theta^{LM}(y,z)\frac{\mathfrak{D}}{\mathfrak{D}A^{L}_{k}(y)}\frac{\mathfrak{D}}{\mathfrak{D}A^{M}_{k}(z)} \mathcal{Q}$ yields 
\begin{eqnarray}
\label{eq:manipulation}
\int_{\mathfrak{K}}\Theta^{LM}(y,z)\nonumber\mathfrak{G}^{A^{L}_{I}(y)A^{M}_{J}(z)}_{\chi}\frac{\mathfrak{D}}{\mathfrak{D}A^{L}_{I}(y)}\frac{\mathfrak{D}}{\mathfrak{D}A^{M}_{J}(z)} \mathcal{Q}\\\nonumber 
=-\frac{8}{\hbar^{4}}\left\{(\widehat{H}-E^{0})(\varphi e^{-S/\hbar})\right\}\left\{(\widehat{H}-E^{0})(\varphi^{\dag} e^{-S/\hbar})\right\}\\\nonumber +\int_{\mathfrak{K}}\mathfrak{G}_{\delta_{\chi}}^{A^{L}_{k}(y)A^{M}_{l}(z)}\Theta^{LM}(y,z)\int_{\mathfrak{K}}\mathfrak{G}_{\delta_{\chi}}^{A^{P}_{i}(x)A^{Q}_{j}(x^{'})}\Theta^{PQ}(x,x^{'})\\\nonumber \left(\mathfrak{R}_{A^{M}_{l}(z)A^{N}_{n}(x^{''})A^{L}_{k}(y)A^{P}_{i}(x)}\frac{\mathfrak{D}\varphi^{\dag}}{\mathfrak{D}A^{N}_{n}(x^{''})}\frac{\mathfrak{D}\varphi}{\mathfrak{D}A^{Q}_{j}(x^{'})}e^{-2S/\hbar}\right.\\\nonumber \left.+\mathfrak{R}_{A^{M}_{l}(z)A^{N}_{n}(x^{''})A^{L}_{k}(y)A^{P}_{i}(x)}\frac{\mathfrak{D}\varphi}{\mathfrak{D}A^{N}_{n}(x^{''})}\frac{\mathfrak{D}\varphi^{\dag}}{\mathfrak{D}A^{Q}_{j}(x^{'})}e^{-2S/\hbar}\right.\\\nonumber 
\left.+\frac{4}{\hbar}\frac{\mathfrak{D}}{\mathfrak{D}A^{P}_{i}(x)}\frac{\mathfrak{D}S}{\mathfrak{D}A^{M}_{k}(z)}\frac{\mathfrak{D}\varphi^{\dag}}{\mathfrak{D}A^{Q}_{j}(x^{'}}\frac{\mathfrak{D}\varphi}{\mathfrak{D}A^{L}_{k}(z)}e^{-2S/\hbar}\right.\\\nonumber 
\left. +2\frac{\mathfrak{D}}{\mathfrak{D}A^{L}_{k}(y)}\frac{\mathfrak{D}\varphi^{\dag}}{\mathfrak{D}A^{P}_{i}(x)}\frac{\mathfrak{D}}{\mathfrak{D}A^{Q}_{i}(x^{'})}\frac{\mathfrak{D}\varphi}{\mathfrak{D}A^{M}_{l}(z)}e^{-2S/\hbar}\right.\\\nonumber 
\left. +\frac{\mathfrak{D}}{\mathfrak{D}A^{P}_{i}(x)}\left(\frac{\mathfrak{D}}{\mathfrak{D}A^{L}_{k}(y)}\frac{\mathfrak{D}\varphi^{\dag}}{\mathfrak{D}A^{M}_{l}(z)}\frac{\mathfrak{D}\varphi}{\mathfrak{D}A^{Q}_{i}(x^{'})}e^{-2S/\hbar}\right)\right.\\\nonumber 
\left.+\frac{\mathfrak{D}}{\mathfrak{D}A^{P}_{i}(x)}\left(\frac{\mathfrak{D}}{\mathfrak{D}A^{L}_{k}(y)}\frac{\mathfrak{D}\varphi}{\mathfrak{D}A^{M}_{l}(z)}\frac{\mathfrak{D}\varphi^{\dag}}{\mathfrak{D}A^{Q}_{j}(x^{'})} e^{-2S/\hbar}\right)\right.\\\nonumber
\left. -\frac{2}{\hbar}\frac{\mathfrak{D}}{\mathfrak{D}A^{L}_{k}(y)}\left(\frac{\mathfrak{D}\varphi^{\dag}}{\mathfrak{D}A^{P}_{i}(x)}\frac{\mathfrak{D}\varphi}{\mathfrak{D}A^{Q}_{j}(x^{'})}\frac{\mathfrak{D}S}{\mathfrak{D}A^{M}_{l}(z)}e^{-2S/\hbar}\right)\right.\\\nonumber 
\left.-\frac{2}{\hbar}\frac{\mathfrak{D}}{\mathfrak{D}A^{P}_{i}(x)}\left(\frac{\mathfrak{D}\varphi}{\mathfrak{D}A^{Q}_{j}(x^{'})}\frac{\mathfrak{D}\varphi^{\dag}}{\mathfrak{D}A^{L}_{k}(y)}\frac{\mathfrak{D}S}{\mathfrak{D}A^{M}_{l}(z)}e^{-2S/\hbar}\right)\right.\\\nonumber 
\left.-\frac{2}{\hbar}\frac{\mathfrak{D}}{\mathfrak{D}A^{P}_{i}(x)}\left(\frac{\mathfrak{D}\varphi^{\dag}}{\mathfrak{D}A^{Q}_{j}(x^{'})}\frac{\mathfrak{D}\varphi}{\mathfrak{D}A^{L}_{k}(y)}\frac{\mathfrak{D}S}{\mathfrak{D}A^{M}_{l}(z)}e^{-2S/\hbar}\right)\right).
\end{eqnarray}

\end{appendix}

\end{document}